\documentclass[twocolumn]{autart}  

\usepackage{mypreamble}

\usepackage[T1]{fontenc} %

\def\extendedversion{1} %
\newcommand{\extver}[2]{%
  \ifx\extendedversion\undefined%
	#1%
  \else%
    #2%
  \fi%
}

\usepackage[skip=1pt, indent=15pt]{parskip}

\newcommand{\thmspace}{\vspace{4pt plus 2pt minus 2pt}}

\newcommand{\maketextstyle}{\textstyle} %

\expandafter\def\expandafter\normalsize\expandafter{%
    \normalsize
    \setlength\abovedisplayskip{4pt}
    \setlength\belowdisplayskip{4pt}
    \setlength\abovedisplayshortskip{4pt}
    \setlength\belowdisplayshortskip{4pt}
}

\begin{document}

\begin{frontmatter}

\title{Analysis of Non-Square Nonlinear MIMO Systems using \\ Scaled Relative Graphs} 
\extver{\vspace{-5mm}}{\subtitle{\small (Extended version)}}

\author[tue]{Julius P.~J. Krebbekx}\ead{j.p.j.krebbekx@tue.nl},
\author[tue,hungary]{Roland T\'oth}\ead{r.toth@tue.nl},
\author[tue]{Amritam Das}\ead{am.das@tue.nl}

\address[tue]{Control Systems group, Department of Electrical Engineering,  Eindhoven University of Technology, The Netherlands}                                          
\address[hungary]{Systems and Control Lab, HUN-REN Institute for Computer Science and Control, Budapest, Hungary}


\begin{abstract}                   
    \emph{Scaled Relative Graphs} (SRGs) provide a novel graphical frequency-domain method for the analysis of nonlinear systems. There have been recent efforts to generalize SRG analysis to \emph{Multiple-Input Multiple-Output} (MIMO) systems. However, these attempts yielded only results for square systems, due to the inherent Hilbert space structure of the SRG. In this paper, we develop an SRG analysis method that accommodates non-square operators. The key element is the embedding of operators to a space of operators acting on a common Hilbert space, while restricting the input space to the original input dimension, to avoid conservatism. We generalize SRG interconnection rules to restricted input spaces and develop stability theorems to guarantee causality, well-posedness and (incremental) $L_2$-gain bounds for the overall interconnection. We show utilization of the proposed theoretical concepts on the analysis of nonlinear systems in a \emph{Linear Fractional Representation} (LFR) form, which is a rather general class of systems, and the LFR is directly utilizable for control. Moreover, we provide formulas for the computation of MIMO SRGs of stable LTI operators and diagonal and non-square static nonlinear operators. Finally, we demonstrate the advantages of our embedding approach on several examples.
\end{abstract}

\end{frontmatter}

\section{Introduction}\label{sec:introduction} 

Graphical methods such as the Nyquist and Bode diagrams are foundational to control engineering, providing intuitive, frequency-domain criteria for stability and performance and enabling practical design tools such as loop shaping and autotuning. Gain and phase margins derived from Nyquist analysis are also foundational to modern robust control. However, these tools are restricted to \emph{Linear Time-Invariant} (LTI) systems, and extending them to multivariable nonlinear systems remains an open challenge.

The \emph{Scaled Relative Graph} (SRG)~\cite{ryuScaledRelativeGraphs2022} was recently proposed as a novel graphical framework for analyzing nonlinear \emph{Single-Input Single-Output} (SISO) systems~\cite{chaffeyGraphicalNonlinearSystem2023}. The SRG offers a non-approximative method yielding sufficient conditions for stability and upper bounds on the (incremental) $L_2$-gain, a key performance metric in practice. It is modular, allowing interconnections to be analyzed by composing the SRGs of subsystems. The SRG recovers classical results such as the small-gain theorem~\cite{chaffeyGraphicalNonlinearSystem2023} and generalizes the circle criterion~\cite{krebbekxScaledRelativeGraph2025}, due to its close relation to the Nyquist diagram. It also enables frequency-dependent gain bounds, forming the basis of a nonlinear Bode diagram and bandwidth definition~\cite{krebbekxNonlinearBandwidthBode2025a}, and accommodates for the LTI notions of phase lead/lag~\cite{eijndenPhaseScaledGraphs2025}. Apart from theoretical developments, SRGs have proven effective in applications such as reset control analysis~\cite{grootDissipativityFrameworkConstructing2025} and design~\cite{krebbekxResetControllerAnalysis2025}, and circuit modeling~\cite{quanScaledRelativeGraphs2024}, to name a few.

Recently, SRG tools have been used to study \emph{Multiple-Input Multiple-Output} (MIMO) systems~\cite{baron-pradaDecentralizedStabilityConditions2025,chenGraphicalDominanceAnalysis2025,chenSoftHardScaled2025,grootExploitingStructureMIMO2025}. However, due to the inherent Hilbert space structure of the SRG, all these approaches focus on interconnections of \emph{square}\footnote{Systems with the same amount of inputs as outputs.} operators only. When operators are not square, the Hilbert space of inputs is not the same as the space of outputs, and naively adding inputs/outputs to make the system square leads to conservatism as in standard squaring up approaches such as~\cite{misraComputationalAlgorithmSquaringup1992}. 

The key idea of this paper is to embed all operators in an interconnection to the same space of operators on a Hilbert space and carry out the SRG interconnections in a common Hilbert space. Conservatism due to the additional input/output dimensions is removed by restricting the space of input signals to the original input dimension. We show that such restriction can result in a passive plant or stable feedback loop, while naive additional input dimensions for the same plant can destroy passivity or fail to result in a stable feedback loop. To use SRG calculus on these embedded operators, we generalize the SRG calculus from~\cite{ryuScaledRelativeGraphs2022} to restricted input spaces\extver{.}{, both for incremental and non-incremental analysis.} Using these interconnection rules, we provide stability theorems for nonlinear systems that guarantee well-posedness, causality and incremental $L_2$-gain bounds, which is achieved by generalizing the (incremental) homotopy results from~\cite{chaffeyHomotopyTheoremIncremental2025} to the non-square MIMO setting. For stable LTI systems, which need not be square, we provide a computationally tractable algorithm to exactly compute its SRG, where the restricted input space for tall operators is respected. We also provide bounds for the SRG of several diagonal and non-square static nonlinear blocks. Throughout the paper, systems in \emph{Linear Fractional Representation} (LFR) form are considered\footnote{LFRs are representations composed of an LTI block feedback connected to a (diagonal) block of static nonlinearities.}, as they are highly useful in control design and they are often used to design non-square feedback systems.

\extver{
    We demonstrate our methods on two examples; a Lur'e system with non-square components and a SISO system with multiple nonlinearities. The first example demonstrates the reduced conservatism compared to naive squaring up due to our embedding approach, and the second example compares our methods with the approach proposed in~\cite{krebbekxScaledRelativeGraph2025}.
}{
    We demonstrate our methods on four examples; a Lur'e system with non-square components, a SISO system with multiple nonlinearities from~\cite{krebbekxScaledRelativeGraph2025}, a MIMO mass-spring-damper system with multiple nonlinear springs, and the example in~\cite{grootExploitingStructureMIMO2025}. The first example demonstrates the reduced conservatism compared to naive squaring due to our embedding approach, the second example compares our methods with the approach proposed in~\cite{krebbekxScaledRelativeGraph2025} and the last example compares our methods with the approach proposed in~\cite{grootExploitingStructureMIMO2025}.
}

The paper is structured as follows\footnote{The approaches developed in this work for the computation of \emph{SRG based analysis of non-square
MIMO nonlinear systems} are publicly available as a software toolbox in Julia: \href{https://github.com/Krebbekx/SrgTools.jl}{\texttt{github.com/Krebbekx/SrgTools.jl}}.}. In Section~\ref{sec:preliminaries}, we present the required preliminaries that lay out the mathematical setting. The SRG of MIMO operators is defined in Section~\ref{sec:mimo_srg_def}, where the embedding procedure is detailed. Next, we state and prove our stability theorems for system analysis in Section~\ref{sec:feedback_systems}, which includes a practical theorem for systems in LFR form. We provide formulas to bound the MIMO SRG of stable LTI operators and common static nonlinear operators in Section~\ref{sec:computing_mimo_srgs} and finally demonstrate our methods on \extver{two}{several} examples in Section~\ref{sec:examples} and present our conclusions in Section~\ref{sec:conclusion}.

\section{Preliminaries}\label{sec:preliminaries}

\subsection{Notation and Conventions}

Let $\N$ ($\N^+$) denote the (nonzero) natural numbers, and $\R, \C$ denote the real and complex number fields, respectively, with $\R_{\geq 0} = [0, \infty)$, $j$ being the imaginary unit, and $\bar{z}$ being the complex conjugate of $z \in \C$. Let $\C_\infty := \C \cup \{ \infty \}$ denote the extended complex plane. \extver{For sets $A, B \subseteq \C$, $A+B$ and $AB$ are the Minkowski sum and product.}{For sets $A, B \subseteq \C$, the sum and product sets are defined as $A+B:= \{ a+b \mid a\in A, b\in B\}$ and $AB:= \{ ab \mid a\in A, b\in B\}$, respectively.} A closed disk in $\C$ is denoted by $D_r(x) = \{ z \in \C \mid |z-x| \leq r \}$, while $D_{[\alpha, \beta]}$ is a complex disk in $\C$ centered on $\R$, intersecting $\R$ on $[\alpha, \beta]$. The radius of a set $\mathcal{C} \subseteq \C$ is defined as $\rmin(\mathcal{C}) := \inf_{r>0} : \mathcal{C} \subseteq D_r(0)$. The distance between two sets $\mathcal{C}_1,\mathcal{C}_2 \subseteq \C_\infty$ is defined as $\dist(\mathcal{C}_1,\mathcal{C}_2) := \inf_{z_1 \in \mathcal{C}_1, z_2 \in \mathcal{C}_2} |z_1-z_2|$, where $|\infty-\infty|:=0$. We denote a transfer function of a state-space realization $(A,B,C,D)$ as $G(s) =C(s I  - A)^{-1}B+D=  \scalebox{0.6}{$\left[\begin{array}{c|c} A&B\\ \hline C&D \end{array} \right]$}$ and $RH_\infty^{q \times p}$ as the space of proper and stable real rational $q \times p$ transfer matrices.

\subsection{Relations and Operators}

A relation $R : X \to Y$ is a possibly multi-valued map, defined by $Rx \subseteq Y$ for all $x \in X = : \dom(R)$, and the range is defined as $\ran(R) := \{ y\in Y \mid \exists x \in X : y \in Rx \} \subseteq Y$.\extver{}{ The graph of a relation $R$ is the set $\{ (x,y) \in X \times Y \mid x \in X, y \in Rx \}$. Given the sets $X,Y,Z$ and relations $R, S :X\to Y$ and $T : Y \to Z$, the inverse $R^{-1}$, sum $R+S$ and product $TR$ are defined as
\begin{subequations}\label{eq:relations}
\begin{align}
    R^{-1} &= \{ (y,x) \mid (x,y) \in R \}, \label{eq:relations-inverse}\\
    R+S &= \{ (x,y+z) \mid (x,y) \in R, (x,z) \in S \}, \label{eq:relations-sum}\\
    T R &= \{ (x,z) \mid \exists y:(x,y) \in R, \, (y,z)\in T \}. \label{eq:relations-product}
\end{align}
\end{subequations}} %
A single-valued relation is called an operator, where $Rx \in Y$ is understood. We denote the set of operators from $X$ to $Y$ as 
\begin{equation}\label{eq:set-of-operators}
    \mathcal{N}(X,Y):= \{ R : X \to Y \mid \forall x \in X, \, Rx \in Y \},
\end{equation}
where $\mathcal{N}(X,X) =: \mathcal{N}(X)$. %
The identity operator on a space $X$ is defined as $I_X x = x, \, \forall x \in X$. If $R \in \mathcal{N}(X,Y)$ is injective, then $R^{-1} \in \mathcal{N}(\ran(R),X)$ such that $R^{-1} R = I_X$ and $R R^{-1} = I_{\ran(R)}$. If $R$ is not injective, then $R^{-1}$ is multivalued.

For an operator $R: X \to Y$, where $(X,\norm{\cdot}_X)$ and $(Y,\norm{\cdot}_Y)$ are normed spaces, we define the induced incremental norm as (similar to the notation in~\cite{vanderschaftL2GainPassivityTechniques2017})
\begin{equation}\label{eq:incremental_induced_norm}
    \maketextstyle \Gamma(R):= \sup_{x_1,x_2 \in X} \frac{\norm{R x_1-R x_2}_Y}{\norm{x_1-x_2}_X} \in [0,\infty].
\end{equation}
We define the induced non-incremental norm as
\begin{equation}\label{eq:non_incremental_induced_norm}
    \maketextstyle \gamma(R):= \sup_{x \in X} \frac{\norm{R x}_Y}{\norm{x}_X} \in [0,\infty].
\end{equation}

\subsection{Signal Spaces}

Let $\mathcal{L}$ denote a Hilbert space, equipped with an inner product $\inner{\cdot}{\cdot} : \mathcal{L} \times \mathcal{L} \to \C$ and norm $\norm{x} := \sqrt{\inner{x}{x}}$. For $d \in \N^+$ the Hilbert spaces of interest are
\begin{equation}\label{eq:L2d-space}
    L_2^d := \{ f: \R_{\geq 0} \to \R^d \mid \norm{f} <\infty \},
\end{equation}
with inner product $\inner{f}{g}:= \int_\R f(t)g(t) d t$, and the superscript $d$ is dropped if $d=1$. For $\R^d$, the inner product is $x y = \sum^d_{i=1} x_i y_i$. \extver{}{Elements $f \in L_2^d$ are denoted as $f=(f_1 \dots f_d)^\top$, where $f_i \in L_2$ for all $i=1,\dots,d$. Furthermore, $0 \in L_2^d$ refers to the map $\R{\geq 0} \ni t \mapsto 0 \in \R^d$.} For $n \leq d$, we define the linear subspaces (which are Banach spaces)
\begin{equation}\label{eq:subspace-Un-zeros}
    \mathcal{U}^d_n := \{ f \in L_2^d(\mathbb{T}, \mathbb{F}) \mid f_i = 0 \text{ for } i>n \},
\end{equation} 
where the superscript $d$ is dropped if it is clear from the context.

For any $T \in \R_{\geq 0}$, define the truncation operator $P_T$ on signals $u: \R_{\geq 0} \to \R^d$ as $(P_T u)(t) = 0$ for all $t>T$, else $(P_T u)(t) = u(t)$.
The extension of $L_2^d$, see Ref.~\cite{desoerFeedbackSystemsInputoutput1975}, is defined as $\Lte^d := \{ u : \R_{\geq 0} \to \R^d \mid \norm{P_T u}_2 < \infty \text{ for all } T \in \R_{\geq 0} \}$ the superscript $d$ is dropped if $d=1$.\extver{}{ The extended space is the natural setting for modeling systems, as it includes periodic signals, which are otherwise excluded from $L_2$. However, extended spaces are not even normed spaces~\cite[Ch. 2.3]{willemsAnalysisFeedbackSystems1971}. Therefore, the Hilbert space $L_2^d(\mathbb{T}, \mathbb{F})$ is the adequate signal space for rigorous functional analytic system analysis.} %

\subsection{Systems and Stability}

An operator $R$ is said to be causal on $L_2^p$ ($\Lte^p$) if it satisfies $R : L_2^p \to L_2^q$ ($R : \Lte^p \to \Lte^q$) and $P_T (Ru) = P_T(R(P_Tu))$ for all $u \in L_2^p$ ($u \in \Lte^p$) and $T \in \R_{\geq 0}$, i.e., the output at time $t$ is independent of the signal at times greater than $t$. \extver{Unless otherwise specified, \emph{we will always assume causality on $L_2^p$ and $R(0)=0$.}}{Causal systems $R:L_2^p \to L_2^q$ are extended to $\Lte^p$ by defining $R : \Lte^p \to \Lte^q$ as $P_T R u := P_T R P_T u$, which is well-defined since $P_T u \in L_2^p$ for all $u \in \Lte^p$. If $R :L_2^p \to L_2^q$ and $R : \Lte^p \to \Lte^q$, then $R$ is causal on $L_2^p$ if and only if $R$ is causal on $\Lte^p$~\cite[Ch. 2.4]{willemsAnalysisFeedbackSystems1971}. If $R: \dom(R) \subsetneq L_2^p \to L_2^q$, then $R$ can only be extended to $\Lte^p$ if $\norm{P_T R u}_2 < \infty$ for all $u\in L_2^p$ and $T \in \R_{\geq 0}$, i.e., no finite escape time. Conversely, if $R: \Lte^p \to \Lte^q$, then it can be that $Ru \notin L_2^q$ for all $u \in L_2^p$ (e.g., consider $u(t) \mapsto \sin(t)$). We model physical systems as operators that take inputs in an extended signal space, i.e., $R: \Lte^p \to \Lte^q$, and we always assume causality and $R(0)=0$, unless otherwise specified.}

A system $R$ is said to be $L_2$-stable, if $R : L_2^p \to L_2^q$. We define the incremental $L_2$-gain as $\Gamma(R)$ from~\eqref{eq:incremental_induced_norm}. When $R:L_2^p \to L_2^q$ and $\Gamma(R) < \infty$, we call the system incrementally stable. Similarly, non-incremental gain and stability is defined in terms of~\eqref{eq:non_incremental_induced_norm}. \extver{}{ The general approach in this work is to show that $\Gamma(R) <\infty$ ($\gamma(R)< \infty$) on $\dom(R) \subseteq L_2^p$ and separately show that $\dom(R) = L_2^p$. The final step is to extend the domain to $\Lte^p$ by proving, or assuming, causality.}

Note that the incremental gain $\Gamma(R)$ is computed using signals in $L_2^p$ only. If $R : L_2^p \to L_2^q$ happens to be causal, this gain carries over to $\Lte^p$ in the sense that~\cite{krebbekxScaledRelativeGraph2025}
\begin{equation*}
    \maketextstyle \Gamma(R) = \sup_{u_1, u_2 \in \Lte^p} \sup_{T \in \R_{\geq 0}} \frac{\norm{P_T R u_1 - P_T R u_2}}{\norm{P_T u_1 - P_T u_2}}.
\end{equation*} 
and similarly for $\gamma(R)$.

\subsection{The Scaled Relative Graph}\label{sec:srg-definition}

Let $\mathcal{L}$ be a Hilbert space, and $R : \dom(R) \subseteq \mathcal{L} \to \mathcal{L}$ a relation. Define the angle between $u, y\in \mathcal{L}$ as 
\begin{equation}\label{eq:def_srg_angle}
    \maketextstyle \angle(u, y) := \cos^{-1} \frac{\mathrm{Re} \inner{u}{y}}{\norm{u} \norm{y}} \in [0, \pi].
\end{equation}
Given distinct $u_1, u_2 \in \mathcal{U} \subseteq \dom(R)$, we define the set
\begin{multline*}
    \maketextstyle z_R(u_1, u_2) := \: \{\infty \mid \text{if } Ru_1 \text{ or } Ru_2 \text{ is multi-valued }\} \; \cup \\ \maketextstyle \left\{ \frac{\norm{y_1-y_2}}{\norm{u_1-u_2}} e^{\pm j \angle(u_1-u_2, y_1-y_2)} \mid y_1 \in R u_1, y_2 \in R u_2 \right\} .
\end{multline*}
The SRG of $R$ over the set $\mathcal{U}$ is defined as
\begin{equation*}
    \maketextstyle \SRG_\mathcal{U} (R) := \bigcup_{u_1, u_2 \in \mathcal{U}, \; u_1 \neq u_2} z_R(u_1, u_2) \subseteq \C_\infty,
\end{equation*}
and we denote $\SRG(R) := \SRG_{\dom(R)}(R)$. Note that $0 \in \SRG_\mathcal{U}(R)$, if and only if there exist $u_1, u_2 \in \mathcal{U}$, $u_1 \neq u_2$, such that $R u_1 = R u_2$. 

\extver{}{One can also define the \emph{Scaled Graph} (SG) around some particular input. }The \extver{\emph{Scaled Graph} (SG)}{SG} of a relation $R$ with one input $u^\star \in \dom(R)$ fixed and the other in set $\mathcal{U}$ is defined as
\begin{equation}
    \SG_{\mathcal{U}, u^\star}(R) := \{ z_R(u, u^\star) \mid u \in \mathcal{U} \setminus u^\star \}.
\end{equation}
We use the shorthand $\SG_{\dom(R), u^\star}(R) = \SG_{u^\star}(R)$.\extver{}{ The SG around $u^\star=0$ is particularly interesting, because the radius of $\SG_{0}(R)$ gives a non-incremental $L_2$-gain bound for the operator $R$.} By definition of the SRG, the incremental gain of an operator $R$, defined in~\eqref{eq:incremental_induced_norm}, is equal to the radius of the SRG, i.e., $\Gamma(R) = \rmin(\SRG(R))$. Similarly, $\gamma(R) = \rmin(\SG_0(R))$ for the non-incremental gain in~\eqref{eq:non_incremental_induced_norm}.

The SRG interconnection rules in~\cite{ryuScaledRelativeGraphs2022} are restricted to full-domain operators and do not apply to restricted-domain SRGs. We therefore prove the following generalization.

\thmspace
\begin{proposition}\label{prop:srg_calculus}
    Let $0 \neq \alpha \in \R$, let $R,S, T : X \to X$ be relations on a Hilbert space $X$ and linear subspaces $\mathcal{U}, \mathcal{Y} \subseteq X$ such that $R(\mathcal{U}) \subseteq \mathcal{Y}$. Then, 
    \begin{enumerate}[label=\alph*.]
        \item\label{eq:srg_calculus_alpha} $\SRG_\mathcal{U}(\alpha R) = \SRG_\mathcal{U}(R \alpha) = \alpha \SRG_\mathcal{U}(R)$,
        \item\label{eq:srg_calculus_plus_one} $\SRG_\mathcal{U}(I_\mathcal{U} + R) = 1 + \SRG_\mathcal{U}(R)$, where $I_\mathcal{U}$ obeys $I_\mathcal{U} u=u$ for all $u \in \mathcal{U}$,
        \item\label{eq:srg_calculus_inverse} $\SRG_{R(\mathcal{U})}(R^{-1}) = (\SRG_\mathcal{U}(R))^{-1} =: \SRG_\mathcal{U}(R)^{-1}$ (where $0,\infty \in \SRG_\mathcal{U}(R)$ are allowed).
        \item\label{eq:srg_calculus_parallel} If at least one of $\SRG_\mathcal{U}(R), \SRG_\mathcal{U}(S)$ satisfies the chord property, then $\SRG_\mathcal{U}(R + S) \subseteq \SRG_\mathcal{U}(R) + \SRG_\mathcal{U}(S)$.
        \item\label{eq:srg_calculus_series} If at least one of $\SRG_\mathcal{U}(R), \SRG_\mathcal{Y}(T)$ satisfies an arc property, then $\SRG_\mathcal{U}(T R) \subseteq \SRG_\mathcal{Y}(T) \SRG_\mathcal{U}(R)$.
    \end{enumerate}
    See \extver{\cite{ryuScaledRelativeGraphs2022} for the chord and arc property definitions}{Definitions~\ref{def:chord_property} and \ref{def:arc_property} for the chord and arc property}. The SRGs above may contain $0, \infty$. If any of the SRGs above are $\emptyset, \{ 0 \}$ or $\{ \infty \}$, extra care is required, see Ref.~\cite{ryuScaledRelativeGraphs2022}. 
\end{proposition}
\thmspace

\extver{
    \begin{proof}
        We prove 1.b as it does not follow trivially from~\cite{ryuScaledRelativeGraphs2022} because it utilizes a projection on the subspace $\mathcal{U}$. From \cite{ryuScaledRelativeGraphs2022} we use $\mathrm{Re} \, z_R(u_1,u_2)= \frac{\inner{Ru_1-Ru_2}{u_1-u_2}}{\norm{u_1-u_2}^2}$ and $\mathrm{Im} \, z_R(u_1,u_2) = \pm \frac{\norm{\pi_{(u_1-u_2)^\perp}(Ru_1-Ru_2)}}{\norm{u_1-u_2}}$, where $\pi_{x^\perp}$ is the projection on the subspace orthogonal to $x$. Since for all $u_1,u_2 \in \mathcal{U}$ it holds that $\inner{(I_\mathcal{U}+R)u_1-(I_\mathcal{U}+R)u_2}{u_1-u_2} = \norm{u_1-u_2}^2 + \inner{Ru_1-Ru_2}{u_1-u_2}$ and $\pi_{(u_1-u_2)^\perp}((I_\mathcal{U}+R)u_1-(I_\mathcal{U}+R)u_2) = \pi_{(u_1-u_2)^\perp}(Ru_1-Ru_2)$, from which it follows that $\SRG_\mathcal{U}(I_\mathcal{U}+R) = 1 + \SRG_\mathcal{U}(R)$.

        The rest of the proof mimics~\cite[Section 4]{ryuScaledRelativeGraphs2022} with the additional condition that the subspaces $\mathcal{U}, \mathcal{Y}$ are appropriately used. The complete proof can be found in the Appendix of the extended version, see~\cite{krebbekxMIMOPaperExtended2026}.
    \end{proof}
}{
    The proof can be found in the Appendix.
}

\subsection{Nonlinear Systems in LFR Form}\label{sec:mimo_lfr_analysis}

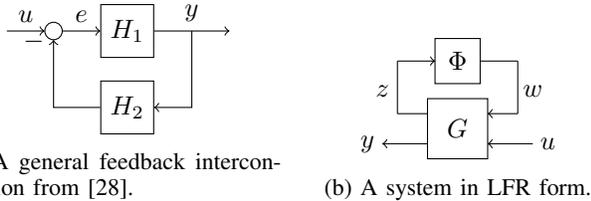
\begin{figure}[t]
    \centering

    \begin{subfigure}[b]{0.49\linewidth}
        \centering

        \tikzstyle{block} = [draw, rectangle, 
        minimum height=1.5em, minimum width=2em]
        \tikzstyle{sum} = [draw, circle, scale=.7, node distance={0.5cm and 0.5cm}]
        \tikzstyle{input} = [coordinate]
        \tikzstyle{output} = [coordinate]
        \tikzstyle{pinstyle} = [pin edge={to-,thin,black}]
        
        \begin{tikzpicture}[auto, node distance = {0.2cm and 0.5cm}]
            \node [input, name=input] {};
            \node [sum, right = of input] (sum) {$ $};
            \node [block, right = of sum] (lti) {$H_1$};
            \node [coordinate, right = of lti] (z_intersection) {};
            \node [output, right = of z_intersection] (output) {}; %
            \node [block, below = of lti] (static_nl) {$H_2$};
        
            \draw [->] (input) -- node {$u$} (sum);
            \draw [->] (sum) -- node {$e$} (lti);
            \draw [->] (lti) -- node [name=z] {$y$} (output);
            \draw [->] (z) |- (static_nl);
            \draw [->] (static_nl) -| node[pos=0.99] {$-$} (sum);
        \end{tikzpicture}
        
        \caption{A general feedback interconnection from~\cite{desoerFeedbackSystemsInputoutput1975}.}
        \label{fig:small-gain-setup}
    \end{subfigure}
    \hfill
    \begin{subfigure}[b]{0.49\linewidth}
        \centering
        \begin{tikzpicture}
            \draw[black] (-.4,-.4) rectangle (.4,.4);
            \node at (0,0) {$G$};

            \draw[black] (-.3,-.3+.9) rectangle (.3,.3+.9);
            \node at (0,.9) {$\Phi$};

            \draw[<-] (-1,-.2) -- (-.4,-.2);
            \node at (-1.2,-.2) {$y$};

            \draw[<-] (.4,-.2) -- (1,-.2);
            \node at (1.2,-.2) {$u$};

            \draw[-] (-.8,.2) -- (-.4,.2);
            \draw[<-] (.4,.2) -- (.8,.2) ;

            \draw[-] (-.8,.2) -- (-.8,.9);
            \draw[-] (.8,.2) -- (.8,.9);

            \draw[-] (.8,.9) -- (.3,.9);
            \draw[->] (-.8,.9) -- (-.3,.9);

            \node at (-1,.5) {$z$};
            \node at (1,.5) {$w$};
        \end{tikzpicture}
        \caption{A system in LFR form.}
        \label{fig:lfr_general}
    \end{subfigure}
    \caption{Two canonical feedback interconnections.}
    \label{fig:feedback_interconnections}
    \vspace{-0.25cm}
\end{figure}

Existing analysis of nonlinear feedback systems from the input/output perspective often considers the system depicted in Fig.~\ref{fig:small-gain-setup}. This interconnection, where $H_1$ and $H_2$ are nonlinear operators, is the canonical system in~\cite{willemsAnalysisFeedbackSystems1971,desoerFeedbackSystemsInputoutput1975} and is used to develop all classical results: the small gain and passivity theorems, and in the case that $H_1$ is LTI, the circle and Popov criteria. Even though the interconnection in Fig.~\ref{fig:small-gain-setup} covers many feedback systems, there are many situations that require a more sophisticated modeling approach.

Instead, we will consider systems in LFR form, as shown in Fig.~\ref{fig:lfr_general}, where $G$ contains all LTI dynamics and $\Phi$ all static or dynamic nonlinearities. The LFR is used thoroughly in robust control~\cite{zhouRobustOptimalControl1996} for uncertain LTI systems, in nonlinear system identification~\extver{\cite{schoukensIdentificationBlockorientedNonlinear2017}}{\cite{paduartIdentificationNonlinearSystems2010,vanbeylenNonlinearLFRBlockOriented2013,schoukensIdentificationBlockorientedNonlinear2017}}, in machine learning~\cite{elghaouiImplicitDeepLearning2021}, and for the analysis of recurrent neural networks~\cite{yinStabilityAnalysisUsing2022}. The reason for this is that a broad class of systems can be written in LFR form.\extver{}{ For example, the system in Fig.~\ref{fig:small-gain-setup} can always be represented in LFR form, and it is even true that for most dynamic nonlinearities, by appropriate loop transformation, all dynamics can be gathered in the LTI block $G$, while $\Phi$ becomes static diagonal. We emphasize, however, that the applicability of the framework developed in this paper reaches beyond the feedback systems in Fig.~\ref{fig:feedback_interconnections}.} A system $R : \Lte^p \to \Lte^q$ is said to be in LFR form if it is decomposable to the interconnection in Fig.~\ref{fig:lfr_general} in terms of an LTI operator $G \in RH_\infty^{(q+n_\mathrm{z}) \times (p+n_\mathrm{w})}$ partitioned as
\begin{equation}\label{eq:lfr_lti_part}
    G = \scalebox{0.75}{$\mat{G_\mathrm{zw} & G_\mathrm{zu} \\ G_\mathrm{yw} & G_\mathrm{yu}} $},
\end{equation}
whith $G_\mathrm{zw} \in RH_\infty^{n_\mathrm{z} \times n_\mathrm{w}}$, $G_\mathrm{zu} \in RH_\infty^{n_\mathrm{z} \times p}$, $G_\mathrm{yw} \in RH_\infty^{q \times n_\mathrm{w}}$ and $G_\mathrm{yu} \in RH_\infty^{q \times p}$, and a nonlinear operator $\Phi : \Lte^{n_\mathrm{z}} \to \Lte^{n_\mathrm{w}}$. The LTI block $G$ describes the dynamic signals evolution from which the static or dynamic nonlinearities, collected in $\Phi$, are lifted out in terms of a feedback connection.\extver{}{ Note that $\Phi$ may also be used to represent uncertain nonlinear effects. In the robust control framework~\cite{zhouRobustOptimalControl1996}, one collects all linear uncertainties in $\Phi$, commonly denoted as $\Delta$ instead.} 

The closed-loop operator $y=Ru$ based on Fig.~\ref{fig:lfr_general} can be written as
\begin{subequations}\label{eqs:lfr_closed_loop}
\begin{align}
    R &= G_\mathrm{yw} \Phi (I - G_\mathrm{zw} \Phi)^{-1} G_\mathrm{zu} + G_\mathrm{yu} \label{eq:lfr_closed_loop_operator} \\ &= G_\mathrm{yw}(\Phi^{-1}-G_\mathrm{zw})^{-1} G_\mathrm{zu}+G_\mathrm{yu}. \label{eq:lfr_closed_loop}
\end{align}
\end{subequations}
Equality of~\eqref{eq:lfr_closed_loop_operator} and~\eqref{eq:lfr_closed_loop} is understood in terms of \emph{relations}, which is proven in \extver{the extended version~\cite{krebbekxMIMOPaperExtended2026}}{Appendix~\ref{app:proof_lfr_relational_equivalence}}. We note that~\eqref{eq:lfr_closed_loop_operator} is the adequate \emph{operator} expression for the closed loop, where well-posedness is determined by invertibility of $I - G_\mathrm{zw} \Phi$. In contrast,~\eqref{eq:lfr_closed_loop} is always valid as a \emph{relation}, even if $\Phi$ is not invertible, and is preferred over~\eqref{eq:lfr_closed_loop_operator} for SRG computations as it involves less operator multiplications.

\textbf{Problem statement:} The main objective in this paper is to analyze the stability, well-posedness and $L_2$-gain of the closed-loop system in~\eqref{eq:lfr_closed_loop}. The strategy to analyze~\eqref{eq:lfr_closed_loop} using SRG methods would be to replace each individual operator with its SRG, and resolve the sum, product and inverse relations using Proposition~\ref{prop:srg_calculus}. However, the dimensions $p,q,n_\mathrm{w},n_\mathrm{z}$ may take any value in $\mathbb{N}_+$, so the operators in~\eqref{eq:lfr_closed_loop} can be non-square. Therefore, the LFR exemplifies the need for an SRG framework that can deal with non-square MIMO operators.

\section{The SRG for MIMO Operators}\label{sec:mimo_srg_def}

At the core of the SRG definition lies a Hilbert space that contains both the domain and range. For an operator $R : L_2^p \to L_2^q$, this definition is problematic when $p \neq q$, i.e., the operator is not square. As a first step towards a general framework for SRG analysis of multivariable feedback systems, we will develop the required mathematical tools to \emph{embed} operators $R \in \mathcal{N}(L_2^p, L_2^q)$ into $\mathcal{N}(L_2^n)$ for some $n \geq p,q$. Using these embeddings, we will define the SRG for a general MIMO operator, called the MIMO SRG. This embedding procedure is nontrivial since the original domain of the operator must be preserved to avoid conservatism.

\subsection{Embedding of MIMO Operators}\label{sec:mimo_srg_technicalities}

Throughout this section, we consider $p,q \in \N_+$ and $R: L_2^p \to L_2^q$, and define $n:=\max \{p,q \}$. One can be in one of three cases: $p>q$ (wide), $p<q$ (tall) and $p=q$ (square). For each case, we will consider the appropriate embedding.

\subsubsection{Wide operators} When $p>q$, we append outputs that map to zero by defining the linear embedding operator 
\begin{equation}\label{eq:injection-operator}
\begin{aligned}
    \iota_{p \leftarrow q} & : L_2^q \to L_2^p,  \\[-2.5ex]
    (f_1,\dots,f_q)^\top & \mapsto (f_1,\dots,f_q,\overbrace{0\dots,0}^{p-q})^\top.
\end{aligned}
\end{equation}
This way, $\iota_{p \leftarrow q}$ embeds $L_2^p$ into $L_2^q$ and $R \mapsto \iota_{p \leftarrow q} R$ embeds $\mathcal{N}(L_2^p,L_2^q)$ into $\mathcal{N}(L_2^p)$. 

\thmspace
\begin{lemma}\label{lemma:injection-isomorphism}
    The embedding operator $\iota_{p \leftarrow q} : L_2^q \to \mathcal{U}_q^p$ is an isometric isomorphism.
\end{lemma}

\begin{proof}%
    For all $u \in L_2^q$ we have $\norm{u}^2=\sum_{i=1}^q \int_{\R_{\geq 0}} |u_i(t)|^2 d t$. By definition $\norm{\iota_{p \leftarrow q} u}^2 = \sum_{i=1}^p \int_{\R_{\geq 0}} |u_i(t)|^2 d t=\sum_{i=1}^q \int_{\R_{\geq 0}} |u_i(t)|^2 d t=\norm{u}^2$, therefore $\norm{\iota_{p \leftarrow q} u} = \norm{u}$.
\end{proof}

For a wide operator $R$ and $\mathcal{U} \subseteq L_2^p$, we \emph{define} the SRG as
\begin{equation}
    \SRG_\mathcal{U}(R) := \SRG_\mathcal{U}(\iota_{p \leftarrow q} R).
\end{equation}
It is important to note that $\frac{\norm{Ru}}{\norm{u}} = \frac{\norm{\iota_{p \leftarrow q} R u}}{\norm{u}}$ for all $u \in L_2^p$, which guarantees the important relations $\Gamma(R) = \rmin(\SRG(R))$ and $\gamma(R) = \rmin(\SG_0(R))$. 

\subsubsection{Tall operators} When $p<q$, we add extra inputs that are ignored in the output. We define the linear projection operator 
\begin{equation}\label{eq:projection-operator}
\begin{aligned}
    \pi_{p \leftarrow q} & : L_2^q \to L_2^p, \\
    (f_1,\dots,f_p,r_1,\dots,r_{q-p})^\top & \mapsto (f_1,\dots,f_p)^\top.
\end{aligned}
\end{equation}
This way, $\pi_{p \leftarrow q}$ projects $L_2^q$ onto $L_2^p$ and $R \mapsto R \pi_{p \leftarrow q}$ embeds $\mathcal{N}(L_2^p,L_2^q)$ into $\mathcal{N}(L_2^q)$. 

\thmspace
\begin{lemma}\label{lemma:projection-isomorphism}
    The projection operator $\pi_{p \leftarrow q} : \mathcal{U}_p^q \to L_2^p$ is an isometric isomorphism with inverse $\pi_{p \leftarrow q}^{-1} = \iota_{q \leftarrow p}$
\end{lemma}

\begin{proof}%
    By Lemma~\ref{lemma:injection-isomorphism}, we know that $\iota_{q \leftarrow p} : L_2^p \to \mathcal{U}_p^q$ is an isomorphism. Since for all $u \in L_2^p$ and $\tilde{u} \in \mathcal{U}_p^q$ it holds that $\pi_{p \leftarrow q} \iota_{q \leftarrow p} u=u$ and $\iota_{q \leftarrow p} \pi_{p \leftarrow q} \tilde{u} = \tilde{u}$, we have shown that $\iota_{q \leftarrow p}^{-1} = \pi_{p \leftarrow q}$, hence $\pi_{p \leftarrow q}^{-1} = \iota_{q \leftarrow p}$. 
\end{proof}

The extra inputs $r_1,\dots,r_{q-p}$ should not belong to the description of $R$. Therefore, we restrict the inputs of $R \pi_{p \leftarrow q}$ to the space $\mathcal{U}_p^q$ in~\eqref{eq:subspace-Un-zeros}, resulting in $\{ r_i = 0\}_{i=1}^{q-p}$. For a tall operator $R$ and $\mathcal{U} \subseteq L_2^p$, we \emph{define} the SRG as
\begin{equation}
    \SRG_\mathcal{U}(R) := \SRG_{\iota_{q \leftarrow p}(\mathcal{U})}(R \pi_{p \leftarrow q}).
\end{equation}
Note that $u = \pi_{p \leftarrow q} \iota_{q \leftarrow p} u$, so $\iota_{q \leftarrow p}(\mathcal{U})$ precisely contains all inputs to $R$. Therefore, we can conclude $\frac{\norm{Ru}}{\norm{u}} = \frac{\norm{Ru}}{\norm{\iota_{q \leftarrow p} u}}$ for all $u \in L_2^p$, which guarantees the important relations $\Gamma(R) = \rmin(\SRG(R))$ and $\gamma(R) = \rmin(\SG_0(R))$. 

\subsubsection{Square operators} Now consider $R : L_2^n \to L_2^n$ and choose any $\tilde{n} > n$. We embed $R$ into $\mathcal{N}(L_2^{\tilde{n}})$ using the map $R \mapsto \iota_{\tilde{n} \leftarrow n} R \pi_{n \leftarrow \tilde{n}}$. Using the appropriate input space, this embedding can be done without introducing conservatism. 

\thmspace
\begin{lemma}\label{lemma:embedding-square}
    Let $R : L_2^n \to L_2^n$ and choose any $\tilde{n} > n$, then
    \begin{equation}
        \SRG_\mathcal{U}(R) = \SRG_{\iota_{\tilde{n} \leftarrow n}(\mathcal{U})}(\iota_{\tilde{n} \leftarrow n} R \pi_{n \leftarrow \tilde{n}}).
    \end{equation}
\end{lemma}

\begin{proof}%
    Denote $\iota:=\iota_{\tilde{n} \leftarrow n}$ and $\pi:= \pi_{n \leftarrow \tilde{n}}$ for brevity. By Lemma~\ref{lemma:injection-isomorphism} and \ref{lemma:projection-isomorphism}, we know that $\iota$ is an isometric isomorphism with isometric isomorphic inverse $\pi$. Therefore, $u_1, u_2 \in \mathcal{U} \iff \tilde{u}_1,\tilde{u}_2 \in \iota(\mathcal{U})$ where $\tilde{u}_1=\iota(u), \tilde{u}_2=\iota(u_2)$ satisfying $\norm{u_1-u_2} =\norm{\tilde{u}_1-\tilde{u}_2}$ and
    \begin{multline*}
        \norm{\iota R \pi \tilde{u}_1 - \iota R \pi \tilde{u}_2} = \norm{\iota R u_1 - \iota R \pi u_2} = \norm{R u_1 - R \pi u_2}, \\
        \inner{\iota R \pi \tilde{u}_1 - \iota R \pi \tilde{u}_2}{\tilde{u}_1 -\tilde{u}_2} = \inner{\iota R u_1 - \iota R u_2}{\tilde{u}_1 -\tilde{u}_2} \\ = \inner{R u_1 - R u_2}{u_1 -u_2}.  
    \end{multline*}
    This shows that $z_R \subseteq \SRG_\mathcal{U}(R) \iff z_R \subseteq \SRG_{\iota(\mathcal{U})}(\iota R \pi)$, where $z = z_R(u_1,u_2) = z_R(\tilde{u}_1, \tilde{u}_2)$, proving $\SRG_\mathcal{U}(R) = \SRG_{\iota(\mathcal{U})}(\iota R \pi)$. 
\end{proof}

\subsection{Definition of the MIMO SRG}\label{sec:mimo_srg_def_subsection}

We are now ready to define the SRG for MIMO operators. Note that the MIMO SG can be defined similarly.

\thmspace
\begin{definition}\label{def:mimo_srg}
    The SRG of an operator $R: L_2^p \to L_2^q$ over $\mathcal{U} \subseteq L_2^p$ is defined as
    \begin{equation}
        \SRG_\mathcal{U}(R) := \SRG_{\iota_{n \leftarrow p}(\mathcal{U})}(\iota_{n \leftarrow q} R \pi_{p \leftarrow n}),
    \end{equation}
    where $n \geq \max\{ p, q \}$. 
\end{definition}
\thmspace

Note that by Lemma~\ref{lemma:embedding-square}, the choice of $n \geq \max\{ p, q \}$ in Definition~\ref{def:mimo_srg} does not affect the MIMO SRG. A key feature of Definition~\ref{def:mimo_srg} is that the embedding of the domain $\iota_{n \leftarrow p}(\mathcal{U}) = \mathcal{U}^n_p \subsetneq L_2^n$ \emph{does not} artificially introduce extra inputs to the embedded operator $\iota_{n \leftarrow q} R \pi_{p \leftarrow n}$ when $p<n$, which are otherwise not present in the description of $R$. This feature embeds operators in the same space of square operators without introducing \emph{any} conservatism in the embedding step, and can make a huge difference when analyzing the $L_2$-gain of feedback systems.

The idea of adding/removing inputs/outputs of a non-square system to make it square is not new. Previously, such concepts were use to construct square plants that are guaranteed to have minimum-phase~\cite{kouvaritakisGeometricApproachAnalysis1976,misraComputationalAlgorithmSquaringup1992}. %
However, the MIMO SRG in Definition~\ref{def:mimo_srg} has been introduced for computing SRGs and due to the different objective it is significantly different from~\cite{kouvaritakisGeometricApproachAnalysis1976,misraComputationalAlgorithmSquaringup1992}. We consider only \emph{adding} inputs/outputs, not removing any, rendering~\cite{kouvaritakisGeometricApproachAnalysis1976} completely different. Our procedure works for LTI and nonlinear operators alike, and makes no assumption aside $L_2$-stability, in contrast to~\cite{misraComputationalAlgorithmSquaringup1992}, which works for LTI systems only and assumes certain rank conditions on the state-space matrices. Furthermore, Definition~\ref{def:mimo_srg} always adds zero outputs to a wide system and when adding inputs to a tall system, or when $n > \max\{ p, q \}$ is taken, our method does not add extra inputs, but instead computes the SRG for the restricted space of inputs that only contains the original amount of inputs. This is a significant difference with~\cite{misraComputationalAlgorithmSquaringup1992}, which uses nonzero artificial extra inputs and outputs to make the plant minimum-phase.

Similarly, allowing nonzero artificial inputs or outputs, for example to enforce minimum phase or passivity like in~\cite{fradkovPassificationNonsquareLinear2003}, is also problematic. As our aim is to analyze feedback systems like those in Fig.~\ref{fig:small-gain-setup}, when extra non-zero outputs are added to, e.g., $H_1$, the squared-up output contains more signals. This increases the gain bound, but more importantly it hinders the restriction of the SRG of $H_2$ to the space of original input dimension, leading to increased conservatism in SRG analysis.

It is important to note that the ordering of extra zeros in the input/output in Definition~\ref{def:mimo_srg} may be permuted arbitrarily. This is a particular degree of freedom in the corresponding SRG analysis.

\section{Analysis of Non-Square Feedback Systems}\label{sec:feedback_systems}

We now have the tools to describe MIMO operators using SRGs and study their interconnections. However, an SRG bound in itself is not enough to characterize the stability, well-posedness and incremental $L_2$-gain performance of a feedback system, as explained in~\cite{krebbekxScaledRelativeGraph2025}.

Hence, we need practical system analysis tools for the systems in Fig.~\ref{fig:feedback_interconnections} that that guarantee stability, well-posedness and incremental $L_2$-gain performance, which are developed in this section. We first focus on the analysis of the troublesome feedback part of Fig.~\ref{fig:small-gain-setup} and then extend this result to the LFR in Fig.~\ref{fig:lfr_general}.

\subsection{Simple Feedback Systems}\label{sec:simple_feedback_systems}

Consider the feedback interconnection in Fig.~\ref{fig:small-gain-setup}, where $H_1 : L_2^p \to L_2^q$, $H_2 : L_2^q \to L_2^p$, and the closed-loop relation reads $y = (H_1^{-1} + H_2)^{-1} u$. As in \cite{chaffeyHomotopyTheoremIncremental2025}, we abbreviate such a relation as $(H_1^{-1}+H_2)^{-1} =: [H_1,H_2]$. For causal $H_1$ and $H_2$, we call the interconnection $[H_1,H_2]$ \emph{well-posed} if $I + H_2 H_1 $ has a causal inverse on $\Lte^p$~\cite{megretskiSystemAnalysisIntegral1997}.

Since~\cite[Thm. 2]{chaffeyHomotopyTheoremIncremental2025} considers square $H_1,H_2$ on $L_2^n$ and~\cite[Lemma 5]{chaffeyHomotopyTheoremIncremental2025} assumes square operators with full-domain SRGs, these results does not apply to Fig.~\ref{fig:small-gain-setup} for non-square operators and their restricted-domain SRGs. The following theorem provides the necessary generalization using the MIMO SRG from Definition~\ref{def:mimo_srg}.

\thmspace
\begin{theorem}\label{thm:incremental-mimo-srg}
    Consider the causal systems $H_1 : L_2^p \to L_2^q$ and $H_2 : L_2^q \to L_2^p$, where at least one of $\SRG(H_1)$, $\SRG(H_2)$ satisfies the chord property. If for all $\tau \in [0,1]$
    \begin{subequations}
    \begin{align}
        &\rmin(\SRG(H_1)) < \infty \text{ and } \rmin(\SRG(H_2)) < \infty, \label{eq:finite_srg_radii}\\
        &\dist(\SRG(H_1)^{-1},- \tau \SRG(H_2)) \geq r_\tau \geq r >0, \label{eq:finite_srg_separation}
    \end{align}
    \end{subequations}
    then $I + H_2 H_1$ is invertible, and the interconnection is incrementally stable with $\Gamma([H_1,H_2]) \leq 1/r_1$.  Moreover, if $H_1$ and $H_2$ are causal, then $[H_1, H_2]$ is well-posed. 
\end{theorem}
\thmspace

\extver{}{
    Using Proposition~\ref{prop:srg_calculus} instead of~\cite{ryuScaledRelativeGraphs2022}, the proof, given in Appendix~\ref{sec:incremental-stability-thms}. The non-incremental version of Theorem~\ref{thm:incremental-mimo-srg} is given in Appendix~\ref{sec:non-incremental-stability-thms}.
}

\extver{

\begin{proof}%
    For $n \geq \max \{p, q\}$, let $\SRG(H_1):=\SRG_{\mathcal{U}_p^n}(\tilde{H}_1)$ and $\SRG(H_2):=\SRG_{\mathcal{U}_q^n}(\tilde{H}_2)$, where $\tilde{H}_1 = \iota_{n \leftarrow q} H_1 \pi_{p \leftarrow n}$ and $\tilde{H}_2 = \iota_{n \leftarrow p} H_1 \pi_{q \leftarrow n}$. Therefore, $\tilde{H}_1 : \mathcal{U}_p^n \to \mathcal{U}_q^n$ and $\tilde{H}_2 : \mathcal{U}_q^n \to \mathcal{U}_p^n$ and~\eqref{eq:finite_srg_radii} implies $\Gamma(\tilde{H}_1)<\infty$, $\Gamma(\tilde{H}_2) < \infty$.
    \extver{
        From~\eqref{eq:finite_srg_separation} and Proposition~\ref{prop:srg_calculus} we know that $\rmin((\SRG(H_1)^{-1} + \tau \SRG(H_2))^{-1}) \leq 1/r$, and so $\Gamma([\tilde{H}_1,\tau \tilde{H}_2]) \leq 1/r_\tau \leq 1/r$. Using a generalization of~\cite[Theorem 2]{chaffeyHomotopyTheoremIncremental2025} to operators on arbitrary Banach spaces\footnote{See Appendix of~\cite{krebbekxMIMOPaperExtended2026} for the derivation.}, we can conclude that $I+ \tau \tilde{H}_2 \tilde{H_1}$ is invertible on $\mathcal{U}_p^n$ for all $\tau \in [0,1]$ and therefore $\tilde{H_1} (I+ \tilde{H}_2 \tilde{H_1})^{-1} = [\tilde{H}_1, \tilde{H_2}] : \mathcal{U}_p^n \to \mathcal{U}_q^n$ with $\Gamma([\tilde{H}_1, \tilde{H_2}]) \leq 1/r_1$. 
    }{
        From~\eqref{eq:finite_srg_separation} and Proposition~\ref{prop:srg_calculus} we know that $\rmin((\SRG(H_1)^{-1} + \tau \SRG(H_2))^{-1}) \leq 1/r$, and so $\Gamma([\tilde{H}_1,\tau \tilde{H}_2]) \leq 1/r_\tau \leq 1/r$. By Lemma~\ref{lemma:banach-space-incremental-homotopy}, we can conclude that $I+ \tau \tilde{H}_2 \tilde{H_1}$ is invertible on $\mathcal{U}_p^n$ for all $\tau \in [0,1]$ and therefore $\tilde{H_1} (I+ \tilde{H}_2 \tilde{H_1})^{-1} = [\tilde{H}_1, \tilde{H_2}] : \mathcal{U}_p^n \to \mathcal{U}_q^n$ with $\Gamma([\tilde{H}_1, \tilde{H_2}]) \leq 1/r_1$. 
    } 
    
    The final step is to transfer the result to $[H_1,H_2]$. We use the fact that $\pi_{p \leftarrow n} : \mathcal{U}_p^n \to L_2^p$ is an isometric isomorphism with inverse $\iota_{n \leftarrow p}$ (by Lemma~\ref{lemma:injection-isomorphism} and \ref{lemma:projection-isomorphism}). For $u, e \in L_2^p, y \in L_2^q$, define $\tilde{u} = \iota_{n \leftarrow p} u \in \mathcal{U}_p^n, \tilde{e} = \iota_{n \leftarrow p} e \in \mathcal{U}_p^n$ and $\tilde{y} = \iota_{n \leftarrow q} y \in \mathcal{U}_q^n$. Since $\iota$ is a bijection, we have \extver{$u = e + H_2 H_1 e  \iff \tilde{u} = \tilde{e} + \tilde{H}_2 \tilde{H}_1 \tilde{e}$ and $y = H_1 e  \iff \tilde{y} = \tilde{H}_1 \tilde{e}$,}{
    \begin{align*}
        u = e + H_2 H_1 e & \iff \tilde{u} = \tilde{e} + \tilde{H}_2 \tilde{H}_1 \tilde{e}, \\ 
        y = H_1 e & \iff \tilde{y} = \tilde{H}_1 \tilde{e},
    \end{align*}}
    hence $I + H_2 H_1$ is invertible on $L_2^p$ if and only if $I+\tilde{H}_2 \tilde{H}_2$ is invertible on $\mathcal{U}_p^n$. Moreover, since $\pi$ is an isomorphism, one has $\norm{u}=\norm{\tilde{u}}, \norm{e}=\norm{\tilde{e}}$ and $\norm{y}=\norm{\tilde{y}}$, resulting in $\Gamma([H_1,H_2]) = \Gamma([\tilde{H}_1, \tilde{H}_2])$. 

    Note that $H_1$ and $H_2$ are causal if and only if $\tilde{H}_1$ and $\tilde{H}_2$ are causal. Causality of $(I+\tilde{H}_2 \tilde{H}_1)^{-1}$ (and hence $(I + H_2 H_1)^{-1}$) follows from the Banach fixed point theorem on the time axis $[0,T]$, see~\cite{desoerFeedbackSystemsInputoutput1975}. Finally, well-posedness of $[H_1, H_2] = H_1 (I+H_2 H_1)^{-1}$ follows from the causality of $H_1$ and causal invertibility of $I+H_2 H_1$.
\end{proof}

}{}

\subsection{Analysis of Systems in LFR Form}\label{sec:lfr_feedback_systems}

While the interconnection in Fig.~\ref{fig:small-gain-setup} often captures the essential parts of analyzing a feedback system, which are stability and well-posedness, not all systems can be represented in this form. Instead, the LFR form described in Section~\ref{sec:mimo_lfr_analysis} can describe a broader class of nonlinear systems.

To analyze an LFR using SRGs, one first has to compute $n = \max\{ p,q,n_\mathrm{w},n_\mathrm{z} \}$ and embed all operators in~\eqref{eq:lfr_closed_loop} in $\mathcal{N}(L_2^n)$ using Definition~\ref{def:mimo_srg}. Note that for the LFR in~\eqref{eq:lfr_closed_loop}, the input/output dimensions match by definition. Therefore, we can use Proposition~\ref{prop:srg_calculus} and Lemma~\ref{lemma:improved_chord_arc_completions} to conclude
\begin{subequations}\label{eq:lfr_interconnect_steps}
\begin{equation}\label{eq:lfr_interconnect_steps_1}
    \SRG(R) \subseteq \overline{\SRG(G_\mathrm{yw}(\Phi^{-1}-G_\mathrm{zw})^{-1} G_\mathrm{zu}) + \SRG(G_\mathrm{yu})}, \vspace*{-0.6em}
\end{equation}
\begin{multline}\label{eq:lfr_interconnect_steps_2}
    \SRG(G_\mathrm{yw}(\Phi^{-1}-G_\mathrm{zw})^{-1} G_\mathrm{zu}) \subseteq \\ \overline{\SRG(G_\mathrm{yw}(\Phi^{-1}-G_\mathrm{zw})^{-1} )  \SRG(G_\mathrm{zu})},
\end{multline}
\begin{multline}\label{eq:lfr_interconnect_steps_3}
    \SRG(\SRG(G_\mathrm{yw}(\Phi^{-1}-G_\mathrm{zw})^{-1} )) \subseteq \\ \overline{\SRG(G_\mathrm{yw})  \SRG((\Phi^{-1}-G_\mathrm{zw})^{-1} )},
\end{multline}
\begin{equation}\label{eq:lfr_interconnect_steps_4}
    \SRG((\Phi^{-1}-G_\mathrm{zw})^{-1} ) = \SRG(\Phi^{-1}-G_\mathrm{zw} )^{-1},
\end{equation}
\begin{equation}\label{eq:lfr_interconnect_steps_5}
    \SRG(\Phi^{-1}-G_\mathrm{zw} ) \subseteq \overline{\SRG(\Phi)^{-1} - \SRG(G_\mathrm{zw})},
\end{equation}
\end{subequations}
where the overline indicates the improved chord/arc completion from Lemma~\ref{lemma:improved_chord_arc_completions}. Here we used Proposition~\ref{prop:srg_calculus}.\ref{eq:srg_calculus_parallel} in~\eqref{eq:lfr_interconnect_steps_1} and~\eqref{eq:lfr_interconnect_steps_5}, Proposition~\ref{prop:srg_calculus}.\ref{eq:srg_calculus_series} in~\eqref{eq:lfr_interconnect_steps_2} and \eqref{eq:lfr_interconnect_steps_3} and Proposition~\ref{prop:srg_calculus}.\ref{eq:srg_calculus_inverse} in~\eqref{eq:lfr_interconnect_steps_4}. We have obtained a bound in~\eqref{eq:lfr_interconnect_steps} for the SRG of $R$ in~\eqref{eq:lfr_closed_loop}, however, we cannot yet conclude stability from this bound alone.

From~\eqref{eq:lfr_closed_loop} we can see that the stability of the closed loop depends on the stability of $(\Phi^{-1}-G_\mathrm{zw})^{-1}$, i.e.
\begin{equation}\label{eq:lfr_loop}
    [\Phi, -G_\mathrm{zw}] : \dom([\Phi, -G_\mathrm{zw}]) \subseteq L_2^{n_\mathrm{z}} \to L_2^{n_\mathrm{w}},
\end{equation}
which is precisely what we can analyze using Theorem~\ref{thm:incremental-mimo-srg}. We call an LFR system \emph{well-posed} if~\eqref{eq:lfr_loop} is well-posed. 

The incremental gain $\Gamma(R)$ of an LFR system is obtained by replacing the operators in~\eqref{eq:lfr_closed_loop} with their SRG, and computing the radius of the resulting set. 

\thmspace
\begin{theorem}\label{thm:srg-lfr-system}
    Consider the system $R$ given by the LFR in~\eqref{eq:lfr_closed_loop}, where $G \in RH_\infty^{(q+n_\mathrm{z}) \times (p+n_\mathrm{w})}$ and $\Phi: L_2^{n_\mathrm{z}} \to L_2^{n_\mathrm{w}}$ are causal and satisfy $\Gamma(G)< \infty$ and $\Gamma(\Phi)< \infty$. If there exists a $\hat{\Gamma} < \infty$ such that $\forall \, \tau \in [0,1]$
    \begin{subequations}\label{eq:lfr_srg_bound}
    \begin{equation}
        \rmin(\mathcal{G}_{R,\tau}) \leq \hat{\Gamma}, \quad \mathcal{G}_{R} := \mathcal{G}_{R,1},
    \end{equation}
    \begin{multline}
        \mathcal{G}_{R,\tau}:= \overline{\SRG(G_\mathrm{yu}) + \SRG(G_\mathrm{yw})} \\ \overline{ \times (\SRG(\Phi)^{-1}- \tau \SRG(G_\mathrm{zw}))^{-1} \SRG(G_\mathrm{zu})},
    \end{multline}
    \end{subequations}
    then $R : \Lte^p \to \Lte^q$ is well-posed and incrementally stable with $\Gamma(R) \leq \rmin(\mathcal{G}_{R}) \leq \hat{\Gamma}$, where the overline indicates the improved chord/arc completion (see Lemma~\ref{lemma:improved_chord_arc_completions}). %
\end{theorem}

\begin{proof}%
    Stability of~\ref{eq:lfr_closed_loop} depends only on the stability of $[\Phi, -G_\mathrm{zw}]$, since $G : L_2^{p + n_\mathrm{w}} \to L_2^{q + n_\mathrm{z}}$ is assumed to have finite incremental gain. Since $G_\mathrm{zw}$ is part of $G$, it follows that $\Gamma(G_\mathrm{zw}) \leq \Gamma(G)$. If $G_\mathrm{yw},G_\mathrm{zu}$ are not identically zero, then both $z_1 \in \SRG(G_\mathrm{yw})$ and $z_2 \in G_\mathrm{zu}$ where $z_1,z_2 \in \C \setminus \{ 0 \}$. Therefore, the assumption implies that for all $\tau \in [0,1]$
    \begin{equation}\label{eq:srg-distance-ineq}
        \dist(\SRG(\Psi)^{-1}, \tau \SRG(G_\mathrm{zw})) \geq \tilde{r} \geq |z_1| |z_2| / \hat{\Gamma} >0,
    \end{equation}
    hence by Theorem~\ref{thm:incremental-mimo-srg}, $[\Phi, -G_\mathrm{zw}] : L_2^{n_\mathrm{z}} \to L_2^{n_\mathrm{w}}$ is well-posed with $\Gamma([\Phi, -G_\mathrm{zw}]) \leq 1/r$, where $r$ is the largest value of $\tilde{r}$ such that~\eqref{eq:srg-distance-ineq} holds.

    Now $R = G_\mathrm{yw} [\Phi,-G_\mathrm{zw}] G_\mathrm{zu}+G_\mathrm{yu}$ is simply a series and parallel interconnection of operators. Since $[\Phi, -G_\mathrm{zw}] : L_2^{n_\mathrm{z}} \to L_2^{n_\mathrm{w}}$ is well-posed, we can conclude that $R : L_2^p \to L_2^q$ is a well-posed LFR. From the MIMO SRG definition and Proposition~\ref{prop:srg_calculus} it follows that
    \begin{multline*}
        \SRG(R) \subseteq \SRG(G_\mathrm{yu}) + \\\SRG(G_\mathrm{yw})(\SRG(\Phi)^{-1}- \SRG(G_\mathrm{zw}))^{-1} \SRG(G_\mathrm{zu}),
    \end{multline*}
    and therefore $\Gamma(R) \leq \rmin(\mathcal{G}_R^1)$. If one of $G_\mathrm{yw},G_\mathrm{zu}$ is zero, then $R=G_\mathrm{yu}$ is well-posed, causal and the gain bound follows from the definition of the SRG.
\end{proof}

\begin{remark}\label{remark:lfr-prop-non-incremental}
    Theorem~\ref{thm:srg-lfr-system} can be restated in the non-incremental setting if $\Phi(0)=0$ and assuming that $[\Phi, -\tau G_\mathrm{zw}]$ is well-posed for all $\tau \in [0,1]$.
\end{remark}

\extver{}{
\begin{remark}
    It is not necessary to use the improved chord/arc completions. It is sufficient when the chord (arc) property is satisfied for each sum (product) in~\eqref{eq:lfr_srg_bound}. However, this may lead to a larger value of $\hat{\Gamma}$, i.e., more conservatism. 
\end{remark}
}

\extver{}{
Note that $\Phi$ in Theorem~\ref{thm:srg-lfr-system} may be any nonlinear operator with a finite SRG bound, not just a diagonal static nonlinear block. Examples of dynamic nonlinearities are time-varying static nonlinearities such as $x \mapsto \sin(t) \sin(x)$, and hybrid systems. For the latter, SG bounds can be found in~\cite{grootDissipativityFrameworkConstructing2025}.

In this section, we have developed stability theorems for the feedback interconnections in Fig.~\ref{fig:feedback_interconnections}. We note, however, that the methods developed in this paper can be used to analyze \emph{any} interconnection of MIMO systems by using the modular interconnection rules from Proposition~\ref{prop:srg_calculus}, and the tools from this section to analyze the feedback loops. 
}

\section{Computing the MIMO SRG}\label{sec:computing_mimo_srgs}

In order to use Theorem~\ref{thm:srg-lfr-system}, one requires expressions for the SRG of the operators $G$ and $\Phi$ in the LFR. For this purpose, we develop formulas to compute an SRG bound of common MIMO operators: all stable LTI operators, the class of nonlinear operators which are diagonal and static, and some non-square nonlinear operators.

\subsection{LTI Operators}

The first part of a system in LFR form, see~\eqref{eq:lfr_closed_loop}, is an LTI operator. We will show how to exactly compute the SRG of a MIMO LTI operator, as defined in Definition~\ref{def:mimo_srg}. Let $G \in RH_\infty^{q \times p}$ be the transfer function $G(s)$ that corresponds to a stable causal LTI operator $G:L_2^p \to L_2^q$.

\thmspace
\begin{theorem}\label{thm:LTI_SRG_bound}
    Let $G \in RH_\infty^{q \times p}$ and $\Upsilon, \Lambda \subseteq \R$, then
    \begin{equation}\label{eq:LTI_SRG_bound}
        \maketextstyle \SRG(G) \subseteq \Bigl( \bigcap_{\alpha \in \Upsilon} D_{\upsilon_\alpha}(\alpha) \Bigl) \setminus \Bigl( \bigcup_{\alpha \in \Lambda} D_{\ell_\alpha}(\alpha) \Bigl) =: \mathcal{G}^G_{\Upsilon, \Lambda }.
    \end{equation}
    If $\Upsilon, \Lambda = \R$, then~\eqref{eq:LTI_SRG_bound} is an equality.
\end{theorem}

\begin{proof}
    Recall that the $H_\infty$ norm has the property that $\norm{G}_\infty = \sup_{u \in L_2^p} \frac{\norm{G u}}{\norm{u}} = \sup_{\omega \in \R} \overline{\sigma}(G(j \omega)) =: \overline{\sigma}(G)$, hence $\SRG(G) \subseteq D_{\overline{\sigma}(G)}(0)$. Now define $n = \max\{ p,q \}$ and
    \begin{equation}\label{eq:G_alpha_def}
        G_\alpha = \scalebox{0.75}{$ \mat{G \\ 0_{(n-q) \times p}} - \mat{\alpha I \\ 0_{(n-p) \times p}} $},
    \end{equation}
    where $I \in \R^{p \times p}$ is the identity, $\alpha \in \R$ and $G_\alpha : L_2^p \to L_2^n$. Observe that $G_\alpha \pi = \iota G \pi - \scalebox{0.6}{$\mat{\alpha I & 0 \\ 0 & 0}$} = \iota G \pi - \alpha I_{\mathcal{U}_p^n}$, where $\pi = \pi_{p \leftarrow n}$ and $\iota = \iota_{n \leftarrow p}$. Since the last $n-p$ inputs are zero in $\mathcal{U}_p^n$, we know from the SRG definition that
    \begin{equation}\label{eq:srg_subspace_alpha_bound}
        \SRG_{\mathcal{U}_p^n}(G_\alpha \pi_{p \leftarrow n}) \subseteq D_{\upsilon_\alpha}(0), \quad \upsilon_\alpha := \overline{\sigma}(G_\alpha).
    \end{equation}
    Using Proposition~\ref{prop:srg_calculus}.\ref{eq:srg_calculus_alpha} and \ref{prop:srg_calculus}.\ref{eq:srg_calculus_plus_one} on~\eqref{eq:srg_subspace_alpha_bound} we can conclude $\SRG_{\mathcal{U}_p^n}(\iota G \pi ) \subseteq D_{\upsilon_\alpha}(\alpha)$, hence $\SRG(G) \subseteq D_{\upsilon_\alpha}(\alpha)$ by Definition~\ref{def:mimo_srg}. For some set $\Upsilon \subseteq \R$, we have the SRG bound
    \begin{equation}\label{eq:LTI_upper_bound}
        \maketextstyle \SRG(G) \subseteq \bigcap_{\alpha \in \Upsilon} D_{\upsilon_\alpha}(\alpha).
    \end{equation}

    Let $\underline{\sigma}(G) := \inf_{\omega \in \R} \underline{\sigma}(G(j \omega))$ be the smallest singular value of $G(s)$ on the imaginary axis. By Parseval's theorem, one has $\norm{G u} \geq \underline{\sigma}(G) \norm{u}$ for all $u \in L_2^p$, and therefore $\SRG(G) \subseteq \C_\infty \setminus D_{\underline{\sigma}(G)}(0)$. Denote the lower bound $\ell_\alpha:=\underline{\sigma}(G_\alpha)$, where $G_\alpha$ is taken from~\eqref{eq:G_alpha_def}. By the same reasoning used to derive the upper bound in~\eqref{eq:LTI_upper_bound}, we have the following SRG lower bound
    \begin{equation}\label{eq:LTI_lower_bound}
        \maketextstyle \SRG(G) \subseteq \C_\infty \setminus \bigcup_{\alpha \in \Lambda} D_{\ell_\alpha}(\alpha)
    \end{equation}
    where $\Lambda \subseteq \R$. The bound~\eqref{eq:LTI_SRG_bound} follows directly from~\eqref{eq:LTI_upper_bound} and \eqref{eq:LTI_lower_bound}. Exactness for $\Upsilon, \Lambda = \R$ follows from~\cite[proof of Thm. 1]{patesScaledRelativeGraph2021}, where it is shown that the SRG of a linear operator is h-convex, i.e., of the form of~\eqref{eq:LTI_SRG_bound}, and the fact that $\upsilon_\alpha$ and $\ell_\alpha$ are the largest and smallest values, respectively, such that~\eqref{eq:LTI_SRG_bound} holds. The proof in~\cite{patesScaledRelativeGraph2021} is easily adapted for the case where the SRG of a linear operator is computed on a closed linear subspace of the domain, which is used for the tall case $p<q$.
\end{proof}

\extver{}{
The result in Theorem~\ref{thm:LTI_SRG_bound} has an intuitive interpretation in terms of disks in the complex plane. For each $\alpha \in \Upsilon$, using $\overline{\sigma}(G_\alpha)$, we compute the disk $D_{\upsilon_\alpha}(\alpha)$ centered at $\alpha$ that contains $\SRG(G)$. Similarly, for each $\alpha \in \Lambda$, using $\underline{\sigma}(G_\alpha)$ we compute the $D_{\ell_\alpha}(\alpha)$ centered at $\alpha$ which does not contain $\SRG(G)$. In~\eqref{eq:LTI_SRG_bound}, we then intersect all disks that contain $\SRG(G)$, and remove all disks that do not contain $\SRG(G)$. 

If $|\Upsilon| = n_\upsilon, |\Lambda| = n_\ell$, then~\eqref{eq:LTI_SRG_bound} amounts to $n_\upsilon+n_\ell$ amount of $H_\infty$ norm computations. To represent $\mathcal{G}^G_{\Upsilon, \Lambda }$, which consists of the union/intersection of $n_\upsilon + n_\ell$ circles, each requiring a radius and center on $\R$ to be uniquely represented, one needs to store only $2(n_\upsilon + n_\ell)$ real numbers. 
}

By abuse of notation, we will use $\SRG(G) = \mathcal{G}^G_{\Upsilon, \Lambda }$. In practice, one takes $\Upsilon = \Lambda$, i.e., the circles for maximum and minimum gain are computed at the same base points for computational efficiency.\extver{}{ Code for computing the bound in~\eqref{eq:LTI_SRG_bound} is available at \href{https://github.com/Krebbekx/SrgTools.jl}{\texttt{github.com/Krebbekx/SrgTools.jl}}.}

The idea of bounding the SRG of linear operators, based on shifting circles on the real axis, has been proposed in~\cite{patesScaledRelativeGraph2021}, where a method for this purpose, restricted to the case when the operator is square, was given via a factorization approach. The same circle bounding approach is used for square LTI systems in~\cite{grootDissipativityFrameworkConstructing2025} in a dissipativity context. Recently, in~\cite{nautaComputableCharacterisationsScaled2025,krebbekxComputingHardScaled2025}, it has been shown that the SRG of square LTI systems is \emph{exactly} of the form~\eqref{eq:LTI_SRG_bound}. The novel aspect of Theorem~\ref{thm:LTI_SRG_bound} is that it computes the SRG of \emph{tall} systems while respecting the original input dimension. All previously mentioned methods would add nonzero artificial inputs to square-up the system, potentially leading to conservatism.

\extver{}{
\begin{example}\label{example:lti_tfs}
    Consider the transfer functions
    \begin{align*}
        G_1(s) &= \mat{\frac{s}{s+1} & \frac{s^2}{s^2+s+1} & \frac{1}{2s+1}}, \\ G_2(s) &= \mat{\frac{s^2}{s^2+s+1} & \frac{1}{2s+1} \\ \frac{s+1}{(s+3)(s^2+s+1)} & \frac{s+3}{s+1} \\ \frac{s^2-1}{(s+3)(s+2)} & \frac{s}{s+2}}.
    \end{align*}
    Their SRGs, computed with Theorem~\ref{thm:LTI_SRG_bound} are shown in Fig.~\ref{fig:example_lti}. 
\end{example}

\begin{figure}[h]
    \centering
    \begin{subfigure}[b]{0.49\linewidth}
        \centering
        \includegraphics[width=\linewidth]{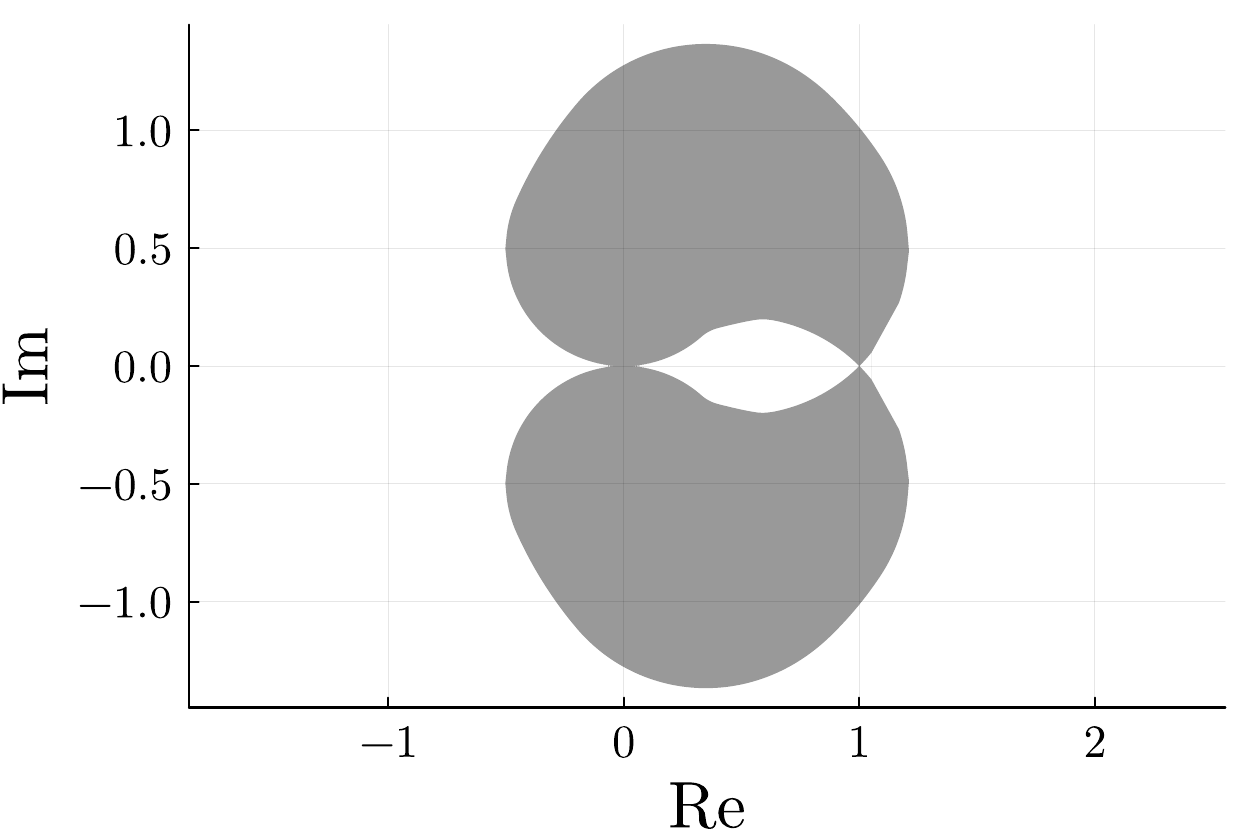}
        \caption{$G_1(s)$.}
        \label{fig:example_lti_1}
    \end{subfigure}
    \hfill
    \begin{subfigure}[b]{0.49\linewidth}
        \centering
        \includegraphics[width=\linewidth]{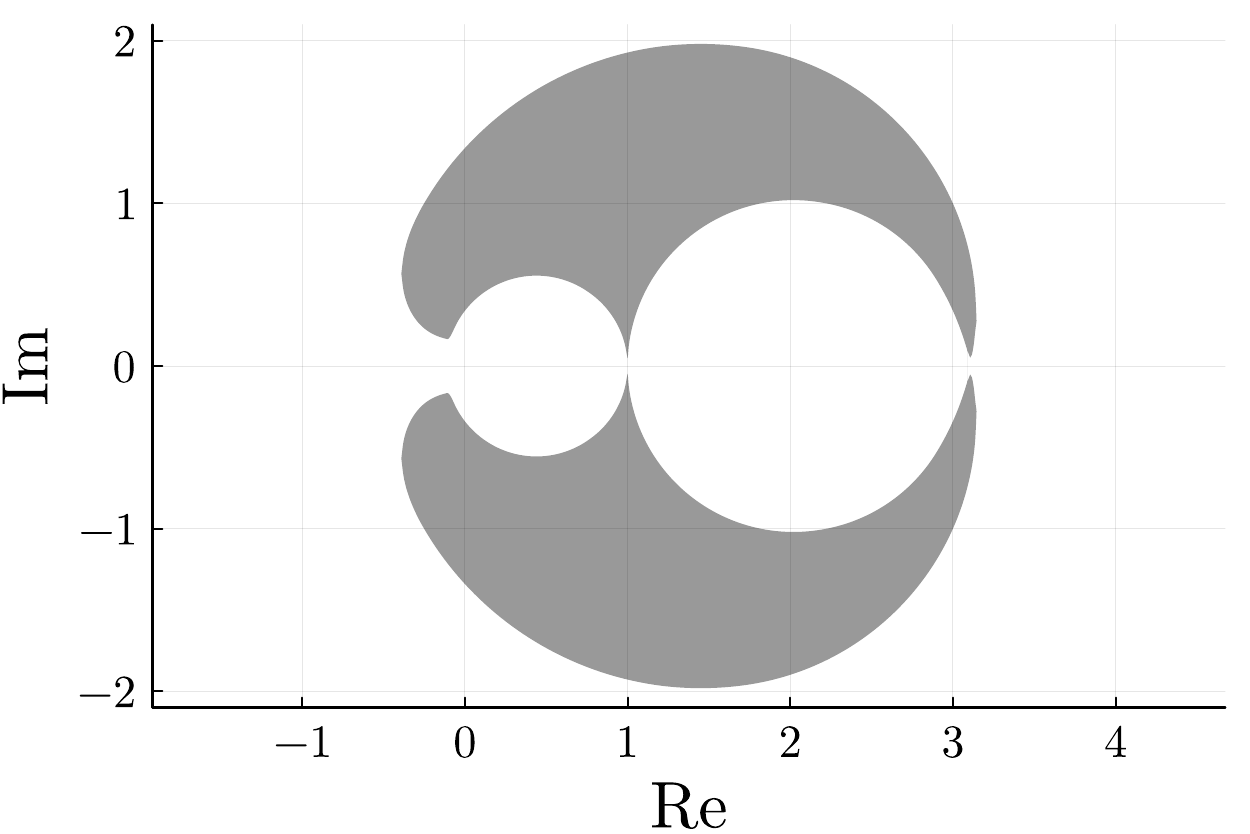}
        \caption{$G_2(s)$.}
        \label{fig:example_lti_2}
    \end{subfigure}
    \caption{SRGs of the MIMO transfer functions in Example~\ref{example:lti_tfs}.}
    \label{fig:example_lti}
\end{figure}
} %

Note that~\eqref{eq:G_alpha_def} assumes that extra zero inputs and outputs are padded at the end. If a different permutation of inputs and outputs is desired,~\eqref{eq:G_alpha_def} and hence Theorem~\ref{thm:LTI_SRG_bound} must be changed accordingly.

Theorem~\ref{thm:LTI_SRG_bound} also provides an effective algorithm to compute the SRG of a \emph{matrix}, which can be viewed as a static transfer matrix. This approach can be used to generalize~\cite{baron-pradaDecentralizedStabilityConditions2025} to the case of non-square LTI systems.

\subsection{Static Nonlinear Operators}

The second component of a system in LFR form is a nonlinear operator. Consider a square operator $\Phi: L_2^d \to L_2^d$ that is defined by the diagonal nonlinear map
\begin{equation}\label{eq:phi_diag_nl}
    (x_1,\dots,x_d)^\top \mapsto (\phi_1 (x_1), \dots, \phi_d (x_d))^\top,
\end{equation}
where $\phi_i : \R \to \R$ is an incrementally sector bounded static nonlinear map defined by  
\begin{equation*}
    \mu_i |x-y|^2 \leq (x-y)(\phi_i(x)-\phi_i(y)) \leq \lambda_i |x-y|^2, \; \forall x,y\in \R,
\end{equation*}
where $\mu_i,\lambda_i \in \R$, which is denoted as $\partial \phi_i \in [\mu_i, \lambda_i]$. Similarly, non-incremental sector bounded operators, denoted as $\phi_i \in [\mu_i, \lambda_i]$, obey $\mu_i |x|^2  \leq x \phi_i(x) \leq \lambda_i |x|^2 , \; \forall x\in \R$. 

The following lemma generalizes~\cite[Prop. 9 and 10]{chaffeyGraphicalNonlinearSystem2023} from SISO static nonlinearities to diagonal static nonlinearities.

\thmspace
\begin{lemma}\label{lemma:diagonal-static-nl}
    Let $\Phi:L_2^d \to L_2^d$ be defined by~\eqref{eq:phi_diag_nl}, then 
    \begin{align*}
        \partial \phi_i \in  [\mu_i,\lambda_i], \forall i=1,\dots, d &\implies \SRG(\Phi) \subseteq D_{[\mu, \lambda]}, \\
        \phi_i \in  [\mu_i,\lambda_i], \forall i=1,\dots, d &\implies \SG_0(\Phi) \subseteq D_{[\mu, \lambda]}, %
    \end{align*}
    where $\mu = \min_i \mu_i$ and $\lambda = \max_i \lambda_i$.
\end{lemma}
\thmspace

\begin{proof}%
    Let $u_1,u_2 \in \R^d$ and define $\Delta u = u_1-u_2$, $\Delta y = y_1-y_2$ where $y_{1,i} = \phi_i(u_{1,i}) - \mu u_{1,i}$ and $y_{2,i} = \phi_i(u_{2,i}) - \mu u_{2,i}$. By \cite[Proposition 9]{chaffeyGraphicalNonlinearSystem2023}, we have 
    \begin{equation}\label{eq:output_strict_i}
        \maketextstyle \Delta u_i \Delta y_i \geq \frac{1}{\lambda - \mu} \Delta y_i^2, \quad \forall i=1,\dots,d.
    \end{equation}
    By summing~\eqref{eq:output_strict_i} and integrating over $\R_{\geq 0}$ one obtains
    \begin{equation}\label{eq:output-strict}
        \maketextstyle \inner{x-y}{\Psi(x)-\Psi(y)} \geq \frac{1}{\lambda-\mu} \norm{\Psi(x)-\Psi(y)}^2,
    \end{equation}
    where $x,y\in L_2^d$, we substituted $u_1 = x(t), u_2=y(t)$, and $\Psi := \Phi - \mu I$. By \cite[Proposition 2]{ryuScaledRelativeGraphs2022}, we conclude $\SRG(\Psi) \subseteq D_{\frac{\lambda-\mu}{2}}(\frac{\lambda-\mu}{2}) = D_{[\mu, \lambda]}$, and hence by Proposition~\ref{prop:srg_calculus}.\ref{eq:srg_calculus_alpha} and \ref{prop:srg_calculus}.\ref{eq:srg_calculus_plus_one} we conclude $\SRG(\Psi) \subseteq D_{[\mu, \lambda]}$. Similarly, $\SG_0(\Phi) \subseteq D_{[\mu, \lambda]}$ follows from $y=0$ above.
\end{proof}

\begin{remark}\label{remark:diagonal-nl-loop-transformation}
    Conservatism of SRG calculations can be reduced by applying loop transformations such that $\mu=\mu_i, \lambda=\lambda_i$ for all $i=1,\dots, d$. 
\end{remark}

Next, we compute the MIMO SRG and SG of a class of tall and wide nonlinear operators, respectively. 

\thmspace
\begin{proposition}\label{prop:tall_nl_srg}
    For $a \in \R$, consider the tall nonlinearity $\Psi_a(x) = \mat{ax + \psi_1(x)\\ \psi_2(x)}$,
    where $\psi_1,\psi_2 : L_2 \to L_2$ obey $\sqrt{\Gamma(\psi_1)^2 + \Gamma(\psi_2)^2} =: \gamma < \infty$,
    then the SRG in Definition~\ref{def:mimo_srg} obeys $\SRG(\Psi_a) \subseteq D_{[a-\gamma,a+\gamma]}$.
\end{proposition}

\begin{proof}
    Note that for $a=0$, $\norm{\Psi_0(x_1)-\Psi_0(x_2)}^2=\scalebox{0.6}{$\norm{\mat{ \psi_1(x_1)- \psi_1(x_2)\\ \psi_2(x_1)-\psi_2(x_2)}}^2$} \leq \Gamma(\psi_1)^2 + \Gamma(\psi_2)^2 \norm{x_1-x_2}^2$, hence $\SRG(\Psi_0) \subseteq D_{[-\gamma, \gamma]}$ by \cite[Prop. 1]{chaffeyGraphicalNonlinearSystem2023}. Since $(x, y) \mapsto (ax ,0)$ acts as $a$ times identity on $\iota_{2 \leftarrow 1}(L_2)$, we can conclude via Proposition~\ref{prop:srg_calculus} that $\SRG(\Psi_a) \subseteq a + D_{[-\gamma, \gamma]}$.
\end{proof} 

\thmspace
\begin{proposition}\label{prop:wide_nl_sg}
    For $a>0$, consider the nonlinearity
    \begin{equation*}
        \maketextstyle \Phi_a(x, y) = \frac{\mathrm{sat}(x)}{a+ y^2}, \quad \mathrm{sat}(x) = 
    \begin{cases}
        x  &\text{ if } |x| \leq 1, \\ 
        x/|x|  &\text{ else}, 
    \end{cases} 
    \end{equation*}
    then $\SG_0(\Phi_a) \subseteq D_{[0,1/a]}$.
\end{proposition}

\begin{proof}
    When squaring $\Phi_a$ up, we get the map $w = (x,y) \mapsto (\frac{\mathrm{sat}(x)}{a+ y^2},0)= z$, for which the inner product on $L_2^2$ obeys $ \langle (x,y), (\frac{\mathrm{sat}(x)}{a+ y^2},0) \rangle = \int_0^\infty \frac{x(t) \mathrm{sat}(x(t))}{a+y(t)^2} \dd t \geq a \int_0^\infty \frac{\mathrm{sat}(x(t))^2}{(a+y(t)^2)^2} \dd t =  \langle (\frac{\mathrm{sat}(x)}{a+ y^2},0),(\frac{\mathrm{sat}(x)}{a+ y^2},0) \rangle$, i.e., $\inner{w}{z} \geq a \norm{z}^2$, which implies~\cite[Prop. 1]{chaffeyGraphicalNonlinearSystem2023} $\SG_0(\Phi_a) \subseteq D_{[0,1/a]}$.
\end{proof}

Even though these examples are two-dimensional, the technique to compute the S(R)G can be extended to higher input/output dimensions.

\section{Examples}\label{sec:examples}

With the following examples\extver{\footnote{For each example, the Julia code used to generate all figures and other results is available at \href{https://github.com/Krebbekx/SrgTools.jl}{\texttt{github.com/Krebbekx/SrgTools.jl}}.}}{} we demonstrate how the system analysis results from Section~\ref{sec:feedback_systems} can be used to analyze the stability and $L_2$-gain performance of example systems. We use the methods from Section~\ref{sec:computing_mimo_srgs} to obtain SRG bounds for the operators that are involved.

\extver{}{For each example, the Julia code used to generate all figures and other results is available at \href{https://github.com/Krebbekx/SrgTools.jl}{\texttt{github.com/Krebbekx/SrgTools.jl}}.}

\subsection{Non-Square Lur'e System}\label{sec:example_nonsquare_lure}

\begin{figure}[tb]
    \centering
    \includegraphics[width=0.7\linewidth]{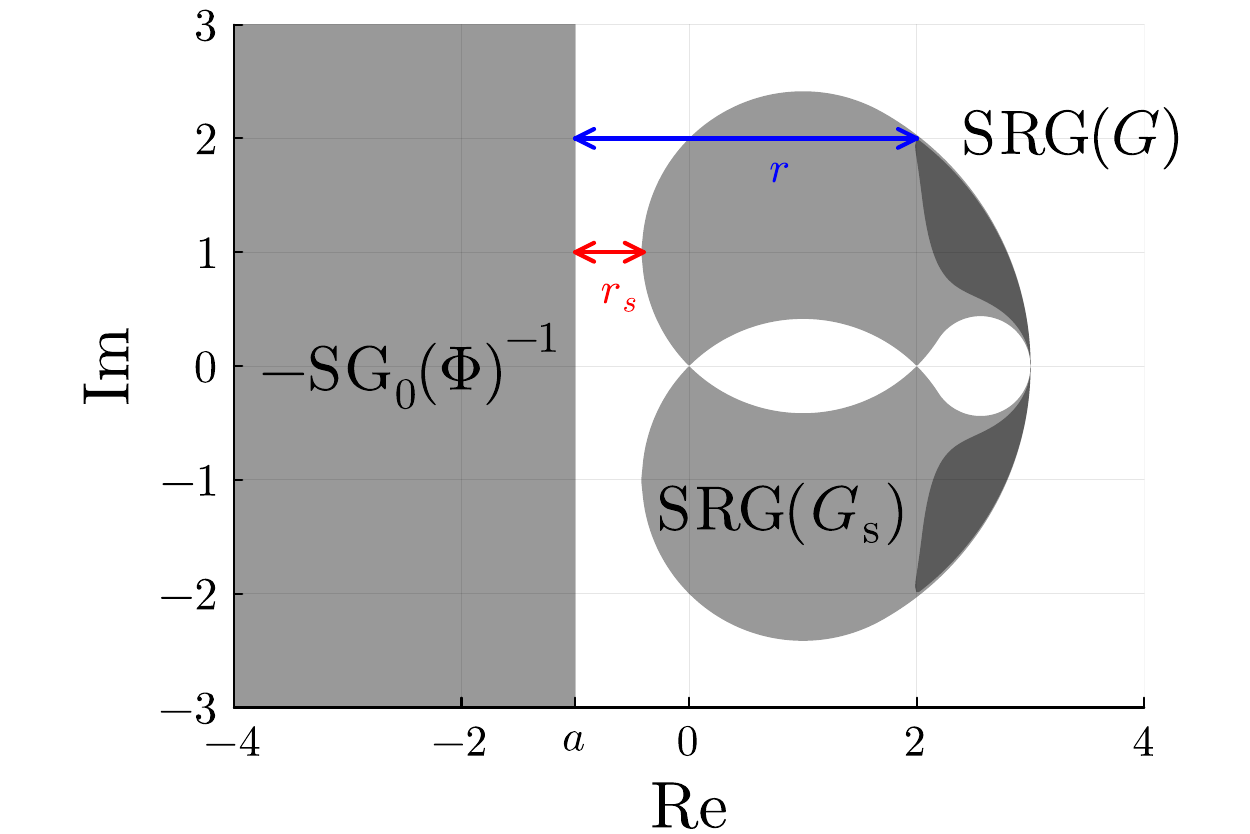}
    \vspace*{-0.8em}
    \caption{SRG analysis of the example in Section~\ref{sec:example_nonsquare_lure}.}
    \label{fig:example_nonsquare_lure}
\end{figure}

Consider the feedback system $R = [\Phi_a, G]$, where
\begin{equation}\label{eq:example_nonsquare}
    G(s) = \mat{ \frac{s+2}{s/3+1} \\ \frac{s+2}{(s+1)^2} }, \quad \Phi_a(x,y) = \frac{\mathrm{sat}(x)}{a+y^2},
\end{equation}
\noindent for $a>0$ and its \enquote{squared} version $R_\mathrm{s} = [\Phi_{a,\mathrm{s}}, G_\mathrm{s}]$ where \extver{$G_\mathrm{s}(s) = \big(G(s) \ 0_{2 \times 1} \big) $ and $ \Phi_{a,\mathrm{s}}(x,y) =\big(\Phi_a(x,y) \ 0 \big)^\top$.}{
\begin{equation}\label{eq:example_nonsquare_squared}
    G_\mathrm{s}(s) = \mat{ \frac{s+2}{s/3+1} & 0 \\ \frac{s+2}{(s+1)^2} & 0}, \quad \Phi_{a,\mathrm{s}}(x,y) = \mat{\frac{\mathrm{sat}(x)}{a+y^2} \\ 0}.
\end{equation}}

Note that by Proposition~\ref{prop:wide_nl_sg}, we know $\SG_0(\Phi_a) = \SG_0(\Phi_{a,\mathrm{s}}) \subseteq D_{[0,1/a]}$. We compute the SRG of both $G$ and $G_\mathrm{s}$ using Theorem~\ref{thm:LTI_SRG_bound}, and plot it together with the SG of the nonlinearity in Fig.~\ref{fig:example_nonsquare_lure}. We emphasize that $\SRG(G) = \SRG_{\iota_{2 \leftarrow 1}(L_2)} (G_\mathrm{s})$, whereas $\SRG(G_\mathrm{s}) = \SRG_{L_2^2}(G_\mathrm{s})$, i.e., $\SRG(G)$ is the SRG of $G_\mathrm{s}$ restricted to the original inputs of $G$, whereas $\SRG(G_\mathrm{s})$ includes the artificial inputs from the squaring up procedure as in~\cite{misraComputationalAlgorithmSquaringup1992}. 

Using this example, we can demonstrate two clear advantages of Definition~\ref{def:mimo_srg} over the naive squaring up method (i.e.,~\cite{misraComputationalAlgorithmSquaringup1992} with nonzero inputs and outputs). 

First, from Fig.~\ref{fig:example_nonsquare_lure} and \extver{the non-incremental version of Theorem~\ref{thm:incremental-mimo-srg}}{Theorem~\ref{thm:non-incremental-mimo-srg}} that $\gamma(R) \leq  r^{-1} = (a+2)^{-1}$ and $\gamma(R_\mathrm{s}) \leq r_\mathrm{s}^{-1} = (a-0.42)^{-1}$. At $a=0.5$, for example, this would amount to $\gamma(R) \leq 0.4$, whereas $\gamma(R_\mathrm{s}) \leq 12.5$, which demonstrates that using Definition~\ref{def:mimo_srg} can greatly reduce conservatism for $L_2$-gain computations. Moreover, when $a < 0.42$ is taken, SRG analysis using $\SRG(G_\mathrm{s})$ cannot even conclude stability, while the system is proven to be stable for all $a>0$ using $\SRG(G)$. 

Second, it is apparent from Fig.~\ref{fig:example_nonsquare_lure} that $G$ is a passive LTI system (since $\mathrm{Re}(\SRG(G)) \geq 0$~\cite{chaffeyGraphicalNonlinearSystem2023}). Including the artificial inputs in $G_\mathrm{s}$ clearly destroys the passivity property, showing the benefit of Definition~\ref{def:mimo_srg}.

\subsection{SISO System with Multiple Nonlinearities}\label{sec:example_1_nl_network}

\begin{figure}[tb]
    \centering
    \tikzstyle{block} = [draw, rectangle, 
    minimum height=2em, minimum width=2em]
    \tikzstyle{sum} = [draw, circle, scale=0.7, node distance={0.5cm and 0.5cm}]
    \tikzstyle{input} = [coordinate]
    \tikzstyle{output} = [coordinate]
    \tikzstyle{pinstyle} = [pin edge={to-,thin,black}]
    
    \begin{tikzpicture}[auto, node distance = {0.1cm and 0.5cm}]
        \node [input, name=input] {};
        \node [sum, right = of input] (sum) {$\Sigma$};
        \node [block, right = of sum] (controller) {$K(s)$};
        \node [block, right = of controller] (saturation) {$\phi_1$};
        \node [sum, right = of saturation] (sigma) {$\Sigma$};
        \node [block, right = of sigma] (lti) {$P(s)$};
        \node [coordinate, right = of lti] (z_intersection) {};
        \node [output, right = of z_intersection] (output) {}; %
        \node [block, below = of lti] (static_nl) {$\phi_2$};
        \node [coordinate, right = of static_nl] (phi_intersection) {};

        \draw [->] (input) -- node {$u$} (sum);
        \draw [->] (sum) -- node {$e$} (controller);
        \draw [->] (controller) -- node {$u$} (saturation);
        \draw [->] (saturation) -- node {$\hat{u}$} (sigma);
        \draw [->] (sigma) -- node {$u'$} (lti);
        \draw [->] (lti) -- node [name=z] {$y$} (output);
        \draw [->] (z) |- (static_nl);
        \draw [->] (static_nl) -| node[pos=0.99] {$-$} (sigma);
        \node [coordinate, below = of static_nl] (tmp1) {$H(s)$};
        \draw [->] (z) |- (tmp1)-| node[pos=0.99] {$-$} (sum);
    
    \end{tikzpicture}
    \caption{Block diagram of a controlled Lur'e plant.}
    \label{fig:controlled_lure_sat}
\end{figure}

\extver{}{
\begin{figure*}[t]
    \centering
     \begin{subfigure}[b]{0.24\linewidth}
         \centering
         \includegraphics[width=\linewidth]{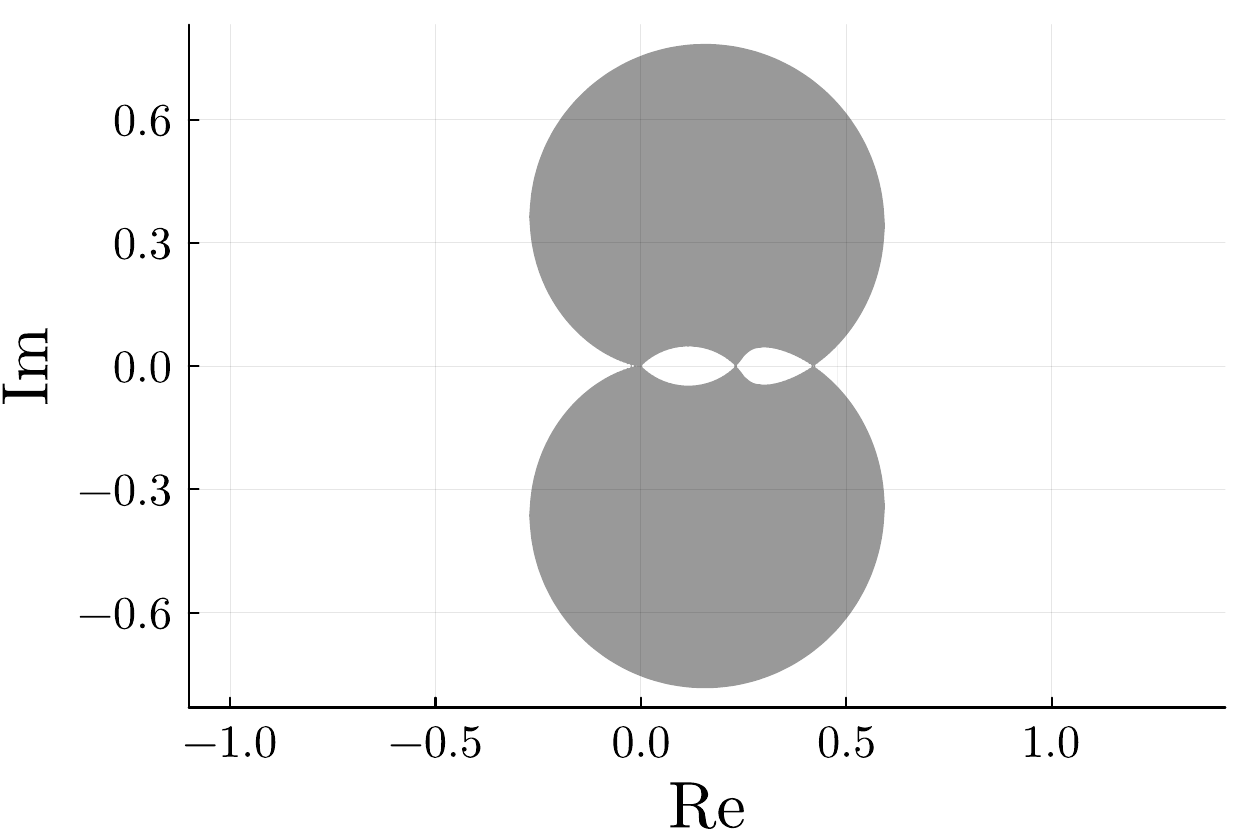}
         \caption{$\SRG(G_\mathrm{yw})$}
         \label{fig:example_1_k23_srgs_Gyw}
     \end{subfigure}
     \hfill
     \begin{subfigure}[b]{0.24\linewidth}
         \centering
         \includegraphics[width=\linewidth]{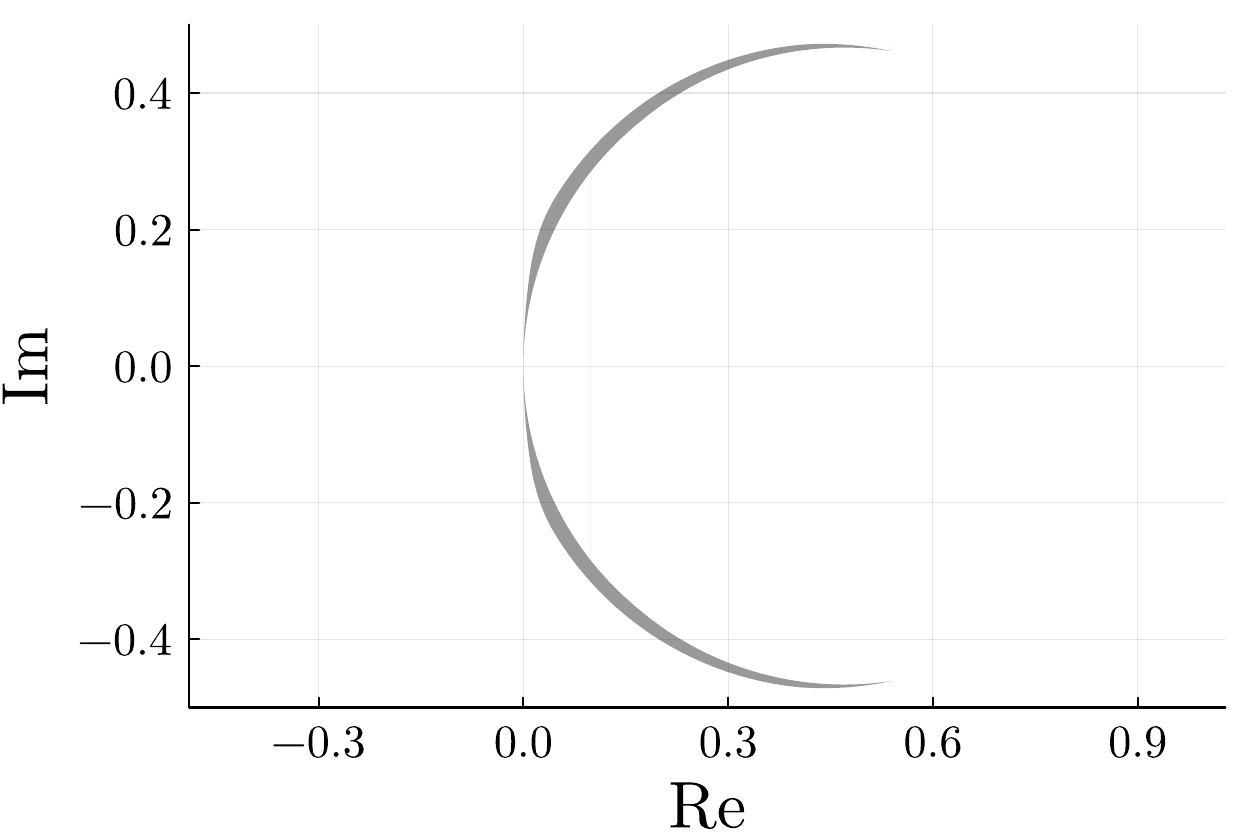}
         \caption{$\SRG(G_\mathrm{zu})$}
         \label{fig:example_1_k23_srgs_Gzu}
     \end{subfigure}
     \hfill
     \begin{subfigure}[b]{0.24\linewidth}
         \centering
         \includegraphics[width=\linewidth]{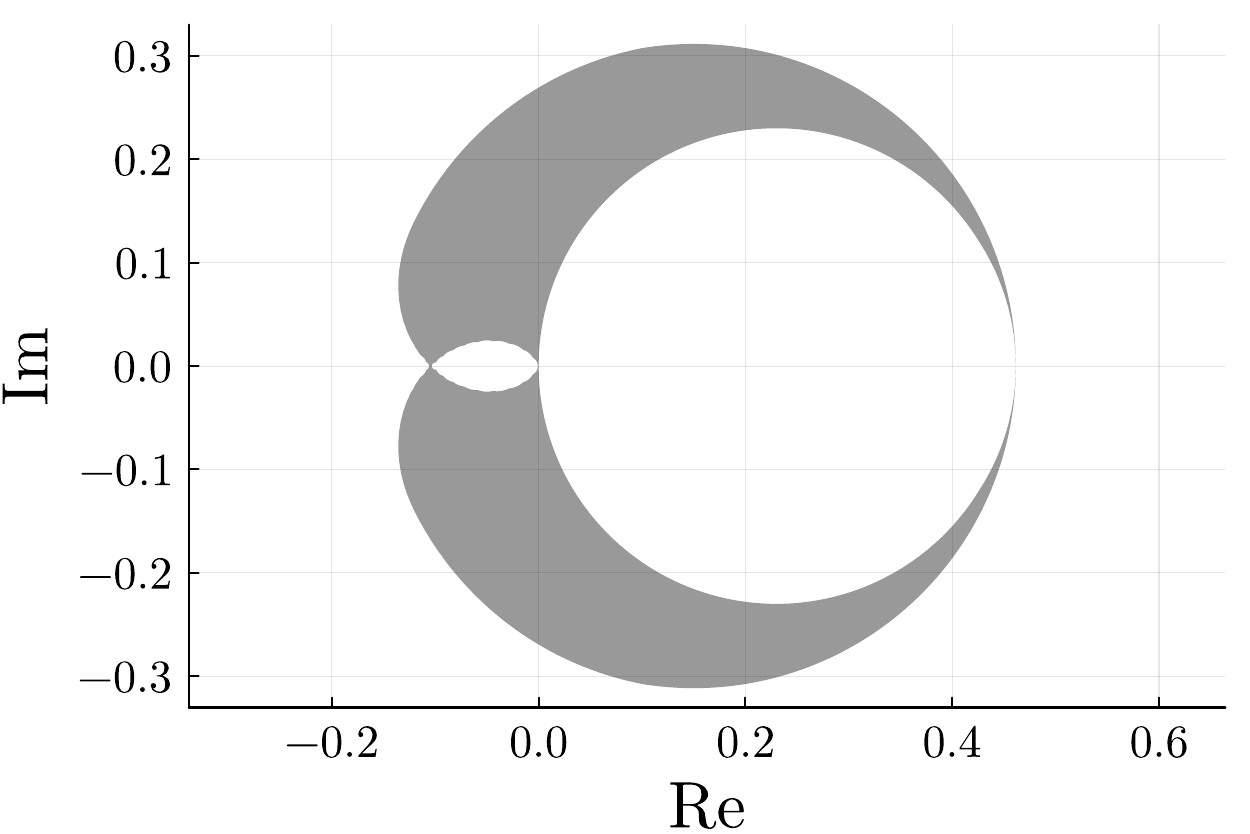}
         \caption{$\SRG(G_\mathrm{yu})$}
         \label{fig:example_1_k23_srgs_Gyu}
     \end{subfigure}
     \hfill
     \begin{subfigure}[b]{0.24\linewidth}
         \centering
         \includegraphics[width=\linewidth]{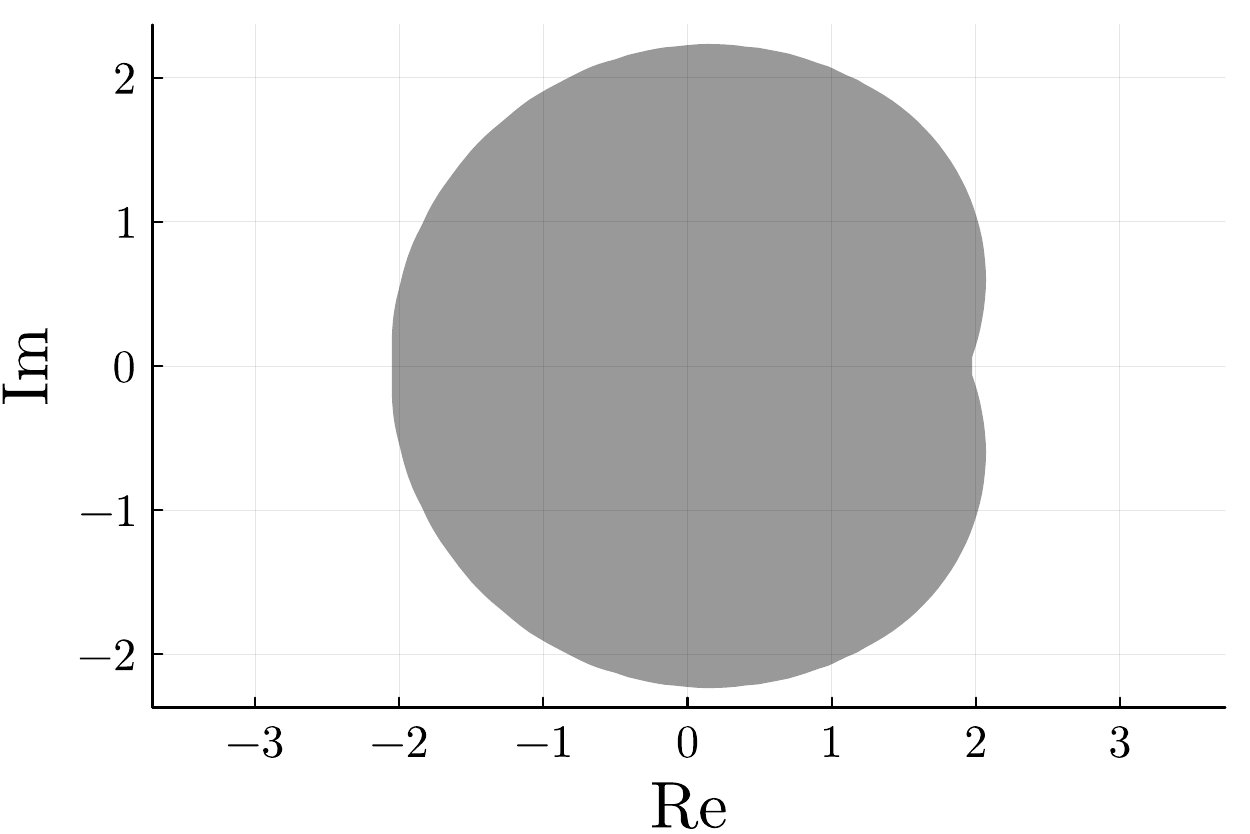}
         \caption{$\mathcal{G}_R$}
         \label{fig:example_1_k23_srgs_full}
     \end{subfigure}
    \caption{SRG analysis of the example in Section~\ref{sec:example_1_nl_network}.}
    \label{fig:example_1_k23_srgs}
\end{figure*}
} %

\begin{figure}[tb]
    \centering
    \begin{subfigure}[b]{0.49\linewidth}
        \centering
        \includegraphics[width=\linewidth]{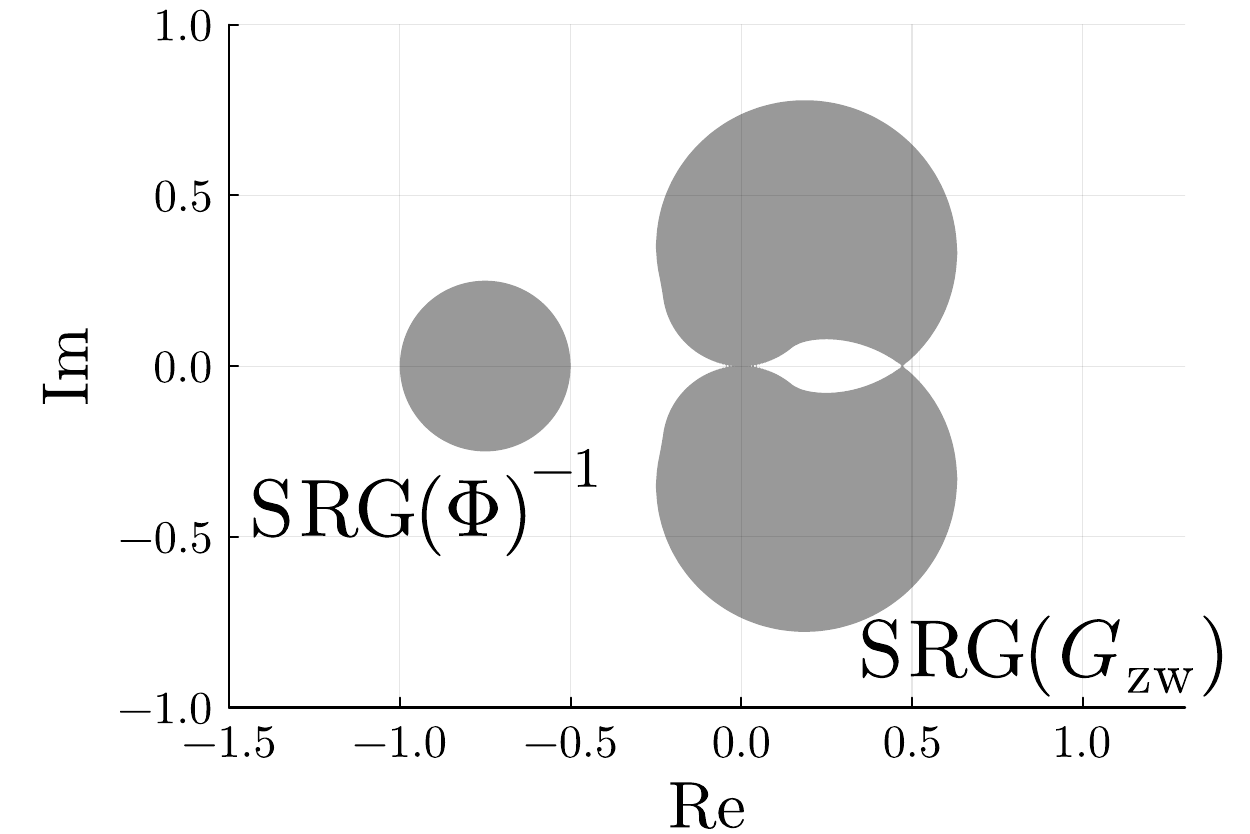}
        \caption{SRGs of $\Phi$ and $G_\mathrm{zw}$.}
        \label{fig:example_1_denom}
    \end{subfigure}
    \hfill
    \begin{subfigure}[b]{0.49\linewidth}
        \centering
        \includegraphics[width=\linewidth]{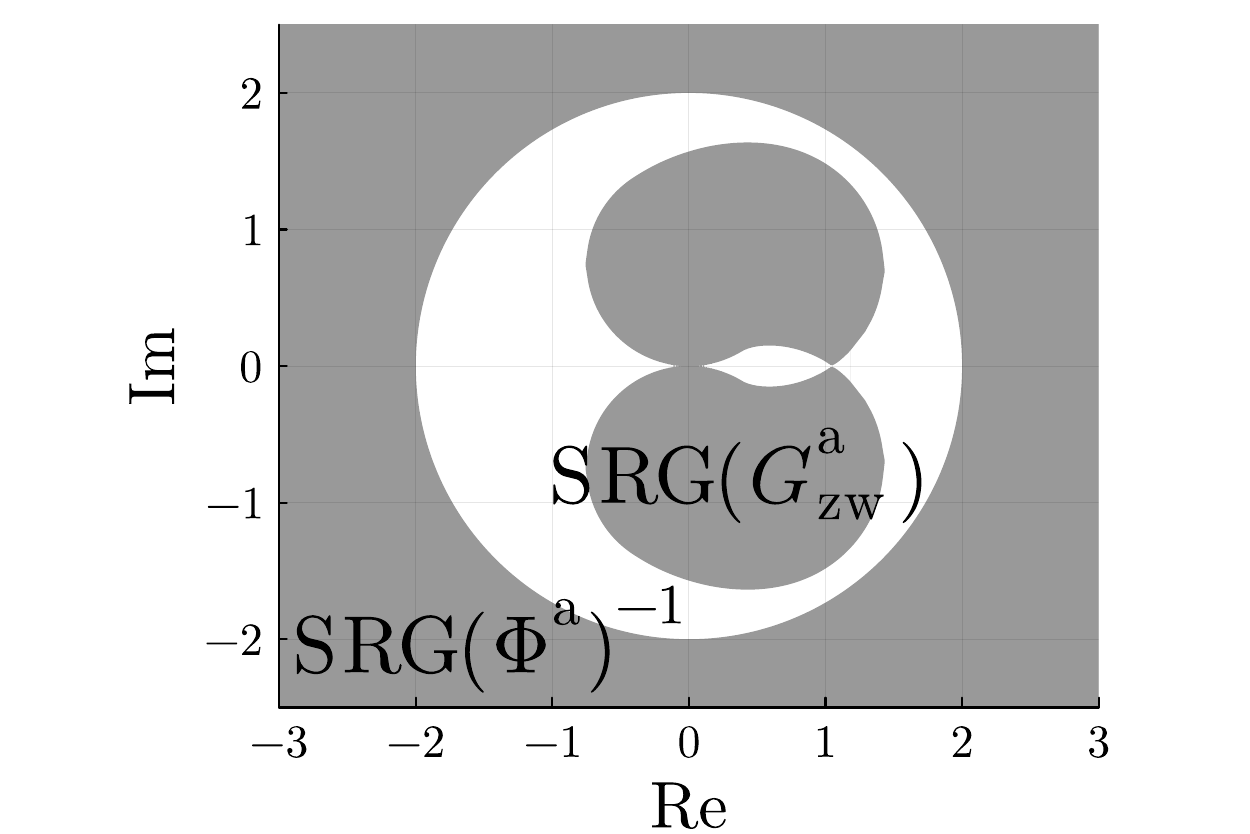}
        \caption{SRGs of $\Phi^a$ and $G_\mathrm{zw}^a$.}
        \label{fig:example_1_denom_a}
    \end{subfigure}
    \vspace*{-0.5em}
    \caption{SRG analysis of the example in Section~\ref{sec:example_1_nl_network}.}
    \label{fig:example_1_denoms}
\end{figure}

Consider the system in Fig.~\ref{fig:controlled_lure_sat}, where $P, K, \phi_1, \phi_2 : \Lte \to \Lte$ are causal SISO systems defined as
\begin{equation*}
\begin{split}
    K(s) = \frac{1}{s+1}, \quad P(s) = \frac{3}{(s-2)(s/10+1)}, \\
    \phi_1(x) = 
    \begin{cases}
        x  \text{ if } |x| \leq 1, \\ 
        x/|x|  \text{ else}, 
    \end{cases} 
    \phi_2(x) = 
    \begin{cases}
        x \text{ if } |x| \leq 1, \\ 
        2x - x/|x|  \text{ else}.
    \end{cases}
\end{split}
\end{equation*}
which is also studied as an example in~\cite{krebbekxScaledRelativeGraph2025}. Define the loop transformations $\varphi_1 := \phi_1 - \kappa_1$ and $\varphi_2 := \phi_2 - \kappa_2$, where $\kappa_1,\kappa_2 \in \R$. To write this system in LFR form $y = R u$ according to~\eqref{eq:lfr_closed_loop}, the nonlinearity becomes $\Phi : \Lte^2 \to \Lte^2$ defined by $(x, y) \mapsto (\varphi_1(x), \varphi_2(y))$. The LTI part in~\eqref{eq:lfr_lti_part} is given by 
\begin{equation}\label{eq:controlled_lure_lfr_lti}
\begin{alignedat}{2}
    G_\mathrm{zw} &= \mat{ -S \tilde{P} K & -S \tilde{P} K \\ S \tilde{P} & S \tilde{P}}, \quad && G_\mathrm{zu} = \mat{ S K \\ S L  } , \\ 
    G_\mathrm{yw} &= \mat{ S \tilde{P} & S \tilde{P}}, && G_\mathrm{yu} = S L ,
\end{alignedat}   
\end{equation}
where $\tilde{P} = \frac{P}{1 + \kappa_2 P}, L = \kappa_1 \tilde{P} K$ and $S = \frac{1}{1+ L}$. Note that $\partial \phi_1 \in [0,1]$ and $\partial \phi_2 \in [1, 2]$, therefore $\partial \varphi_1 \in [-\kappa_1, 1-\kappa_1]$ and $\partial \varphi_2 \in [1-\kappa_2, 2-\kappa_2]$. Then, according to Lemma~\ref{lemma:diagonal-static-nl}, we have the SRG bound
\begin{equation*}
    \SRG(\Phi) \subseteq D_{[\min\{-\kappa_1, 1-\kappa_2 \}, \max\{1-\kappa_1, 2-\kappa_2 \}]}.
\end{equation*}
Before we can apply Theorem~\ref{thm:srg-lfr-system}, we must make sure that $G$, defined by~\eqref{eq:lfr_lti_part} and the transfer functions in~\eqref{eq:controlled_lure_lfr_lti}, is stable. This is achieved by picking loop transformation variables $\kappa_1, \kappa_2$ such that $G$ is stable. 

We fix $\kappa_1 = 2, \kappa_2 = 3$, for which all poles $p$ of $G$ satisfy $\mathrm{Re}(p) <0$. As an alternative choice, denoted by a superscript \enquote{$\mathrm{a}$}, we consider $\kappa_1^\mathrm{a} = 0.5, \kappa_2^\mathrm{a} = 1.5$, for which $G^\mathrm{a}$ is stable as well. Note that these choices result in $\partial \varphi_1, \partial \varphi_2 \in [-2, -1]$ and $\partial \varphi_1^\mathrm{a}, \partial \varphi_2^\mathrm{a} \in [-0.5, 0.5]$ for the loop-transformed nonlinearities. We will now apply Theorem~\ref{thm:srg-lfr-system} for both choices.

The first step is to compute $\SRG(G_\mathrm{zw})$ using Theorem~\ref{thm:LTI_SRG_bound}. Since $\partial \varphi_1, \partial \varphi_2 \in [-2, -1]$, we can conclude that $\SRG(\Phi) \subseteq D_{[-2, -1]}$. The stability of $R$ depends on $[\Phi, -G_\mathrm{zw}]$, is equivalent to the requirement that, for all $\tau \in [0,1]$, $\SRG(\Phi)^{-1}$ and $\tau \SRG(G_\mathrm{zw})$ do not overlap. These graphs are plotted in Fig.~\ref{fig:example_1_denom} and are indeed separated for all $\tau \in [0,1]$. Analogously, for the alternative choice of loop transformation variables, the graphs $\SRG(\Phi^\mathrm{a})^{-1}$ and $\tau \SRG(G_\mathrm{zw}^\mathrm{a})$ are visualized in Fig.~\ref{fig:example_1_denom_a}. From Fig.~\ref{fig:example_1_denoms} it is clear that the smallest separation is attained at $\tau=1$. \extver{
    For the plot of the SRGs of $G_\mathrm{yw}, G_\mathrm{zu}$ and $G_\mathrm{yu}$, computed using Theorem~\ref{thm:LTI_SRG_bound}, see the extended version~\cite{krebbekxMIMOPaperExtended2026}. The final step is to evaluate $\mathcal{G}_R$ in~\eqref{eq:lfr_srg_bound} using the SRG interconnection rules in Proposition~\ref{prop:srg_calculus}. In each step, we applied the improved chord and arc completions from Lemma~\ref{lemma:improved_chord_arc_completions}, yielding the set $\mathcal{G}_R$ with $\rmin(\mathcal{G}_R) \leq 2.33$, therefore we can conclude that $R : \Lte \to \Lte$ is a well-posed and causal system which satisfies $\Gamma(R) \leq 2.33$. 
}{
    The SRGs of $G_\mathrm{yw}, G_\mathrm{zu}$ and $G_\mathrm{yu}$ are also computed using Theorem~\ref{thm:LTI_SRG_bound}, and visualized in Figs.~\ref{fig:example_1_k23_srgs_Gyw}, \ref{fig:example_1_k23_srgs_Gzu} and \ref{fig:example_1_k23_srgs_Gyu}, respectively. The final step is to evaluate $\mathcal{G}_R$ in~\eqref{eq:lfr_srg_bound} using the SRG interconnection rules in Proposition~\ref{prop:srg_calculus}. In each step, we apply the improved chord and arc completions from Lemma~\ref{lemma:improved_chord_arc_completions}, yielding the set $\mathcal{G}_R$ in Fig.~\ref{fig:example_1_k23_srgs_full}. Since $\rmin(\mathcal{G}_R) \leq 2.33$, we can conclude that $R : \Lte \to \Lte$ is a well-posed and causal system which satisfies $\Gamma(R) \leq 2.33$. 
}

A similar computation for the alternative loop transformation variables yields $\rmin(\mathcal{G}_R^\mathrm{a}) \leq 6.13$. This shows that the outcome of Theorem~\ref{thm:srg-lfr-system} can be optimized over all choices of loop transformations that stabilize $G$. 

This example is also studied in~\cite{krebbekxScaledRelativeGraph2025}, where a bound \mbox{$\Gamma(R) \leq 4.81$} is obtained by using SISO SRG tools only. Hence, we see that our MIMO approach yields a tighter incremental $L_2$-gain bound. Moreover, we obtain causality and well-posedness of $R$ via Theorem~\ref{thm:srg-lfr-system}, which are both properties that had to be \emph{assumed} in~\cite{krebbekxScaledRelativeGraph2025}. These improvements with respect to~\cite{krebbekxScaledRelativeGraph2025} underline the strengths and value of the proposed approach.

\extver{}{

\subsection{Two Mass-Spring-Damper System}\label{sec:example_2_msd}

We consider a system of two masses $m_1$ and $m_2$ that are connected to the solid ground and each other with a linear spring and damper, and a nonlinear spring, as depicted in Fig.~\ref{fig:nl_msd_setup}. We take external forces $u_1,u_2$ on the masses $m_1,m_2$ as inputs and as outputs the positions $x_1,x_2$ of the masses. The system is governed by
\begin{equation}\label{eq:msd_ode}
\begin{aligned}
    m_1 \ddot{x}_1 &= u_1 -k_1 x_1 - d_1 \dot{x}_1 - k_{12}(x_1-x_2) - d_{12}(\dot{x}_1 - \dot{x}_2) \\ & + \phi_1(x_1) + \phi_{12}(x_1-x_2), \\ 
    m_2 \ddot{x}_2 &= u_2 -k_2 x_2 - d_2 \dot{x}_2 + k_{12}(x_1-x_2) + d_{12}(\dot{x}_1 - \dot{x}_2) \\ & + \phi_2(x_2) - \phi_{12}(x_1-x_2),
\end{aligned}
\end{equation}
where $k_1,k_2$ and $d_1,d_2$ are the linear spring and damper coefficients, respectively. We choose parameters $m_1=0.5, m_2=3, k_1=1, k_2=2, d_1=0.3, d_2=1, d_{12}=1, k_{12}=0.5$ and nonlinear springs $\phi_1(x)=\phi_2(x)=\phi(x) := -\tanh(x)$ and $\phi_{12}(x) = 2 \tanh(x)-x$. Hence, $\phi$ is a saturating spring and $\phi_{12}$ is a negative spring for small deflection, and a regular spring for large deflection.

\begin{figure}[tb]
    \centering
    \includegraphics[width=0.65\linewidth]{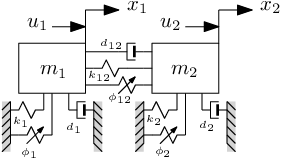}
    \caption{The nonlinear mass-spring-damper setup.}
    \label{fig:nl_msd_setup}
\end{figure}

To write~\eqref{eq:msd_ode} in LFR form~\eqref{eq:lfr_closed_loop}, we take the nonlinear function $\Phi(x, y, z) =(\phi(x),\phi(y), \phi_{12}(z))^\top$, and signals $u = (u_1,u_2)^\top, w = (\phi(x_1), \phi(x_2), \phi_{12}(x_1-x_2))^\top, y = (x_1, x_2)^\top$ and $z = (x_1,x_2,x_1-x_2)^\top$. The transfer function $G$ in~\eqref{eq:lfr_lti_part} is obtained from the state-space representation
\begin{multline*}
    G(s) = C(s I  - A)^{-1}B+D = \scalebox{0.77}{$\left[\begin{array}{c|c} A & B  \\ \hline C & D\end{array}\right]$} = \\ 
    \scalebox{0.77}{$
    \left[\begin{array}{cccc|ccccc}
        0&1&0&0 & 0&0&0&0&0 \\ 
        \frac{-k_1-k_{12}}{m_1} & \frac{-d_1-d_{12}}{m_1} & \frac{k_{12}}{m_1} & \frac{d_{12}}{m_1} & \frac{1}{m_1}&0&\frac{1}{m_1}&\frac{1}{m_1}&0 \\ 
        0&0&0&1 & 0&0&0&0&0 \\ 
        \frac{k_{12}}{m_2} & \frac{d_{12}}{m_2} & \frac{-k_2-k_{12}}{m_2} & \frac{-d_2-d_{12}}{m_2} &0&\frac{1}{m_2}&\frac{-1}{m_2}&0&\frac{1}{m_2} \\ \hline
        1&0&0&0 & 0&0&0&0&0\\ 
        0&0&1&0 & 0&0&0&0&0\\ 
        1&0&-1&0 & 0&0&0&0&0\\ 
        1&0&0&0 & 0&0&0&0&0\\ 
        0&0&1&0 & 0&0&0&0&0 \end{array}\right].
        $}
\end{multline*}

Since $\partial \phi \in [-1,0]$ and $\partial \phi_{12} \in [-1,1]$, which are not the same intervals, we want to shift and scale $\phi$ and $\phi_{12}$ such that they satisfy the same sector bound in order to tighten the bound in Lemma~\ref{lemma:diagonal-static-nl}. To that end, we define $\varphi(x) = \phi(x) + \frac{1}{2}x$ and $\varphi_{12}(x) = \frac{1}{2} \phi_{12}(x)$ such that $\partial \varphi \in [-\frac{1}{2},\frac{1}{2}]$ and \mbox{$\partial \varphi_{12} \in [-\frac{1}{2},\frac{1}{2}]$}. To accommodate for this loop transformation in the LFR, we define $\tilde{G}(s)$ by $\tilde{k}_1 = k_1+\frac{1}{2}, \tilde{k}_2 = k_2 + \frac{1}{2}$ and $\tilde{G}_{\mathrm{x} w_3} = 2 G_{\mathrm{x}  w_3}$ where $\mathrm{x}  \in \{ y_1,y_2,z_1,z_2,z_3 \}$, and all other values are the same as for $G(s)$. The nonlinear function in~\eqref{eq:lfr_closed_loop} becomes $\tilde{\Phi}(x, y, z) = (\varphi(x),\varphi(y), \varphi_{12}(z))^\top$, hence $\SRG(\tilde{\Phi}) \subseteq D_{[-\frac{1}{2}, \frac{1}{2}]}$ by Lemma~\ref{lemma:diagonal-static-nl}. 

\begin{figure}[tb]
    \centering
    \begin{subfigure}[b]{0.49\linewidth}
        \centering
        \includegraphics[width=\linewidth]{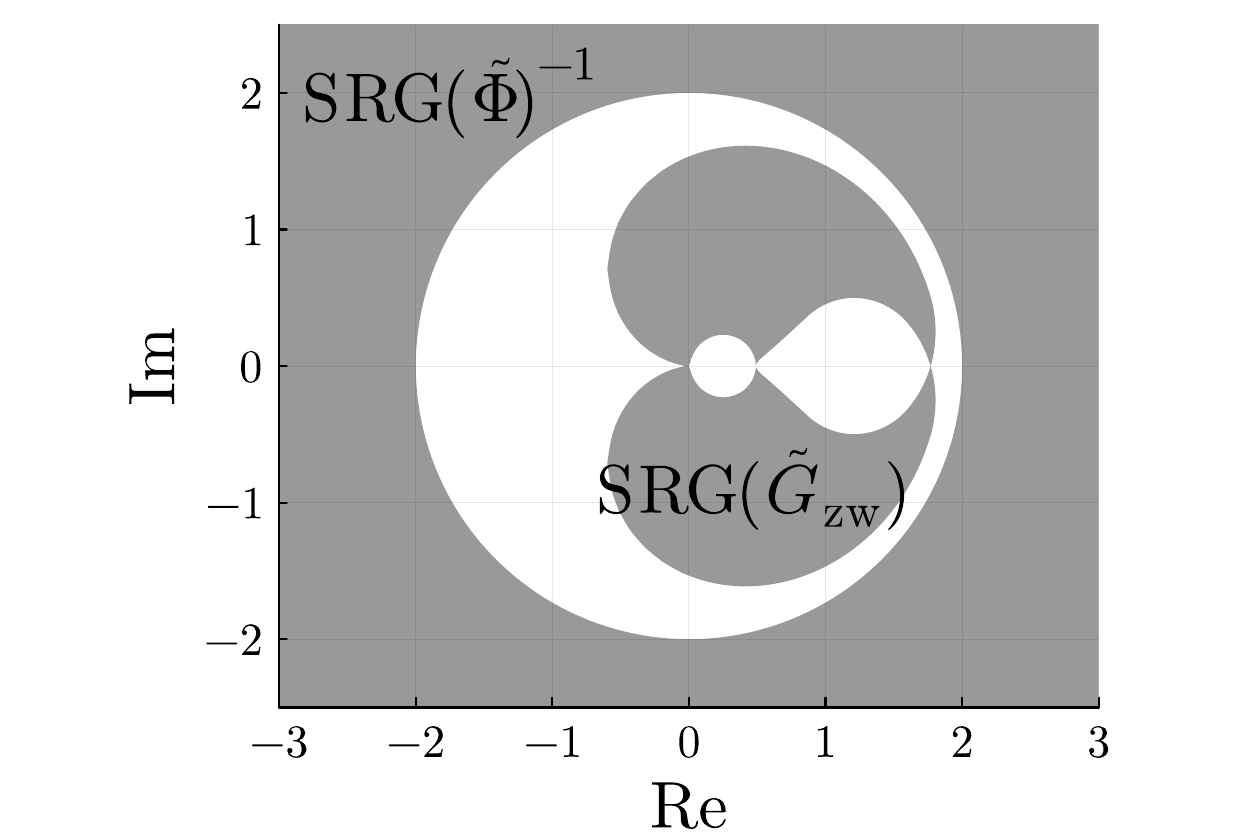}
        \caption{SRGs of $\tilde{\Phi}$ and $\tilde{G}_\mathrm{zw}$.}
        \label{fig:example_2_denom}
    \end{subfigure}
    \hfill
    \begin{subfigure}[b]{0.49\linewidth}
        \centering
        \includegraphics[width=\linewidth]{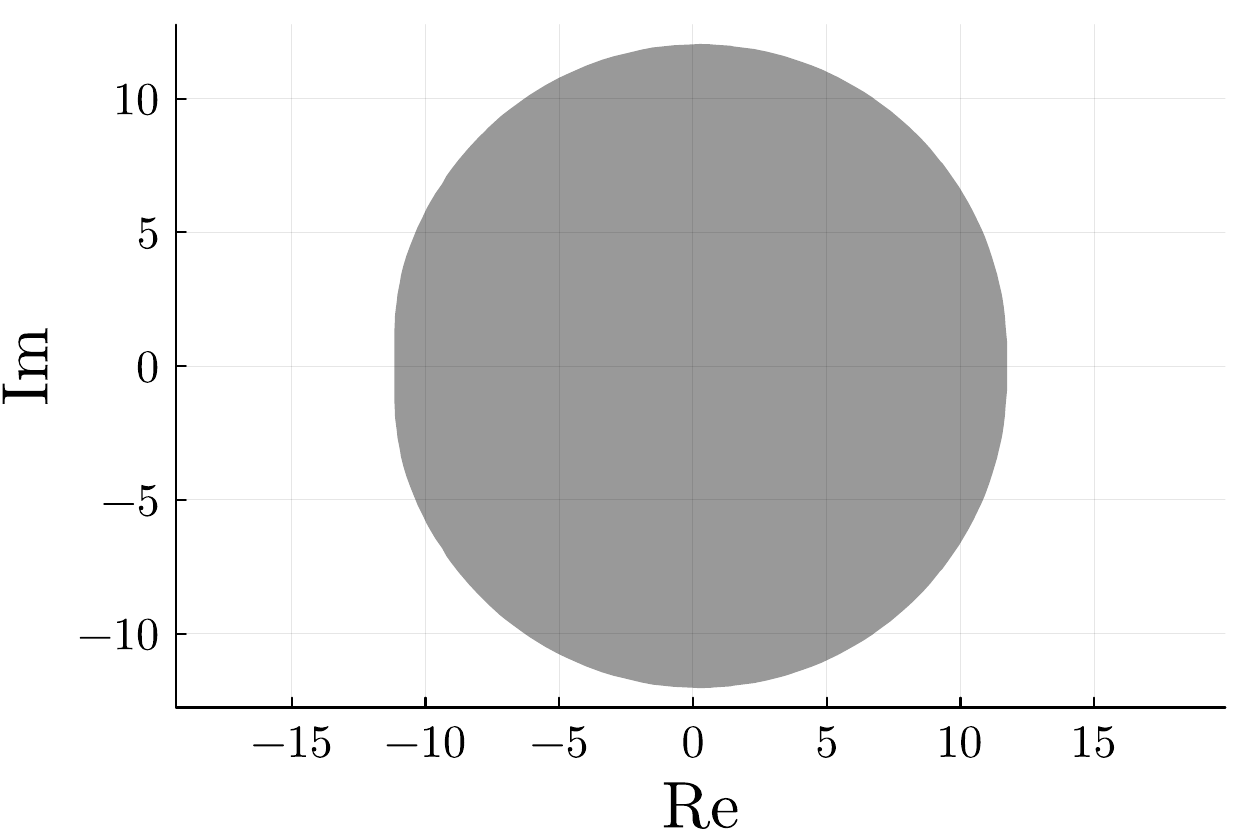}
        \caption{SRG bound $\mathcal{G}_R$ from~\eqref{eq:lfr_srg_bound}.}
        \label{fig:example_2_all}
    \end{subfigure}
\caption{SRG analysis of the example in Section~\ref{sec:example_2_msd}.}
\label{fig:example_2}
\end{figure}

To apply Theorem~\ref{thm:srg-lfr-system} we must first apply the loop transformation to obtain $\tilde{G}(s)$ and check if $\tilde{G}(s) \in RH_\infty^{5 \times 5}$, which is the case. Second, we check stability of $[\tilde{\Phi}, -\tilde{G}_\mathrm{zw}]$ by plotting their SRGs, see Fig~\ref{fig:example_2_denom}, where is clear that $\tau \SRG(\tilde{G}_\mathrm{zw})$ and $\SRG(\tilde{\Phi})^{-1}$ do not touch for all $\tau \in [0,1]$. Finally, we compute the SRG bound in~\eqref{eq:lfr_srg_bound} which yields $\rmin(\mathcal{G}_R) = 12.09$, see Fig.~\ref{fig:example_2_all}. From Theorem~\ref{thm:srg-lfr-system}, we can conclude that the system is causal and well-posed with incremental $L_2$-gain bound $\Gamma(R) \leq 12.09$. The large gain bound can be understood by the fact that the spring between $m_1$ and $m_2$ is negative, i.e., active, for small deflections. Therefore, for small inputs $u_1,u_2$, relatively large outputs $x_1,x_2$ can be expected. Also, we did not yet optimize the bound $\hat{\Gamma}$ over all possible loop transformations.

Note that in this example, we have used a loop transformation that consists of both a shift in elements of $\Phi$, and a scaling between the output of $\Phi$ and input of $G(s)$.

\subsection{Comparison with IQC Based SRG Results}\label{sec:example_3_iqc}

\begin{figure}[t]
    \centering
    \includegraphics[width=0.6\linewidth]{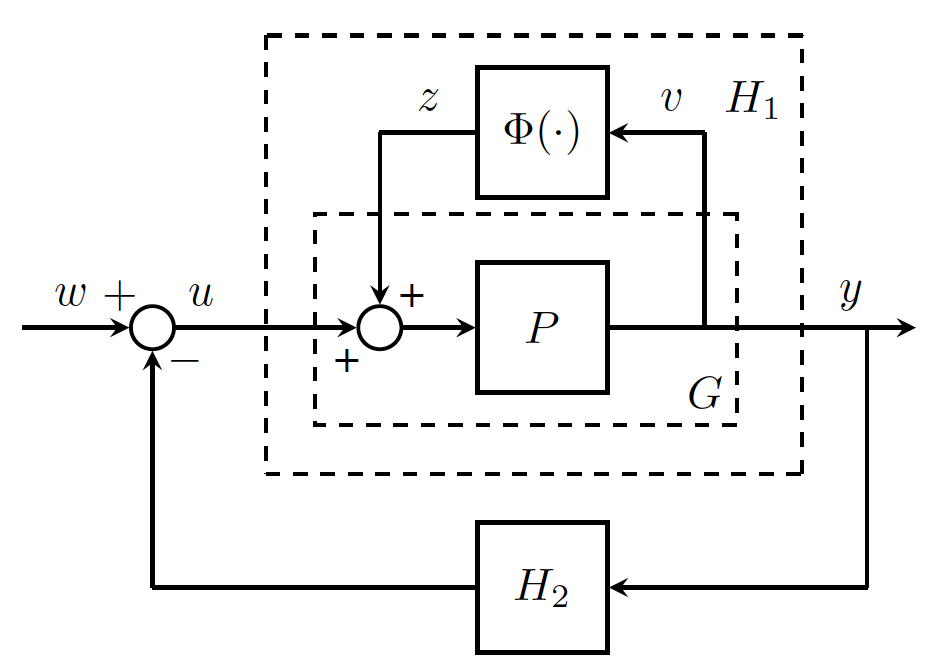}
    \caption{Feedback diagram of the example in Section~\ref{sec:example_3_iqc} (image taken from~\cite{grootExploitingStructureMIMO2025}).}
    \label{fig:example_iqc_diagram}
\end{figure}

We will now treat an example where non-incremental stability is analyzed. Consider the feedback diagram in Fig.~\ref{fig:example_iqc_diagram} where 
\begin{align*}
    P(s) &= \mat{\frac{0.1}{s+1} & \frac{1}{s^3+5s^2+2s+1} \\ \frac{0.1}{s^3+5s^2+2s+1} & \frac{0.2}{s+5}}, \\ H_2(s) &= \mat{\frac{1.7}{s^2+2s+1} & 0 \\ 0 & \frac{1.7}{s^2+3s+3}}, \quad \SG_0(\Phi) \subseteq D_{\sqrt{0.1}}(0),
\end{align*}
i.e., the nonlinearity $\Phi$ is any operator that satisfies $\gamma(\Phi) \leq \sqrt{0.1}$. This system is studied in~\cite{grootExploitingStructureMIMO2025}, where a bound for the SG of $H_1:= [P, -\Phi]$ is obtained using an IQC-based method. 

The SRG of $P$, obtained using Theorem~\ref{thm:LTI_SRG_bound}, is plotted in Fig.~\ref{fig:example_3_H1_P}. To analyze $H_1=[P, -\Phi]$ using Theorem~\ref{thm:non-incremental-mimo-srg}, we plot $SRG(P)^{-1}$ and $\SG_0(\Phi)$ in Fig.~\ref{fig:example_3_H1_denom}. Since $\SG_0(\Phi)$ is a disk and hence satisfies the chord property, we can compute the bound $(\SRG(P)^{-1}-\SG_0(\Phi))^{-1}$ for $\SG_0(H_1)$ using Proposition~\ref{prop:sg_calculus}, which is shown in Fig.~\ref{fig:example_3_H1_closedloop}. We note that the SG bound in Fig.~\ref{fig:example_3_H1_closedloop} is \emph{identical} to the result obtained in~\cite{grootExploitingStructureMIMO2025}, which uses an IQC-based analysis approach. 

\begin{figure}[t]
    \centering
    \begin{subfigure}[b]{0.27\linewidth}
        \centering
        \includegraphics[width=\linewidth]{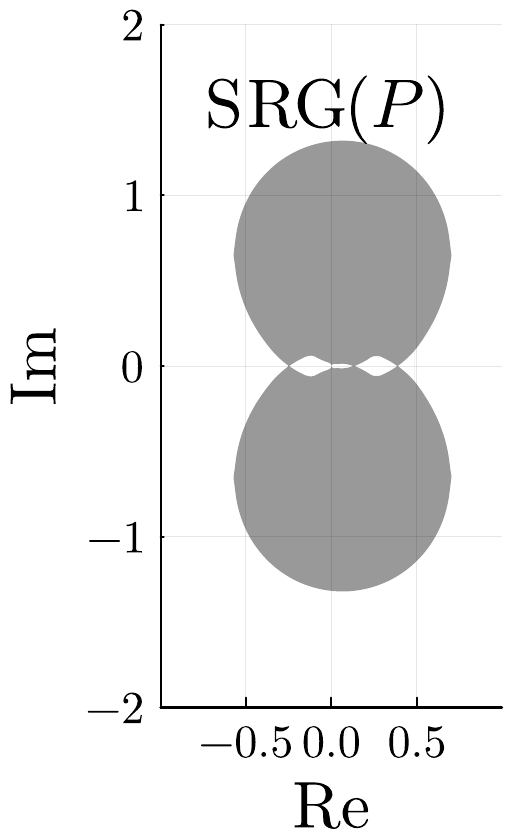}
        \caption{$\SRG(P)$}
        \label{fig:example_3_H1_P}
    \end{subfigure}
    \hfill
    \begin{subfigure}[b]{0.6\linewidth}
        \centering
        \includegraphics[width=\linewidth]{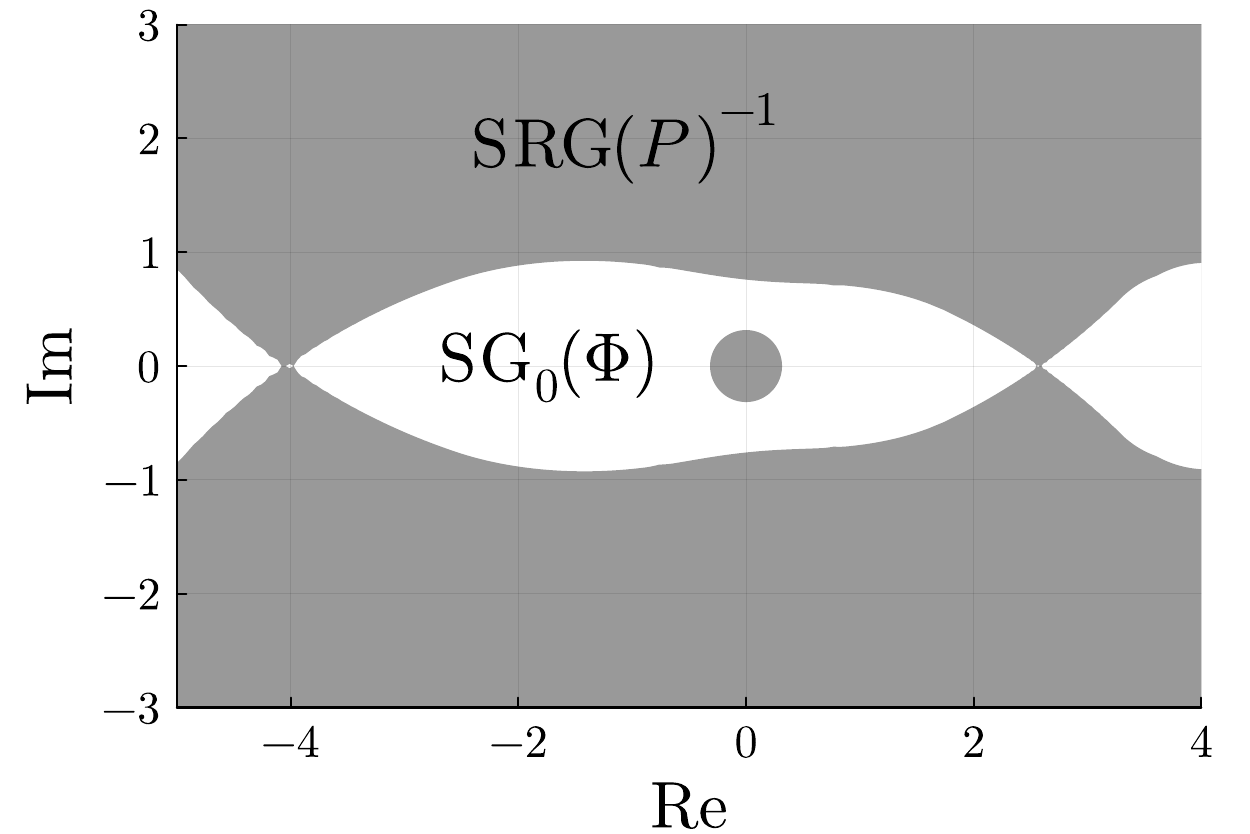}
        \caption{$\SRG(P)^{-1}$ and $\SG_0(\Phi)$}
        \label{fig:example_3_H1_denom}
    \end{subfigure}
    \hfill
    \begin{subfigure}[b]{0.37\linewidth}
        \centering
        \includegraphics[width=\linewidth]{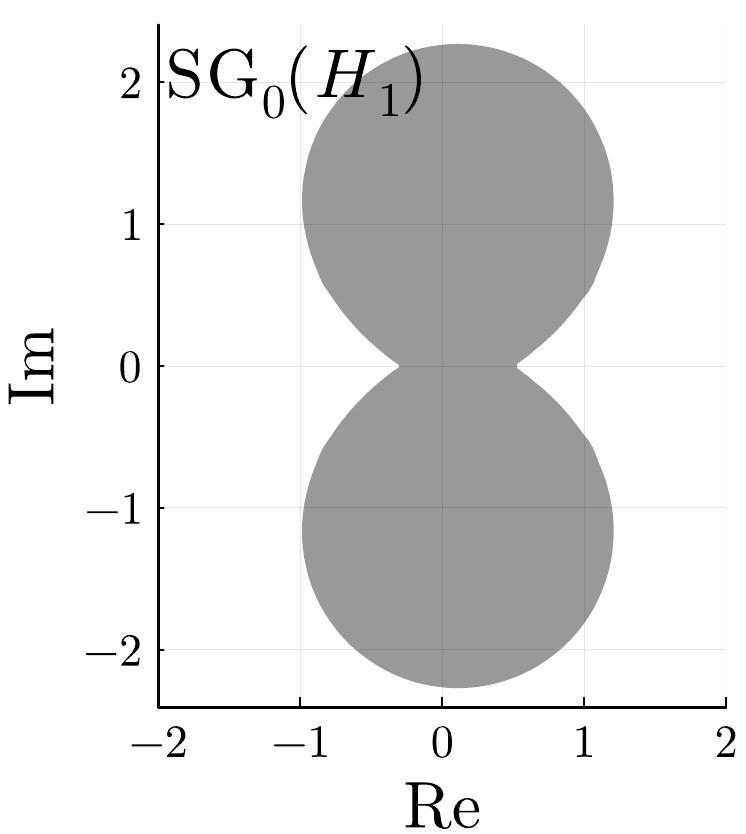}
        \caption{$\SG_0(H_1)$}
        \label{fig:example_3_H1_closedloop}
    \end{subfigure}
    \hfill
    \begin{subfigure}[b]{0.55\linewidth}
        \centering
        \includegraphics[width=\linewidth]{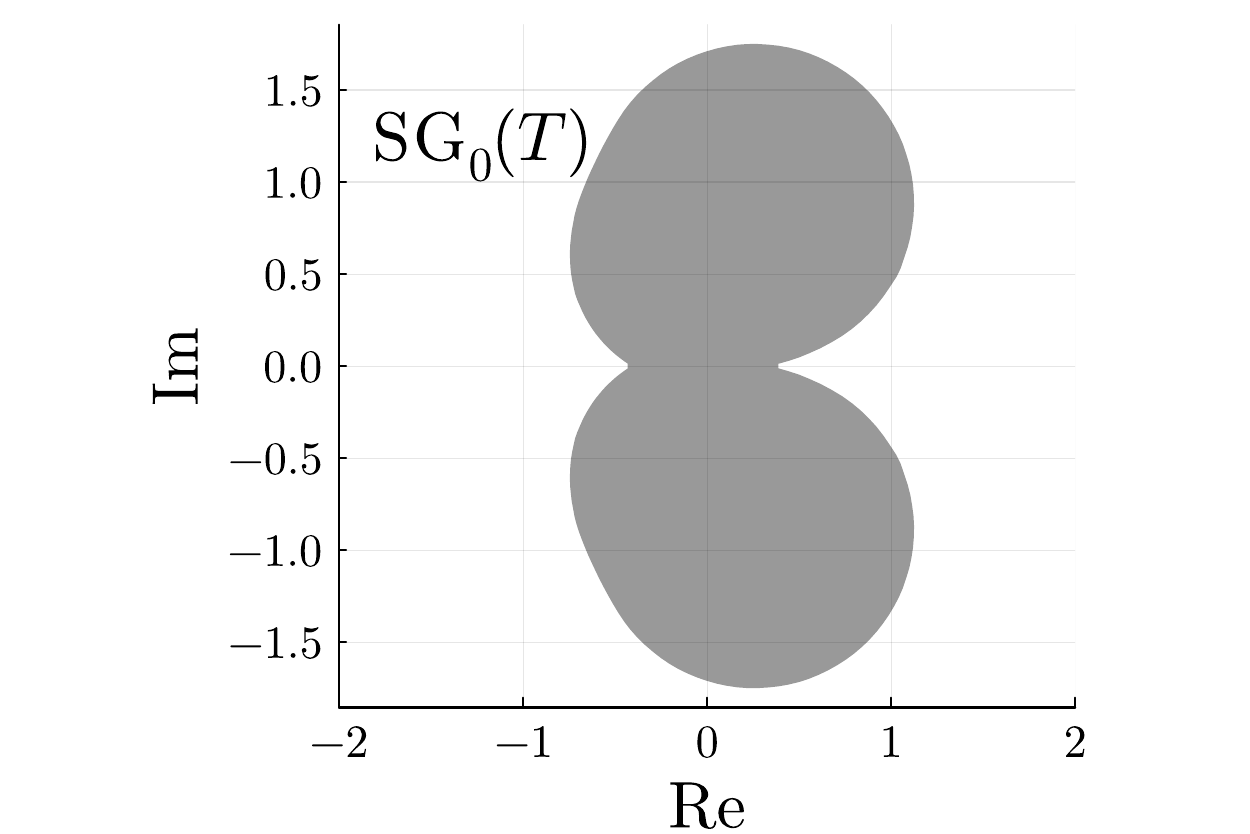}
        \caption{$\SG_0(T)$}
        \label{fig:example_3_T}
    \end{subfigure}
\caption{SG analysis of example in Section~\ref{sec:example_3_iqc}.}
\label{fig:example_3_H1}
\end{figure}

We can also compute an SG bound for $T:=[H_1, H_2]$ by noting that $[H_1, H_2] = [G,-\Phi]$ where $G=[P,H_2]$. The SRG of $G$ is computed with Theorem~\ref{thm:LTI_SRG_bound}, and the stability analysis is done using Theorem~\ref{thm:non-incremental-mimo-srg}, entirely analogous to the analysis of $H_1$. The resulting SG bound for $T$ is shown in Fig.~\ref{fig:example_3_T}, from which we conclude that $\gamma(T) \leq \rmin(\SG_0(T)) \leq 1.79 $. In~\cite{grootExploitingStructureMIMO2025}, an explicit value of the gain bound is not computed, but from their results it can be deduced that they obtain $\gamma(T) \leq \sqrt{16.38} \approx 4.05$ in their approach.   

Since we work non-incrementally, we must \emph{assume} that the systems $H_1$ and $T$ are well-posed.%

} %

\section{Conclusion}\label{sec:conclusion}

In this paper we have developed SRG analysis of nonlinear multivariable feedback systems that are interconnections of non-square MIMO operators. We started by constructing an embedding for MIMO operators into a space of square maps, while restricting the inputs to a relevant subspace when computing the SRG. Especially for tall systems, we demonstrated that this restriction can greatly reduce conservatism in the analysis of feedback systems. Then, focusing on MIMO systems in LFR form, we provided practical stability and (incremental) $L_2$-gain performance results. We derived explicit formulas for the MIMO SRG for the operators in the LFR form, giving a computable toolchain for the SRG of the entire LFR.
\extver{
    Finally, we applied our results to two examples, showing several advantages of our embedding approach over existing squaring up approaches.
}{
    Finally, we applied our results to four examples; a Lur'e system with non-square components which demonstrated several advantages of our embedding approach over existing squaring up approaches, a SISO system with multiple nonlinearities from~\cite{krebbekxScaledRelativeGraph2025}, a nonlinear MIMO mass-spring-damper system, and a MIMO feedback system from~\cite{grootExploitingStructureMIMO2025}. The advantage of our framework is that if the input/output dimensions are compatible, the user can perform computations with the MIMO SRG without carrying out the embedding explicitly. Even though this paper focused on systems in LFR form, core results of our framework can be applied to any interconnection of MIMO systems. Throughout the paper, we mentioned that there are certain degrees of freedom in MIMO SRG calculations; loop transformations that shift, scale and possibly permute the inputs and outputs of MIMO systems, These are not treated in detail here, and are the topic of possible further research.
}

\appendix

\extver{}{

\section{Interconnecting MIMO Operators}\label{sec:mimo_interconnections}

Using the MIMO SRG from Definition~\ref{def:mimo_srg}, we have the tools to define the SRG of each operator in~\eqref{eq:lfr_closed_loop}. The next step is to derive formulas to study \emph{interconnections} of these operators, i.e., operator inverses, sums and compositions. Formulas that bound the SRG of operator inverses, sums and compositions are derived in~\cite{ryuScaledRelativeGraphs2022}. However, little attention is given to the \emph{domain} and \emph{range} of the operators under consideration, which can lead to errors as pointed out in~\cite{krebbekxScaledRelativeGraph2025}. Additionally, the domain of the embedded operator plays a crucial role in the description of the original operator as explained in Section~\ref{sec:mimo_srg_def_subsection}.

Therefore, we will develop the necessary \enquote{calculus} for the SRG and SG, as defined in Section~\ref{sec:srg-definition}, that carefully handles the domain and range of operators. We also discuss the notion of adding chords and arcs to an SRG bound, which are necessary for parallel and series connections, respectively. Then, we show how these general rules for SRGs are used to study interconnections of operators using the MIMO SRG from Definition~\ref{def:mimo_srg}, focusing on systems in LFR form. The main result of this section is that once the input/output dimensions match when interconnecting operators, the user does not have to keep track of the technicalities of the embedding in Section~\ref{sec:mimo_srg_def}, and one can analyze interconnections of MIMO operators using the MIMO SRG from Definition~\ref{def:mimo_srg} and the familiar SRG calculus from~\cite{ryuScaledRelativeGraphs2022}. 

\subsection{Interconnection Rules for the SRG}\label{sec:interconnection_rules_srg}

Before we state and prove our interconnection theorems, we discuss how the domain and range of an operator influence the interconnection rules. The five operations we consider are: 1) multiplication with a nonzero real constant, 2) addition with identity, 3) inversion, 4) parallel interconnection and 5) series interconnection. The mathematical definition and their effect on the domain and range are listed in Table~\ref{tab:operations}.

As Table~\ref{tab:operations} demonstrates, the domain and range of an operator play a nontrivial role during interconnections. For example when connecting two operators $R: \dom(R) \to X, S:\dom(S) \to X$ in parallel, i.e., $R+S$, one may ask the following question: What domain of inputs is contained in $\SRG(R+S)$, and does $\SRG(R)+\SRG(S)$ (analogous to \cite{ryuScaledRelativeGraphs2022}) include all relevant inputs? To address this question, we formulate an SRG interconnection theorem with explicit dependence on the domain and range. 

\begin{table}[t]
    \centering
    \caption{Operations and their effect on the domain and range on relations $R,S$ ($0 \neq \alpha \in \R$). }
    \label{tab:operations}
    \begin{tabular}{l|ll}
       Operation & Domain & Range \\ \hline
       $\alpha R$ / $R \alpha$ & invariant / $(1/a) \dom(R)$ & $\alpha \, \ran(R)$ / invariant \\
       $I+R$ & $\dom(I+R) = \dom(R)$ & $\ran(I+R)\subseteq $ \\
       &  & $\ran(R) + \ran(I)$ \\
       $R^{-1}$ & $\ran(R)$ & $\dom(R)$ \\
       $R+S$ & $\dom(R+S) = $ & $\ran(R+S) \subseteq $ \\
       & $\dom(R) \cap \dom(S)$ & $ \ran(R)+\ran(S)$ \\
       $RS$ & $\dom(RS) \subseteq \dom(S)$ & $\ran(RS) \subseteq \ran(R)$ 
       \end{tabular}
\end{table}

\begin{proof}[Proof of Proposition~\ref{prop:srg_calculus}]
    The proof extends~\cite{ryuScaledRelativeGraphs2022}, where we now take the details of the domain and range into account. We prove each point separately. 

    $a$. As $\mathcal{U}$ is a linear subspace, we have that for any $u \in \mathcal{U} \implies \alpha u \in \mathcal{U}$. By~\eqref{eq:def_srg_angle} we have $\angle(\alpha u, y) = \angle(u, \alpha y) = \angle(u, y)$. We have $\SRG_\mathcal{U}(\alpha R) = \alpha \SRG_\mathcal{U}(R)$ per definition and since $(1/a)u \in \mathcal{U}$, we also have $\SRG_\mathcal{U}(R \alpha ) = \alpha \SRG_\mathcal{U}(R)$.
    
    $b$. From \cite{ryuScaledRelativeGraphs2022} we use $\mathrm{Re} \, z_R(u_1,u_2)= \frac{\inner{Ru_1-Ru_2}{u_1-u_2}}{\norm{u_1-u_2}^2}$ and $\mathrm{Im} \, z_R(u_1,u_2) = \pm \frac{\norm{\pi_{(u_1-u_2)^\perp}(Ru_1-Ru_2)}}{\norm{u_1-u_2}}$, where $\pi_{x^\perp}$ is the projection on the subspace orthogonal to $x$. Since for all $u_1,u_2 \in \mathcal{U}$ it holds that $\inner{(I_\mathcal{U}+R)u_1-(I_\mathcal{U}+R)u_2}{u_1-u_2} = \norm{u_1-u_2}^2 + \inner{Ru_1-Ru_2}{u_1-u_2}$ and $\pi_{(u_1-u_2)^\perp}((I_\mathcal{U}+R)u_1-(I_\mathcal{U}+R)u_2) = \pi_{(u_1-u_2)^\perp}(Ru_1-Ru_2)$, from which it follows that $\SRG_\mathcal{U}(I_\mathcal{U}+R) = 1 + \SRG_\mathcal{U}(R)$. 
    
    $c$. Per definition of the M\"obius inverse $r e^{j \phi} \mapsto (1/r) e^{j \phi}$ one has $(\SRG_\mathcal{U}(R))^{-1} \setminus \{0, \infty \} = \\ \{\frac{\norm{u_1-u_2}}{\norm{y_1-y_2}} e^{\pm j \angle(u_1-u_2,y_1-y_2)} \mid (u_1,y_1), (u_2,y_2) \in R, u_1\neq u_2,y_1\neq y_2, u_1,u_2 \in \mathcal{U} \} = \\ \{\frac{\norm{u_1-u_2}}{\norm{y_1-y_2}} e^{\pm j \angle(u_1-u_2,y_1-y_2)} \mid (y_1,u_1), (y_2,u_2) \in R^{-1}, u_1\neq u_2,y_1\neq y_2, y_1,y_2 \in R(\mathcal{U}) \} = \SRG_{R(\mathcal{U})} \setminus \{0,\infty \}$. By \cite[p. 588]{ryuScaledRelativeGraphs2022}, $0 \in \SRG_\mathcal{U}(R) \iff \infty \in \SRG_{R(\mathcal{U})}(R^{-1})$ and $\infty \in \SRG_\mathcal{U}(R) \iff 0 \in \SRG_{R(\mathcal{U})}(R^{-1})$. 
    
    $d$. The case $\infty \notin \SRG_\mathcal{U}(R) \cup \SRG_\mathcal{U}(S)$ is proven entirely analogously to \cite[Theorem 6]{ryuScaledRelativeGraphs2022} by taking $\dom(R)=\dom(S)=\mathcal{U}$. If $\SRG_\mathcal{U}(R), \SRG_\mathcal{U}(S) \neq \emptyset$, the theorem holds for $\infty \in \SRG_\mathcal{U}(R) \cup \SRG_\mathcal{U}(S)$ by defining $\infty+\infty = \infty$. %
    
    $e$. The case $\infty \notin \SRG_\mathcal{U}(R) \cup \SRG_\mathcal{Y}(T)$ is proven entirely analogously to \cite[Theorem 7]{ryuScaledRelativeGraphs2022} by taking $\dom(R)=\mathcal{U}$ and $R(\mathcal{U})\subseteq \mathcal{Y} = \dom(T)$. If $\SRG_\mathcal{U}(R), \SRG_\mathcal{Y}(T) \notin \{ \emptyset, \{ 0 \}\}$, then the theorem holds for $\infty \in \SRG_\mathcal{U}(R) \cup \SRG_\mathcal{Y}(T)$ by defining $0 \cdot \infty = \infty$ and $\infty \cdot \infty = \infty$. 
\end{proof}

\begin{remark}
    If $\mathcal{U}$ is just a set and not a linear subspace, the only thing that changes is Proposition~\ref{prop:srg_calculus}.\ref{eq:srg_calculus_alpha}, which becomes $\SRG_\mathcal{U}(\alpha R) = \SRG_{(1/a)\mathcal{U}}(R \alpha) = \alpha \SRG_\mathcal{U}(R)$, see Table~\ref{tab:operations}. 
\end{remark}

\subsection{Interconnection Rules for the SG}\label{sec:interconnection_rules_sg}

Upon minor modifications of~\cite{ryuScaledRelativeGraphs2022}, Proposition~\ref{prop:srg_calculus} can be restated for the SG, which is useful for studying non-incremental stability as opposed to incremental stability. The important novel aspect is to keep track of where $u^\star$ in $\SG_{\mathcal{U}, u^\star}(R)$ is mapped to under $R$. 

\begin{proposition}\label{prop:sg_calculus}
    Let $0 \neq \alpha \in \R$, let $R: X \to X$ be an operator and $S,T : X \to X$, be relations on a Hilbert space $X$ and linear subspaces $\mathcal{U}, \mathcal{Y} \subseteq X$ such that $R(\mathcal{U}) \subseteq \mathcal{Y}$ and $u^\star \in X, y^\star = R u^\star$. Then, 
    \begin{enumerate}[label=\alph*.]
        \item\label{eq:sg_calculus_alpha} $\SG_{\mathcal{U},u^\star}(\alpha R) = \alpha \SG_{\mathcal{U},u^\star}(R)$ and if $u^\star=0$ then also $\SG_{\mathcal{U},u^\star}(R \alpha ) = \alpha \SG_{\mathcal{U},u^\star}(R)$,
        \item\label{eq:sg_calculus_plus_one} $\SG_{\mathcal{U},u^\star}(I_\mathcal{U} + R) = 1 + \SRG_\mathcal{U}(R)$, where $I_\mathcal{U}$ obeys $I_\mathcal{U} u=u$ for all $u \in \mathcal{U} \cup \{ u^\star \}$,
        \item\label{eq:sg_calculus_inverse} $\SG_{R(\mathcal{U}), y^\star}(R^{-1}) = (\SG_{\mathcal{U},u^\star}(R))^{-1} =: \SG_{\mathcal{U},u^\star}(R)^{-1}$ (where $0,\infty \in \SG_{\mathcal{U},u^\star}(R)$ are allowed).
        \item\label{eq:sg_calculus_parallel} If at least one of $\SG_{\mathcal{U},u^\star}(R), \SG_{\mathcal{U},u^\star}(S)$ satisfies the chord property, then $\SG_{\mathcal{U},u^\star}(R + S) \subseteq \SG_{\mathcal{U},u^\star}(R) + \SG_{\mathcal{U},u^\star}(S)$.
        \item\label{eq:sg_calculus_series} If at least one of $\SG_{\mathcal{U},u^\star}(R), \SG_{\mathcal{Y},y^\star}(T)$ satisfies an arc property, then $\SG_{\mathcal{U},u^\star}(T R) \subseteq \SG_{\mathcal{Y},y^\star}(T) \SG_{\mathcal{U},u^\star}(R)$.
    \end{enumerate}
    See Definitions~\ref{def:chord_property} and \ref{def:arc_property} for the chord and arc property. The SGs above may contain $0, \infty$. If any of the SRGs above are $\emptyset, \{ 0 \}$ or $\{ \infty \}$, extra care is required, see Ref.~\cite{ryuScaledRelativeGraphs2022}. 
\end{proposition}

\begin{proof}%
    The proof largely mimics~\cite{ryuScaledRelativeGraphs2022}, where now one input is held fixed. We will only discuss the differences with the proof of Proposition~\ref{prop:srg_calculus}.

    $a$. Since $u^\star$ is fixed, we cannot use the trick $\frac{\norm{R \alpha u - R \alpha u^\star}}{\norm{u-u^\star}} = \frac{\norm{R u - R u^\star}}{\norm{(1/\alpha)(u-u^\star)}}$ anymore, since $(1/\alpha) u^\star \neq u^\star$, unless $u^\star = 0$ ($\alpha=1$ is trivial). Therefore, the result from Proposition~\ref{prop:srg_calculus} holds without the $R \alpha$ case. 
    
    $b$. Since we require that $\inner{I_\mathcal{U}(u-u^\star)}{u-u^\star} = \norm{u-u^\star}^2$, it must hold additionally that $I_\mathcal{U} u^\star = u^\star$. 
    
    $c$. Identical to Proposition~\ref{prop:srg_calculus}.
    
    $d$. Identical to Proposition~\ref{prop:srg_calculus}.
    
    $e$. Identical to Proposition~\ref{prop:srg_calculus}, but for the composition of operators one must obey $Ru^\star = y^\star$. 
\end{proof}

\begin{remark}
    In Proposition~\ref{prop:sg_calculus}, we assume that $R$ is single-valued. The theorem can be proven also in the case of relations, but one has to define the SG w.r.t. a set of inputs $u^\star \subseteq X$ instead of $u^\star \in X$, which is not explored here.
\end{remark}

The most frequently used case of Proposition~\ref{prop:sg_calculus} is when $u^\star = 0$ and $R(u^\star)=y^\star=0$. When $\mathcal{U} = L_2$, this situation corresponds to computing the non-incremental gain $\gamma(R)$.

\subsection{Interconnecting MIMO Operators}

Until now, the results in this section are general, and pertain to the SRG and SG as defined in Section~\ref{sec:srg-definition}. We will now discuss how these results can be applied to interconnections of MIMO systems, which may be non-square, using the mathematical tools from Section~\ref{sec:mimo_srg_def}. The SG case, using Proposition~\ref{prop:sg_calculus}, is analogous. 

Consider the systems $R : L_2^{p_R} \to L_2^{q_R}, S : L_2^{p_S} \to L_2^{q_S}$ and $T : L_2^{p_T} \to L_2^{q_T}$, for which we will demonstrate how interconnections are studied. The first step is to compute $n = \max_{i \in \{R,S,T\} } \{p_i, q_i\}$. Then, one computes the MIMO SRG for each operator, as defined in Definition~\ref{def:mimo_srg}. For each operation in Proposition~\ref{prop:srg_calculus}, we discuss below how the input and output dimensions influence the SRG analysis of the interconnection. 
\begin{enumerate}[label=\alph*.]
    \item \textbf{Pre/post multiplication with a real gain:} For $R$, one has $\mathcal{U} = L_2^{p_R}$, which is a linear subspace of $\mathcal{U}_{p_R}^n = \iota_{n \leftarrow p_R}(\mathcal{U})$, hence one can apply \ref{prop:srg_calculus}.\ref{eq:srg_calculus_alpha} without any further conditions. 
    \item \textbf{Addition with identity:} This operation is well-defined for systems if $p_R=q_R$. If $p_R > q_R$, then $I_\mathcal{U}$ will have more output dimensions, effectively giving $R$ an extra $p_R-q_R$ identically zero outputs. Conversely, if $p_R < q_R$, then the identity only feeds through the first $p_R$ inputs and outputs zero for the remaining channels.
    \item \textbf{Inversion:} This operation is always well-defined. One must keep in mind that $R(L_2^{p_R}) \subsetneq L_2^{q_R}$ in general. Therefore, if $\SRG(R)^{-1}$ has finite radius, one can only conclude that $R^{-1}$ has finite incremental gain on $\dom(R^{-1}) = \ran(R)$, see Table~\ref{tab:operations}. 
    \item \textbf{Parallel interconnection:} Proposition~\ref{prop:srg_calculus}.\ref{eq:srg_calculus_parallel} assumes a priori that $p_R = p_S$. For a parallel interconnection to be well-defined, the output dimensions should match as well, i.e., $q_R = q_S$, which is not enforced by Proposition~\ref{prop:srg_calculus}. If (w.l.o.g.) $q_R < q_S$, then the system $R$ has effectively gained $q_S-q_R$ identically zero outputs. 
    \item \textbf{Series interconnection:} Proposition~\ref{prop:srg_calculus}.\ref{eq:srg_calculus_series} assumes a priori that $q_R \leq p_T$ by assuming $\iota_{n \leftarrow q_R} R \pi_{p_R \leftarrow n} (\mathcal{U}_{p_R}^n) \subseteq \mathcal{U}_{p_T}^n$. For a series interconnection to be well-defined, the output dimension of $R$ should match the input dimension of $T$, i.e., $q_R = p_T$. If $q_R < p_T$, then $R$ has effectively gained $p_T-q_R$ identically zero output channels. This makes the SRG calculations conservative, since the SRG of $T$ contains more input dimensions than $R$ provides. Conversely, if $q_R > p_T$, the assumption $R(\mathcal{U}) \subseteq \mathcal{Y}$ is violated. 
\end{enumerate}

The above shows that, as long as the input/output dimensions match when interconnecting operators, the user can simply use SRG calculus with the symbol $\SRG(\cdot)$, which represents the MIMO SRG in Definition~\ref{def:mimo_srg}. In other words, once the dimensions are right when interconnecting, the subscripts, which indicate the input spaces in Propositions~\ref{prop:srg_calculus} and \ref{prop:sg_calculus}, can be dropped. 

Note that if the input/output dimensions do not match, the SRG operations in Proposition~\ref{prop:srg_calculus} are still \emph{mathematically} well-defined. However, the interconnection of operators as a system has a different interpretation.

\section{Stability Theorems}\label{sec:stability-thms}

In this section, we develop S(R)G-based stability criteria for the interconnection in Fig.~\ref{fig:small-gain-setup}. First, the incremental theory, using the SRG, is developed in Section~\ref{sec:incremental-stability-thms}. Subsequently we develop analog results for the non-incremental setting in Section~\ref{sec:non-incremental-stability-thms}, using the SG.

\subsection{Incremental Stability Theorems}\label{sec:incremental-stability-thms}

We state and prove the necessary results to rigorously prove Theorem~\ref{thm:incremental-mimo-srg}. These are 1) the small-gain theorem on arbitrary Banach spaces, and 2) a generalization of~\cite[Thm. 2]{chaffeyHomotopyTheoremIncremental2025} to arbitrary Banach spaces.

Let $n \geq \max\{p,q\}$, and embed $H_1$ and $H_2$ into $\mathcal{N}(L_2^n)$ as $\tilde{H}_1 = \iota_{n \leftarrow q} H_1 \pi_{p \leftarrow n}$ and $\tilde{H}_2 = \iota_{n \leftarrow p} H_1 \pi_{q \leftarrow n}$, respectively. Note that these embeddings obey $\tilde{H}_1 : \mathcal{U}_p^n \to \mathcal{U}_q^n$ and $\tilde{H}_2 : \mathcal{U}_q^n \to \mathcal{U}_p^n$, where $\mathcal{U}_p^n, \mathcal{U}_q^n$ are Banach spaces. This motivates the following statement of the incremental small gain theorem~\cite{desoerFeedbackSystemsInputoutput1975} on Banach spaces.

\begin{lemma}\label{lemma:banach-incremental-small-gain}
    Let $H_1 : \mathcal{U} \to  \mathcal{Y}$, $H_2: \mathcal{Y} \to \mathcal{U}$, where $\mathcal{U},  \mathcal{Y}$ are Banach spaces. If $\Gamma(H_1) \Gamma(H_2) < 1$, then $I + H_2 H_1 : \mathcal{U} \to \mathcal{U}$ is invertible and $\Gamma([H_1,H_2]) < \infty$. 
\end{lemma}

\begin{proof}%
    Fix $u \in \mathcal{U}$ and define $T_u : \mathcal{U} \to \mathcal{U}$ as $T_u x:=  u - H_2 H_1 x$. Note that for all $x,y \in \mathcal{U}$, one has
    \begin{multline*}
        \norm{T_u x - T_u y} = \norm{H_2 H_1 x - H_2 H_1 y} \\ \leq \Gamma(H_2 H_1) \norm{x-y} \leq L \norm{x-y},
    \end{multline*}
    where $L = \Gamma(H_1) \Gamma(H_2) < 1$. By the Banach fixed point theorem~\cite[Thm. 5.1.1]{kreyszigIntroductoryFunctionalAnalysis1989}, there exists a unique solution $e \in \mathcal{U}$ such that $T_u e = e$. Let $y = H_1 e$, then it is clear that $e = u- H_2 y = T_u e$ holds, i.e., $I + H_2 H_1$ is invertible. 
    
    It remains to prove incremental stability. Take $u_1,u_2 \in \mathcal{U}$, then by $e_1-e_2 = u_1-u_2 + H_2 H_1  e_2 - H_2 H_1  e_1$ we obtain $\norm{e_1-e_2} = \norm{u_1-u_2} + \Gamma(H_1) \Gamma(H_2) \norm{e_1 -e_2}$, hence $\norm{e_1-e_2} \leq \frac{1}{1-\Gamma(H_1) \Gamma(H_2)} \norm{u_1-u_2}$, proving continuity of $e$ in $u$. Continuity of $y$ in $u$ follows from $\norm{y_1 - y_2} \leq \Gamma(H_1) \norm{e_1-e_2}$, which proves stability. 
\end{proof}

The important feature of Lemma~\ref{lemma:banach-incremental-small-gain} is that it allows one to analyze a feedback interconnection in terms of $H_1,H_2$, or their embeddings $\tilde{H}_1,\tilde{H}_2$, on equal footing. Since~\cite[Thm. 2]{chaffeyHomotopyTheoremIncremental2025} assumes square operators with full-domain, it does not apply to Fig.~\ref{fig:small-gain-setup} for non-square operators, i.e., mapping from one Banach space to another. The following lemma provides the necessary generalization.

\begin{lemma}\label{lemma:banach-space-incremental-homotopy}
    Let $H_1 : \mathcal{U} \to  \mathcal{Y}$ and $H_2: \mathcal{Y} \to \mathcal{U}$, where $\mathcal{U},  \mathcal{Y}$ are Banach spaces, such that 
    \begin{itemize}
        \item $\Gamma(H_1) < \infty$ and $\Gamma(H_2) < \infty$,
        \item $\exists \, \hat{\Gamma}>0$ such that $\Gamma([H_1, \tau H_2]) \leq \hat{\Gamma}$, for all $\tau \in [0,1]$.
    \end{itemize}
    Then, $I + \tau H_2 H_1$ is invertible for all $\tau \in [0,1]$. 
\end{lemma}

\begin{proof}%
    Similar to \cite[Theorem 2]{chaffeyHomotopyTheoremIncremental2025}, except that we must use Lemma~\ref{lemma:banach-incremental-small-gain}, as~\cite{chaffeyHomotopyTheoremIncremental2025} is limited to $\mathcal{U} = \mathcal{Y} = L_2$. 
\end{proof}

The calculus of SRGs, as developed in Section~\ref{sec:interconnection_rules_srg}, provides us with a bound $\Gamma([\tilde{H}_1, \tilde{H}_2])$, while $\dom([\tilde{H}_1, \tilde{H}_2]) = \ran(\tilde{H}_1^{-1}+\tilde{H}_2) \subseteq \mathcal{U}_p^n$ is unknown. The practical use of Lemma~\ref{lemma:banach-space-incremental-homotopy} is to establish that $\dom([\tilde{H}_1, \tilde{H}_2]) = \mathcal{U}_p^n$. We are now in shape give the proof of Theorem~\ref{thm:incremental-mimo-srg} for the analysis of incrementally stable systems using SRGs.

\begin{proof}[Proof of Theorem~\ref{thm:incremental-mimo-srg}]
    For $n \geq \max \{p, q\}$, it is understood that $\SRG(H_1):=\SRG_{\mathcal{U}_p^n}(\tilde{H}_1)$ and $\SRG(H_2):=\SRG_{\mathcal{U}_q^n}(\tilde{H}_2)$, where $\tilde{H}_1 = \iota_{n \leftarrow q} H_1 \pi_{p \leftarrow n}$ and $\tilde{H}_2 = \iota_{n \leftarrow p} H_1 \pi_{q \leftarrow n}$. Therefore, $\tilde{H}_1 : \mathcal{U}_p^n \to \mathcal{U}_q^n$ and $\tilde{H}_2 : \mathcal{U}_q^n \to \mathcal{U}_p^n$ and~\eqref{eq:finite_srg_radii} implies $\Gamma(\tilde{H}_1)<\infty$, $\Gamma(\tilde{H}_2) < \infty$. 
    
    From~\eqref{eq:finite_srg_separation} and Proposition~\ref{prop:srg_calculus} we know that $\rmin((\SRG(H_1)^{-1} + \tau \SRG(H_2))^{-1}) \leq 1/r$, and so $\Gamma([\tilde{H}_1,\tau \tilde{H}_2]) \leq 1/r_\tau \leq 1/r$. By Lemma~\ref{lemma:banach-space-incremental-homotopy}, we can conclude that $I+ \tau \tilde{H}_2 \tilde{H_1}$ is invertible on $\mathcal{U}_p^n$ for all $\tau \in [0,1]$ and therefore $\tilde{H_1} (I+ \tilde{H}_2 \tilde{H_1})^{-1} = [\tilde{H}_1, \tilde{H_2}] : \mathcal{U}_p^n \to \mathcal{U}_q^n$ with $\Gamma([\tilde{H}_1, \tilde{H_2}]) \leq 1/r_1$. 

    The final step is to transfer the result to $[H_1,H_2]$. We use the fact that $\pi_{p \leftarrow n} : \mathcal{U}_p^n \to L_2^p$ is an isometric isomorphism with inverse $\iota_{n \leftarrow p}$ (by Lemma~\ref{lemma:injection-isomorphism} and \ref{lemma:projection-isomorphism}). For $u, e \in L_2^p, y \in L_2^q$, define $\tilde{u} = \iota_{n \leftarrow p} u \in \mathcal{U}_p^n, \tilde{e} = \iota_{n \leftarrow p} e \in \mathcal{U}_p^n$ and $\tilde{y} = \iota_{n \leftarrow q} y \in \mathcal{U}_q^n$. Since $\iota$ is a bijection, we have 
    \begin{align*}
        u = e + H_2 H_1 e & \iff \tilde{u} = \tilde{e} + \tilde{H}_2 \tilde{H}_1 \tilde{e}, \\ 
        y = H_1 e & \iff \tilde{y} = \tilde{H}_1 \tilde{e},
    \end{align*}
    hence $I + H_2 H_1$ is invertible on $L_2^p$ if and only if $I+\tilde{H}_2 \tilde{H}_2$ is invertible on $\mathcal{U}_p^n$. Moreover, since $\pi$ is an isomorphism, one has $\norm{u}=\norm{\tilde{u}}, \norm{e}=\norm{\tilde{e}}$ and $\norm{y}=\norm{\tilde{y}}$, resulting in 
    \begin{equation*}
        \Gamma([H_1,H_2]) = \Gamma([\tilde{H}_1, \tilde{H}_2]). 
    \end{equation*}

    Note that $H_1$ and $H_2$ are causal if and only if $\tilde{H}_1$ and $\tilde{H}_2$ are causal. Causality of $(I+\tilde{H}_2 \tilde{H}_1)^{-1}$ (and hence $(I + H_2 H_1)^{-1}$) follows from the Banach fixed point theorem on the time axis $[0,T]$, see~\cite{desoerFeedbackSystemsInputoutput1975}. Finally, well-posedness of $[H_1, H_2] = H_1 (I+H_2 H_1)^{-1}$ follows from the causality of $H_1$ and causal invertibility of $I+H_2 H_1$.
\end{proof}

\subsection{Non-Incremental Stability Theorems}\label{sec:non-incremental-stability-thms}

The results of~\cite{chaffeyHomotopyTheoremIncremental2025} and~\cite[Thm. 4]{vandeneijndenScaledGraphsReset2024} are restricted to square operators with full-domain SRGs and are therefore inapplicable to the restricted-domain MIMO SRGs (Definition~\ref{def:mimo_srg}) for analyzing Fig.~\ref{fig:small-gain-setup}. Moreover, the non-incremental theory is not formally developed in~\cite{chaffeyHomotopyTheoremIncremental2025} (even though it is mentioned). To address these gaps, we state the non-incremental small-gain theorem and derive the corresponding restricted-domain generalizations of~\cite[Thm.~2]{chaffeyHomotopyTheoremIncremental2025} and~\cite[Thm. 4]{vandeneijndenScaledGraphsReset2024}, yielding Lemma~\ref{lemma:non-incremental-homotopy-theorem} and Theorem~\ref{thm:non-incremental-mimo-srg}. These results are required for a rigorous treatment of restricted input spaces and follow via direct generalizations of the arguments in~\cite{chaffeyHomotopyTheoremIncremental2025}.

In the incremental setting of Section~\ref{sec:incremental-stability-thms}, causality of $(I+H_2 H_1)^{-1}$ could be obtained causality of $H_1,H_2$ alone, which is not possible in the non-incremental setting~\cite{willemsAnalysisFeedbackSystems1971}. Therefore, we assume throughout that $H_1, H_2$ are causal and allow inputs in extended signal spaces. Also, we must assume well-posedness of $[H_1, H_2]$ if we want to guarantee a causal feedback system.

\begin{lemma}\label{lemma:non-incremental-small-gain-theorem}
    Consider the causal systems $H_1 : L_2^p \to L_2^q$ and $H_2 : L_2^q \to L_2^p$. If $I + H_2 H_1$ has a causal inverse on $\Lte^p$ and $\gamma(H_1) \gamma(H_2) < 1$, then $[H_1, H_2] : L_2^p \to L_2^q$ is well-posed with finite non-incremental gain.
\end{lemma}

\begin{proof}
    See~\cite[Theorem III.2.1]{desoerFeedbackSystemsInputoutput1975}. 
\end{proof}

\begin{lemma}\label{lemma:non-incremental-homotopy-theorem}
    Consider the causal systems $H_1 : L_2^p \to L_2^q$ and $H_2 : L_2^q \to L_2^p$. If $\gamma(H_1) < \infty$ and $\gamma(H_2) < \infty$ and
    \begin{itemize}
        \item $I+ \tau H_2 H_1$ has a causal inverse on $\Lte^p$ for all $\tau \in [0,1]$,
        \item $\exists \hat{\gamma}>0$ such that $\gamma([H_1, \tau H_2]) \leq \hat{\gamma}$ for all $\tau \in [0,1]$,
    \end{itemize}
    then $[H_1, H_2] : L_2^p \to L_2^q$ is well-posed with non-incremental gain bound $\gamma([H_1, H_2]) \leq \hat{\gamma}$. 
\end{lemma}

\begin{proof}%
    Let $\nu \in [0,1/(\gamma(H_1)\gamma(H_2))$ and write $T_\nu = [H_1, \nu H_2]$. By the causal invertibility assumption one has $T_\nu : L_2^p \to \Lte^q$. By Lemma~\ref{lemma:non-incremental-small-gain-theorem}, one has $T_\nu : L_2^p \to L_2^q$ with $\gamma(T_\nu) \leq \hat{\gamma}$. For all $\tau \in [0,1/(\hat{\gamma} \gamma(H_2))$, one again applies the small gain theorem Lemma~\ref{lemma:non-incremental-small-gain-theorem} to conclude that $T_{\nu+\tau} : L_2 \to L_2$ with $\gamma(T_{\nu+\tau}) \leq \hat{\gamma}$. Proceeding inductively $N$ times until $\nu + N \tau=1$, as in the proof of~\cite[Theorem 2]{chaffeyHomotopyTheoremIncremental2025} proves the result.
\end{proof}

Now we can state and prove the non-incremental analog of Theorem~\ref{thm:incremental-mimo-srg}, which generalizes~\cite[Thm. 4]{vandeneijndenScaledGraphsReset2024} to the non-square case.

\begin{theorem}\label{thm:non-incremental-mimo-srg}
    Consider the causal systems $H_1 : L_2^p \to L_2^q$ and $H_2 : L_2^q \to L_2^p$, where at least one of $\SG_0(H_1), \SG_0(H_2)$ satisfies the chord property. If for all $\tau \in [0,1]$, the map $I + H_2 H_1$ has a causal inverse on $\Lte^p$ and
    \begin{subequations}
    \begin{align}
        &\rmin(\SG_0(H_1)) < \infty \text{ and } \rmin(\SG_0(H_2)) < \infty, \label{eq:non-incr-finite_srg_radii}\\
        &\dist(\SG_0(H_1)^{-1},- \tau \SG_0(H_2)) \geq r_\tau \geq r >0, \label{eq:non-incr-finite_srg_separation}
    \end{align}
    \end{subequations}
    then $[H_1, H_2] : L_2^p \to L_2^q$ is well-posed and non-incrementally stable with $\gamma([H_1,H_2]) \leq 1/r_1$.
\end{theorem}

\begin{proof}%
    The proof mimics the proof of Theorem~\ref{thm:incremental-mimo-srg}. Note that we assume $H_1(0)=0, H_2(0)=0$. It is understood that $\SG_0(H_1):=\SG_{\mathcal{U}_p^n, 0}(\tilde{H}_1)$ and $\SG_0(H_2):=\SG_{\mathcal{U}_q^n, 0}(\tilde{H}_2)$, where $\tilde{H}_1 = \iota_{n \leftarrow q} H_1 \pi_{p \leftarrow n}$ and $\tilde{H}_2 = \iota_{n \leftarrow p} H_1 \pi_{q \leftarrow n}$. Therefore, $\tilde{H}_1 : \mathcal{U}_p^n \to \mathcal{U}_q^n$ and $\tilde{H}_2 : \mathcal{U}_q^n \to \mathcal{U}_p^n$ and~\eqref{eq:non-incr-finite_srg_radii} implies $\gamma(\tilde{H}_1)<\infty$, $\gamma(\tilde{H}_2) < \infty$. 
    
    From~\eqref{eq:non-incr-finite_srg_separation} we know that $\rmin((\SG_0(H_1)^{-1} + \tau \SG_0(H_2))^{-1}) \leq 1/r$, hence $\gamma([\tilde{H}_1,\tau \tilde{H}_2]) \leq 1/r_\tau \leq 1/r$. By Lemma~\ref{lemma:non-incremental-homotopy-theorem} and the causal invertibility assumption, we can conclude that $[\tilde{H}_1, \tilde{H_2}] : \mathcal{U}_p^n \to \mathcal{U}_q^n$ is well-posed with $\gamma([\tilde{H}_1, \tilde{H_2}]) \leq 1/r_1$. 

    Since $\mathcal{U}_p^n$ and $L_2^p$ ($\mathcal{U}_q^n$ and $L_2^q$) are isometrically isomorphic via $\pi_{p \leftarrow n}$ and its inverse $\iota_{n \leftarrow p}$ ($\pi_{q \leftarrow n}$ and $\iota_{n \leftarrow q}$), we can conclude that $[H_1, H_2] : L_2^p \to L_2^q$ is well-posed and $ \gamma([H_1,H_2]) = \gamma([\tilde{H}_1, \tilde{H}_2])$, proving the claim.
\end{proof}

} %

\section{Adding Chords and Arcs in an Improved Way}\label{sec:improved_chord_arc_completions}

\extver{}{
For parallel and series interconnections, we require the definitions of chords and arcs~\cite{ryuScaledRelativeGraphs2022}. Denote the line segment between $z_1, z_2 \in \C$ as $[z_1, z_2] := \{ \alpha z_1 + (1-\alpha) z_2 \mid \alpha \in [0, 1] \}$. Let the right-hand arc, denoted by $\operatorname{Arc}^+(z, \bar{z})$, be the circle segment of the circle that is centered at the origin and intersects $z,\bar{z}$, with real part greater than $\mathrm{Re} (z)$. The left-hand arc, denoted by $\operatorname{Arc}^-(z, \bar{z})$, is similarly defined, but with real part smaller than $\mathrm{Re} (z)$. More precisely
\begin{align*}
    \operatorname{Arc}^+(z, \bar{z}) =& \{ r e^{j(1-2 \alpha)\phi} \\ & \, \mid z = r e^{j \phi}, \phi \in (-\pi, \pi ], \alpha \in [0, 1] \}, \\
    \operatorname{Arc}^-(z, \bar{z}) =& -\operatorname{Arc}^+(-z, -\bar{z}).
\end{align*}

\begin{definition}\label{def:chord_property}
    A set $\mathcal{C} \subseteq \C$ is said to satisfy the chord property if for all $z \in \C$, it holds that $[z, \bar{z}] \subseteq \mathcal{C}$. 
\end{definition}

\vspace{1mm}

\begin{definition}\label{def:arc_property}
    A set $\mathcal{C} \subseteq \C$ is said to satisfy the left-arc (right-arc) property if for all $z \in \C$, it holds that $\operatorname{Arc}^-(z, \bar{z}) \subseteq \mathcal{C}$ ($\operatorname{Arc}^+(z, \bar{z}) \subseteq \mathcal{C}$). If $\mathcal{C}$ satisfies the left-arc and/or right-arc property, it is said to satisfy an arc property.
\end{definition}

When using Proposition~\ref{prop:srg_calculus}\extver{}{ (or Proposition~\ref{prop:sg_calculus})} to analyze \mbox{parallel/series} interconnections of operators, one must make sure that at least one of the SRG\extver{}{ (SG)} bounds involved satisfies the chord/arc property, respectively. If this is not the case, one must add chords or arcs to the relevant SRG\extver{}{ (SG)} bound.} In this appendix, we discuss how to efficiently add \emph{chords} and \emph{arcs} to complex sets that bound the SRG\extver{}{ (or SG)}. 

\thmspace
\begin{definition}
    For $\mathcal{C} \subseteq \C$, define the chord, left-arc ($-$) and right-arc ($+$) completions, respectively, as
    \begin{equation*}
        \maketextstyle \mathcal{C}^\mathrm{c} := \bigcup_{z \in \mathcal{C}} [z, \bar{z}], \quad 
        \mathcal{C}^{\mp} := \bigcup_{z \in \mathcal{C}} \operatorname{Arc}^\mp(z, \bar{z}).
    \end{equation*}
    Note that $\mathcal{C} \subseteq \mathcal{C}^{\mathfrak{s}}$, where $\mathfrak{s} \in \{\mathrm{c}, -, +\}$\extver{, where the arcs $\operatorname{Arc}^\mp$ are defined in~\cite{ryuScaledRelativeGraphs2022}.}{.}  
\end{definition}
\thmspace

\noindent We define improved chord/arc completions as follows.

\thmspace
\begin{definition}\label{def:improved_chord_arc_completions}
    For $\mathcal{C}_1, \mathcal{C}_2 \subseteq \C$, the improved chord completion of the sum is defined as 
    \begin{equation}\label{eq:improved_chord_def}
        \overline{\mathcal{C}_1 + \mathcal{C}_2} := (\mathcal{C}_1^\mathrm{c} + \mathcal{C}_2) \cap (\mathcal{C}_1 + \mathcal{C}_2^\mathrm{c}).
    \end{equation}
    For products, the improved arc completion is defined as
    \begin{equation}\label{eq:improved_arc_def}
        \overline{\mathcal{C}_1 \mathcal{C}_2} := (\mathcal{C}_1^+ \mathcal{C}_2) \cap (\mathcal{C}_1 \mathcal{C}_2^+)  \cap (\mathcal{C}_1^- \mathcal{C}_2) \cap (\mathcal{C}_1 \mathcal{C}_2^-).
    \end{equation}
\end{definition}
\thmspace

The following lemma is useful for bounding the SRG of a sum or product of operators, while adding the least amount of chords/arcs as possible.  

\thmspace
\begin{lemma}\label{lemma:improved_chord_arc_completions}
    Let $R,S : X \to Y$, $T:Y \to Z$ be relations on Hilbert spaces $X,Y,Z$ and linear subspaces $\mathcal{U}\subseteq X, \mathcal{Y}\subseteq Y$ such that $R(\mathcal{U}) \subseteq \mathcal{Y}$ and $u^\star \in X, y^\star = R u^\star$, then
    \begin{align*}
        & \SRG_\mathcal{U}(R+S) \subseteq \overline{\SRG_\mathcal{U}(R) + \SRG_\mathcal{U}(S)}, \\
        & \SRG_\mathcal{U}(T R) \subseteq \overline{\SRG_\mathcal{Y}(T) \SRG_\mathcal{U}(R)}\extver{.}{,}
    \end{align*}
    \extver{}{which holds also for $\SRG_\mathcal{U} \to \SG_{\mathcal{U}, u^\star}$, $\SRG_\mathcal{Y} \to \SG_{\mathcal{Y}, y^\star}$.}
\end{lemma}

\begin{proof}%
    \extver{}{First we prove the sum rule. }By Proposition~\ref{prop:srg_calculus}.\ref{eq:srg_calculus_parallel}, it holds that $\SRG_\mathcal{U}(R+S)$ is contained in both sets in the right hand side of~\eqref{eq:improved_chord_def}, and therefore also in their intersection.\extver{}{ Using Proposition~\ref{prop:sg_calculus}.\ref{eq:sg_calculus_parallel}, the same result holds for $\SG_{\mathcal{U}, u^\star}(R+S)$.} \extver{}{The product rule follows analogously. }By Proposition~\ref{prop:srg_calculus}.\ref{eq:srg_calculus_series}, it holds that $\SRG_\mathcal{U}(T R )$ is contained each of the four sets in the right hand side of~\eqref{eq:improved_arc_def}, and therefore also in their intersection.\extver{}{ Using Proposition~\ref{prop:sg_calculus}.\ref{eq:sg_calculus_series}, the same result holds for $\SG_{\mathcal{U}, u^\star}(T R)$.}
\end{proof}

\extver{}{We note that algorithms for performing sums and products with improved chord/arc completions are available at \href{https://github.com/Krebbekx/SrgTools.jl}{\texttt{github.com/Krebbekx/SrgTools.jl}}.}

\extver{}{

\onecolumn

\section{On Relations and Operators and the Proof of~\eqref{eqs:lfr_closed_loop}}\label{app:proof_lfr_relational_equivalence}

In order to prove relational equivalence in~\eqref{eqs:lfr_closed_loop}, we only need to prove $(\Phi^{-1} - G_\mathrm{zw})^{-1} = \Phi (1- G_\mathrm{zw} \Phi)^{-1}$ as relations, which is the ``problematic'' feedback part, and the only difference between~\eqref{eq:lfr_closed_loop_operator} and~\eqref{eq:lfr_closed_loop}. Then,~\eqref{eqs:lfr_closed_loop} follows trivially by pre- and post-multipying $(\Phi^{-1} - G_\mathrm{zw})^{-1} $ and $ \Phi (1- G_\mathrm{zw} \Phi)^{-1}$ with $G_\mathrm{zu}$ and $G_\mathrm{yw}$, respectively, and adding $G_\mathrm{yu}$.

Let $\mathcal{H}$ be a Hilbert space\footnote{For proving~\eqref{eq:relational-equality}, we only need $\Sigma_1 : \mathcal{X} \to \mathcal{Y}$ and $\Sigma_2: \mathcal{Y} \to \mathcal{X}$, where $\mathcal{X}, \mathcal{Y}$ are vector spaces. The Hilbert space structure is required for the SRG later.}. A \emph{relation} is a set-valued map $\Sigma : \mathcal{H} \supseteq \dom(\Sigma) \to 2^\mathcal{H}$, where $\dom(\Sigma) = \{ u \in \mathcal{H} \mid \emptyset \neq \Sigma u \subseteq \mathcal{H} \}$ is the domain. Equivalently, relations are also uniquely characterized by their \emph{graph}, defined as $\Sigma = \{ (u, y) \mid u \in \dom(\Sigma), \, y \in \Sigma u \} \subseteq \mathcal{H} \times \mathcal{H}$. For operators $\Sigma_1, \Sigma_2$ on $\mathcal{H}$, we want to prove (in a relational sense)
\begin{equation}\label{eq:relational-equality}
    \Sigma_1 (I+ \Sigma_2 \Sigma_1)^{-1} = (\Sigma_1^{-1} + \Sigma_2)^{-1}.
\end{equation}
To prove~\eqref{eq:relational-equality}, we will use the relational identities in~\eqref{eq:relations}.
Now, consider
\begin{align}
    (\Sigma_1^{-1} + \Sigma_2)^{-1} 
    &= \{ (u, y) \mid u \in (\Sigma_1^{-1} + \Sigma_2)y \} \nonumber \\
    &= \{ (u_1 + u_2, y) \mid (y, u_1) \in \Sigma_1^{-1} , \, (y, u_2) \in \Sigma_2 \} \nonumber \\
    &= \{ (u_1 + u_2, y) \mid (u_1, y) \in \Sigma_1 , \, (y, u_2) \in \Sigma_2  \} \nonumber \\
    \textcolor{gray}{(\text{If }\Sigma_1, \Sigma_2 \text{ are operators })} &= \{ (u_1 + u_2, y) \mid y = \Sigma_1 u_1 , \, u_2 = \Sigma_2 y \}  \label{eq:relation-srg-expression}
\end{align}
and
\begin{align}
    \Sigma_1 (I + \Sigma_2 \Sigma_1)^{-1} 
    &= \{ (u, y) \mid y \in \Sigma_1 (I + \Sigma_2 \Sigma_1)^{-1} u \} \nonumber \\
    &= \{ (u, y) \mid \exists e \text{ s.t. } (e,y) \in \Sigma_1, \, (u, e) \in (I+\Sigma_2 \Sigma_1)^{-1} \} \nonumber \\
    &= \{ (u, y) \mid \exists e \text{ s.t. } (e,y) \in \Sigma_1, \, (e, u) \in (I+\Sigma_2 \Sigma_1) \} \nonumber \\
    &= \{ (u_1+u_2, y) \mid \exists e \text{ s.t. }  (e,y) \in \Sigma_1, \, \textcolor{gray}{\underbrace{\textcolor{black}{(e, u_1) \in I}}_{\implies e=u_1}}, \, (e, u_2) \in \Sigma_2 \Sigma_1 \} \nonumber \\
    &= \{ (u_1 + u_2, y) \mid (u_1,y) \in \Sigma_1, \, (u_1, u_2) \in \Sigma_2 \Sigma_1 \} \nonumber \\
    \textcolor{gray}{(\text{If }\Sigma_1, \Sigma_2 \text{ are operators })} &=  \{ (u_1 + u_2, y) \mid y = \Sigma_1 u_1, \, u_2 = \Sigma_2 \Sigma_1 u_1 = \Sigma_2 y \} \label{eq:relation-operator-expression}
\end{align}
which proves~\eqref{eq:relational-equality} by comparing~\eqref{eq:relation-srg-expression} and~\eqref{eq:relation-operator-expression}.

When modeling system interconnections as relations, we are often interested in single-valued relations with full domain, i.e., for all $u \in \dom(\Sigma) = \mathcal{H}$, one has $\Sigma u = \{ y \}$ for some $y \in \mathcal{H}$. Single-valued relations are identified uniquely as \emph{operators} via the understanding that $y = \Sigma u$, i.e., the output is an element in $\mathcal{H}$ instead of a subset of $\mathcal{H}$. 

We can use SRG analysis to investigate if~\eqref{eq:relational-equality} is a \emph{single-valued} relation, i.e., an \emph{operator}. If $\Sigma$ is a relation on a Hilbert space $\mathcal{H}$, then by definition of the SRG~\cite{ryuScaledRelativeGraphs2022}
\begin{equation*}
    \Sigma \text{ single-valued } \iff \infty \notin \SRG(\Sigma).
\end{equation*}

Now, consider $\Sigma_1, \Sigma_2 : \mathcal{H} \to \mathcal{H}$, which implies that their SRGs do not contain $\infty$ as they are operators. We have the following equivalent SRG statements
\begin{subequations}
\begin{align}
    & \SRG(\Sigma_1)^{-1} \cap - \SRG(\Sigma_2) = \emptyset \iff \label{eq:srg_sep_condition1} \\ 
    & 0 \notin \SRG(\Sigma_1)^{-1} + \SRG(\Sigma_2) \iff \label{eq:srg_sep_condition2} \\ 
    & 0 \notin 1 + \SRG(\Sigma_2) \SRG(\Sigma_1), \label{eq:srg_sep_condition3}
\end{align}
\end{subequations}
which are easily proven by contradiction: 
\begin{itemize}
    \item \eqref{eq:srg_sep_condition1} $\iff$  \eqref{eq:srg_sep_condition2}: note that $\SRG(\Sigma_1)^{-1} \cap - \SRG(\Sigma_2)  \neq \emptyset \iff $ $ \exists z\in \C: z \in -\SRG(\Sigma_2) \text{ and } z \in \SRG(\Sigma_1)^{-1} \iff 0 \in \SRG(\Sigma_1)^{-1} + \SRG(\Sigma_2)$, which proves the claim by contradiction.
    \item \eqref{eq:srg_sep_condition2} $\iff$ \eqref{eq:srg_sep_condition3}: using the fact that $\infty \notin \SRG(\Sigma_1),\SRG(\Sigma_2)$ we have $0 \in \SRG(\Sigma_1)^{-1} + \SRG(\Sigma_2) \iff \exists 0 \neq z \in \C: -z^{-1} \in \SRG(\Sigma_1) \text{ and } z \in \SRG(\Sigma_2) \iff 0 \in 1+ \SRG(\Sigma_2) \SRG(\Sigma_1)$, proving the claim by contradiction.
\end{itemize}

Note that we have shown an interesting feature of the feedback interconnection $y = \Sigma_1 e, e = u-\Sigma_2 y$: if the SRG separation condition above holds, then this implies that 1) the relation $(u, y)$ in~\eqref{eq:relational-equality} is single-valued and 2) the relation $(u, e)$ defined by $(I + \Sigma_2 \Sigma_1)^{-1}$ is single-valued. 

When working non-incrementally, i.e. considering the SG, then the SRG cannot be used anymore as a tool to check single-valuedness of~\eqref{eq:relational-equality}. Instead, one assumes \emph{well-posedness} for the feedback system~\cite{vandeneijndenScaledGraphsReset2024,megretskiSystemAnalysisIntegral1997}, which includes single-valuedness of~\eqref{eq:relational-equality}.

\twocolumn

} %

\bibliographystyle{plain} 
\bibliography{bibliography}

@article{baron-pradaDecentralizedStabilityConditions2025,
  title = {On {{Decentralized Stability Conditions}} Using {{Scaled Relative Graphs}}},
  author = {{Baron-Prada}, Eder and Anta, Adolfo and D{\"o}rfler, Florian},
  year = 2025,
  journal = {IEEE Control Systems Letters},
  pages = {1--1},
  issn = {2475-1456},
  doi = {10.1109/LCSYS.2025.3577230},
  urldate = {2025-06-16},
  copyright = {https://ieeexplore.ieee.org/Xplorehelp/downloads/license-information/IEEE.html}
}

@article{chaffeyGraphicalNonlinearSystem2023,
  title = {Graphical {{Nonlinear System Analysis}}},
  author = {Chaffey, Thomas and Forni, Fulvio and Sepulchre, Rodolphe},
  year = 2023,
  month = oct,
  journal = {IEEE Transactions on Automatic Control},
  volume = {68},
  number = {10},
  pages = {6067--6081},
  issn = {1558-2523},
  doi = {10.1109/TAC.2023.3234016},
  urldate = {2024-03-31}
}

@article{chaffeyHomotopyTheoremIncremental2025,
  title = {A Homotopy Theorem for Incremental Stability},
  author = {Chaffey, Thomas and Kharitenko, Andrey and Forni, Fulvio and Sepulchre, Rodolphe},
  year = 2025,
  journal = {IEEE Transactions on Automatic Control},
  pages = {1--7},
  issn = {0018-9286, 1558-2523, 2334-3303},
  doi = {10.1109/TAC.2025.3632433},
  urldate = {2026-02-26},
  copyright = {https://ieeexplore.ieee.org/Xplorehelp/downloads/license-information/IEEE.html}
}

@article{chenGraphicalDominanceAnalysis2025,
  title = {Graphical {{Dominance Analysis}} for {{Linear Systems}}: {{A Frequency-Domain Approach}}},
  shorttitle = {Graphical {{Dominance Analysis}} for {{Linear Systems}}},
  author = {Chen, Chao and Chaffey, Thomas and Sepulchre, Rodolphe},
  year = 2025,
  month = apr,
  journal = {arXiv:2504.14394},
  eprint = {2504.14394},
  primaryclass = {eess},
  doi = {10.48550/arXiv.2504.14394},
  urldate = {2025-05-14},
  archiveprefix = {arXiv}
}

@misc{chenSoftHardScaled2025,
  title = {Soft and {{Hard Scaled Relative Graphs}} for {{Nonlinear Feedback Stability}}},
  author = {Chen, Chao},
  year = 2025
}

@book{desoerFeedbackSystemsInputoutput1975,
  title = {Feedback Systems: Input-Output Properties},
  shorttitle = {Feedback Systems},
  author = {Desoer, Charles A. and Vidyasagar, Mathukumalli},
  year = 1975,
  series = {Electrical Science Series},
  publisher = {Acad. Press},
  address = {New York},
  isbn = {978-0-12-212050-3},
  langid = {english}
}

@article{eijndenPhaseScaledGraphs2025,
  title = {On Phase in Scaled Graphs},
  author = {van den Eijnden, Sebastiaan and Chen, Chao and Scheres, Koen and Chaffey, Thomas and Lanzon, Alexander},
  year = 2025,
  month = may,
  journal = {arXiv:2504.21448},
  eprint = {2504.21448},
  primaryclass = {eess},
  doi = {10.48550/arXiv.2504.21448},
  urldate = {2025-05-14},
  archiveprefix = {arXiv}
}

@article{elghaouiImplicitDeepLearning2021,
  title = {Implicit {{Deep Learning}}},
  author = {El Ghaoui, Laurent and Gu, Fangda and Travacca, Bertrand and Askari, Armin and Tsai, Alicia},
  year = 2021,
  month = jan,
  journal = {SIAM Journal on Mathematics of Data Science},
  volume = {3},
  number = {3},
  pages = {930--958},
  publisher = {Society for Industrial \& Applied Mathematics (SIAM)},
  issn = {2577-0187},
  doi = {10.1137/20m1358517},
  urldate = {2025-07-17},
  langid = {english}
}

@article{fradkovPassificationNonsquareLinear2003,
  title = {Passification of {{Non-square Linear Systems}} and {{Feedback Yakubovich-Kalman-Popov Lemma}}},
  author = {Fradkov, Alexander},
  year = 2003,
  month = jan,
  journal = {European Journal of Control},
  volume = {9},
  number = {6},
  pages = {577--586},
  issn = {09473580},
  doi = {10.3166/ejc.9.577-586},
  urldate = {2026-01-21},
  copyright = {https://www.elsevier.com/tdm/userlicense/1.0/},
  langid = {english}
}

@article{grootDissipativityFrameworkConstructing2025,
  title = {A {{Dissipativity Framework}} for {{Constructing Scaled Graphs}}},
  author = {de Groot, Timo and {heemels}, Maurice and van den Eijnden, Sebastiaan},
  year = 2025,
  month = jul,
  journal = {arXiv.2507.08411},
  eprint = {2507.08411},
  primaryclass = {math},
  doi = {10.48550/arXiv.2507.08411},
  urldate = {2025-07-16},
  archiveprefix = {arXiv}
}

@article{grootExploitingStructureMIMO2025,
  title = {Exploiting {{Structure}} in {{MIMO Scaled Graph Analysis}}},
  author = {de Groot, Timo and Oomen, Tom and van den Eijnden, Sebastiaan},
  year = 2025,
  month = apr,
  journal = {arXiv:2504.10135},
  eprint = {2504.10135},
  primaryclass = {eess},
  doi = {10.48550/arXiv.2504.10135},
  urldate = {2025-05-14},
  archiveprefix = {arXiv}
}

@article{kouvaritakisGeometricApproachAnalysis1976,
  title = {Geometric Approach to Analysis and Synthesis of System Zeros {{Part}} 2. {{Non-square}} Systems},
  author = {Kouvaritakis, B. and MacFarlane, A. G. J.},
  year = 1976,
  month = feb,
  journal = {International Journal of Control},
  volume = {23},
  number = {2},
  pages = {167--181},
  issn = {0020-7179, 1366-5820},
  doi = {10.1080/00207177608922150},
  urldate = {2026-01-21},
  langid = {english}
}

@article{krebbekxComputingHardScaled2025,
  title = {Computing the {{Hard Scaled Relative Graph}} of {{LTI Systems}}},
  author = {Krebbekx, J. P. J. and {Baron-Prada}, Eder and Toth, Roland and Das, Amritam},
  year = 2025,
  month = nov,
  journal = {arXiv:},
  doi = {10.48550/arXiv.2511.17297}
}

@article{krebbekxMIMOPaperExtended2026,
  title = {Analysis of {{Non-Square Nonlinear MIMO Systems}} Using {{Scaled Relative Graphs}} - {{Extended Version}}},
  author = {Krebbekx, Julius P. J. and T{\'o}th, Roland and Das, Amritam},
  year = 2026,
  journal = {arXiv.2507.16513},
  eprint = {2507.16513},
  primaryclass = {eess},
  doi = {10.48550/arXiv.2507.16513},
  urldate = {2025-08-19},
  archiveprefix = {arXiv}
}

@inproceedings{krebbekxNonlinearBandwidthBode2025a,
  title = {Nonlinear {{Bandwidth}} and {{Bode Diagrams}} Based on {{Scaled Relative Graphs}}},
  booktitle = {2025 {{IEEE}} 64th {{Conference}} on {{Decision}} and {{Control}} ({{CDC}})},
  author = {Krebbekx, Julius P. J. and T{\'o}th, Roland and Das, Amritam},
  year = 2025,
  month = dec,
  pages = {5557--5564},
  publisher = {IEEE},
  address = {Rio de Janeiro, Brazil},
  doi = {10.1109/CDC57313.2025.11312491},
  urldate = {2026-02-23},
  copyright = {https://doi.org/10.15223/policy-029},
  isbn = {979-8-3315-2627-6}
}

@article{krebbekxResetControllerAnalysis2025,
  title = {Reset {{Controller Analysis}} and {{Design}} for {{Unstable Linear Plants}} Using {{Scaled Relative Graphs}}},
  author = {Krebbekx, Julius P. J. and T{\'o}th, Roland and Das, Amritam},
  year = 2025,
  month = jun,
  journal = {arXiv.2506.13518},
  eprint = {2506.13518},
  primaryclass = {eess},
  doi = {10.48550/arXiv.2506.13518},
  urldate = {2025-06-30},
  archiveprefix = {arXiv}
}

@article{krebbekxScaledRelativeGraph2025,
  title = {Scaled {{Relative Graph Analysis}} of {{General Interconnections}} of {{SISO Nonlinear Systems}}},
  author = {Krebbekx, Julius P. J. and T{\'o}th, Roland and Das, Amritam},
  year = 2025,
  month = jul,
  journal = {arXiv.2507.15564},
  eprint = {2507.15564},
  primaryclass = {eess},
  doi = {10.48550/arXiv.2507.15564},
  urldate = {2025-07-22},
  archiveprefix = {arXiv}
}

@book{kreyszigIntroductoryFunctionalAnalysis1989,
  title = {Introductory Functional Analysis with Applications},
  author = {Kreyszig, Erwin},
  year = 1989,
  series = {Wiley Classics Library},
  publisher = {Wiley},
  address = {New York},
  isbn = {978-0-471-50731-4 978-0-471-50459-7},
  langid = {english}
}

@article{megretskiSystemAnalysisIntegral1997,
  title = {System Analysis via Integral Quadratic Constraints},
  author = {Megretski, A. and Rantzer, A.},
  year = 1997,
  month = jun,
  journal = {IEEE Transactions on Automatic Control},
  volume = {42},
  number = {6},
  pages = {819--830},
  issn = {00189286},
  doi = {10.1109/9.587335},
  urldate = {2024-05-06},
  copyright = {https://ieeexplore.ieee.org/Xplorehelp/downloads/license-information/IEEE.html}
}

@inproceedings{misraComputationalAlgorithmSquaringup1992,
  title = {A Computational Algorithm for Squaring-up. {{I}}. {{Zero}} Input-Output Matrix},
  booktitle = {[1992] {{Proceedings}} of the 31st {{IEEE Conference}} on {{Decision}} and {{Control}}},
  author = {Misra, P.},
  year = 1992,
  pages = {149--150},
  publisher = {IEEE},
  address = {Tucson, AZ, USA},
  doi = {10.1109/CDC.1992.371770},
  urldate = {2026-01-21},
  isbn = {978-0-7803-0872-5}
}

@article{nautaComputableCharacterisationsScaled2025,
  title = {Computable {{Characterisations}} of {{Scaled Relative Graphs}} of {{Closed Operators}}},
  author = {Nauta, Talitha and Pates, Richard},
  year = 2025,
  journal = {arXiv.2511.08420},
  publisher = {arXiv},
  doi = {10.48550/ARXIV.2511.08420},
  urldate = {2025-11-17},
  copyright = {arXiv.org perpetual, non-exclusive license}
}

@article{paduartIdentificationNonlinearSystems2010,
  title = {Identification of Nonlinear Systems Using {{Polynomial Nonlinear State Space}} Models},
  author = {Paduart, Johan and Lauwers, Lieve and Swevers, Jan and Smolders, Kris and Schoukens, Johan and Pintelon, Rik},
  year = 2010,
  month = apr,
  journal = {Automatica},
  volume = {46},
  number = {4},
  pages = {647--656},
  publisher = {Elsevier BV},
  issn = {0005-1098},
  doi = {10.1016/j.automatica.2010.01.001},
  urldate = {2025-07-17},
  copyright = {https://www.elsevier.com/tdm/userlicense/1.0/},
  langid = {english}
}

@article{patesScaledRelativeGraph2021,
  title = {The {{Scaled Relative Graph}} of a {{Linear Operator}}},
  author = {Pates, Richard},
  year = 2021,
  month = aug,
  journal = {arXiv:2106.05650},
  eprint = {2106.05650},
  primaryclass = {math},
  urldate = {2024-04-17},
  archiveprefix = {arXiv},
  langid = {english}
}

@article{quanScaledRelativeGraphs2024,
  title = {Scaled {{Relative Graphs}} for {{Nonmonotone Operators}} with {{Applications}} in {{Circuit Theory}}},
  author = {Quan, Jan and Evens, Brecht and Sepulchre, Rodolphe and Patrinos, Panagiotis},
  year = 2024,
  month = nov,
  journal = {arXiv.2411.17419},
  eprint = {2411.17419},
  primaryclass = {math},
  doi = {10.48550/arXiv.2411.17419},
  urldate = {2025-07-04},
  archiveprefix = {arXiv}
}

@article{ryuScaledRelativeGraphs2022,
  title = {Scaled Relative Graphs: Nonexpansive Operators via {{2D Euclidean}} Geometry},
  shorttitle = {Scaled Relative Graphs},
  author = {Ryu, Ernest K. and Hannah, Robert and Yin, Wotao},
  year = 2022,
  month = jul,
  journal = {Mathematical Programming},
  volume = {194},
  number = {1-2},
  pages = {569--619},
  issn = {0025-5610, 1436-4646},
  doi = {10.1007/s10107-021-01639-w},
  urldate = {2024-04-17},
  langid = {english}
}

@article{schoukensIdentificationBlockorientedNonlinear2017,
  title = {Identification of Block-Oriented Nonlinear Systems Starting from Linear Approximations: {{A}} Survey},
  shorttitle = {Identification of Block-Oriented Nonlinear Systems Starting from Linear Approximations},
  author = {Schoukens, Maarten and Tiels, Koen},
  year = 2017,
  month = nov,
  journal = {Automatica},
  volume = {85},
  pages = {272--292},
  issn = {00051098},
  doi = {10.1016/j.automatica.2017.06.044},
  urldate = {2024-09-05},
  langid = {english}
}

@article{vanbeylenNonlinearLFRBlockOriented2013,
  title = {Nonlinear {{LFR Block-Oriented Model}}: {{Potential Benefits}} and {{Improved}}, {{User-Friendly Identification Method}}},
  shorttitle = {Nonlinear {{LFR Block-Oriented Model}}},
  author = {Vanbeylen, Laurent},
  year = 2013,
  month = dec,
  journal = {IEEE Transactions on Instrumentation and Measurement},
  volume = {62},
  number = {12},
  pages = {3374--3383},
  publisher = {{Institute of Electrical and Electronics Engineers (IEEE)}},
  issn = {0018-9456, 1557-9662},
  doi = {10.1109/tim.2013.2272868},
  urldate = {2025-07-17},
  copyright = {https://ieeexplore.ieee.org/Xplorehelp/downloads/license-information/IEEE.html}
}

@article{vandeneijndenScaledGraphsReset2024,
  title = {Scaled Graphs for Reset Control System Analysis},
  author = {Van Den Eijnden, Sebastiaan and Chaffey, Thomas and Oomen, Tom and (Maurice) Heemels, W.P.M.H.},
  year = 2024,
  month = nov,
  journal = {European Journal of Control},
  volume = {80},
  pages = {101050},
  issn = {09473580},
  doi = {10.1016/j.ejcon.2024.101050},
  urldate = {2025-05-14},
  langid = {english}
}

@book{vanderschaftL2GainPassivityTechniques2017,
  title = {L2-{{Gain}} and {{Passivity Techniques}} in {{Nonlinear Control}}},
  author = {Van Der Schaft, Arjan},
  year = 2017,
  series = {Communications and {{Control Engineering}}},
  publisher = {Springer International Publishing},
  address = {Cham},
  doi = {10.1007/978-3-319-49992-5},
  urldate = {2024-10-31},
  copyright = {http://www.springer.com/tdm},
  isbn = {978-3-319-49991-8 978-3-319-49992-5}
}

@book{willemsAnalysisFeedbackSystems1971,
  title = {The {{Analysis}} of {{Feedback Systems}}},
  author = {Willems, Jan C.},
  year = 1971,
  month = jan,
  publisher = {The MIT Press},
  doi = {10.7551/mitpress/1258.001.0001},
  urldate = {2024-11-26},
  isbn = {978-0-262-25720-6},
  langid = {english}
}

@article{yinStabilityAnalysisUsing2022,
  title = {Stability {{Analysis Using Quadratic Constraints}} for {{Systems With Neural Network Controllers}}},
  author = {Yin, He and Seiler, Peter and Arcak, Murat},
  year = 2022,
  month = apr,
  journal = {IEEE Transactions on Automatic Control},
  volume = {67},
  number = {4},
  pages = {1980--1987},
  publisher = {{Institute of Electrical and Electronics Engineers (IEEE)}},
  issn = {0018-9286, 1558-2523, 2334-3303},
  doi = {10.1109/tac.2021.3069388},
  urldate = {2025-07-17},
  copyright = {https://ieeexplore.ieee.org/Xplorehelp/downloads/license-information/IEEE.html}
}

@book{zhouRobustOptimalControl1996,
  title = {Robust and Optimal Control},
  author = {Zhou, Kemin and Doyle, John C. and Glover, Keith and Doyle, John Comstock},
  year = 1996,
  publisher = {Prentice Hall},
  address = {Upper Saddle River, NJ},
  isbn = {978-0-13-456567-5},
  langid = {english}
}

\scriptsize

\vspace{\baselineskip}

\begin{wrapfigure}{l}{20mm} 
    \includegraphics[width=1in,clip,keepaspectratio]{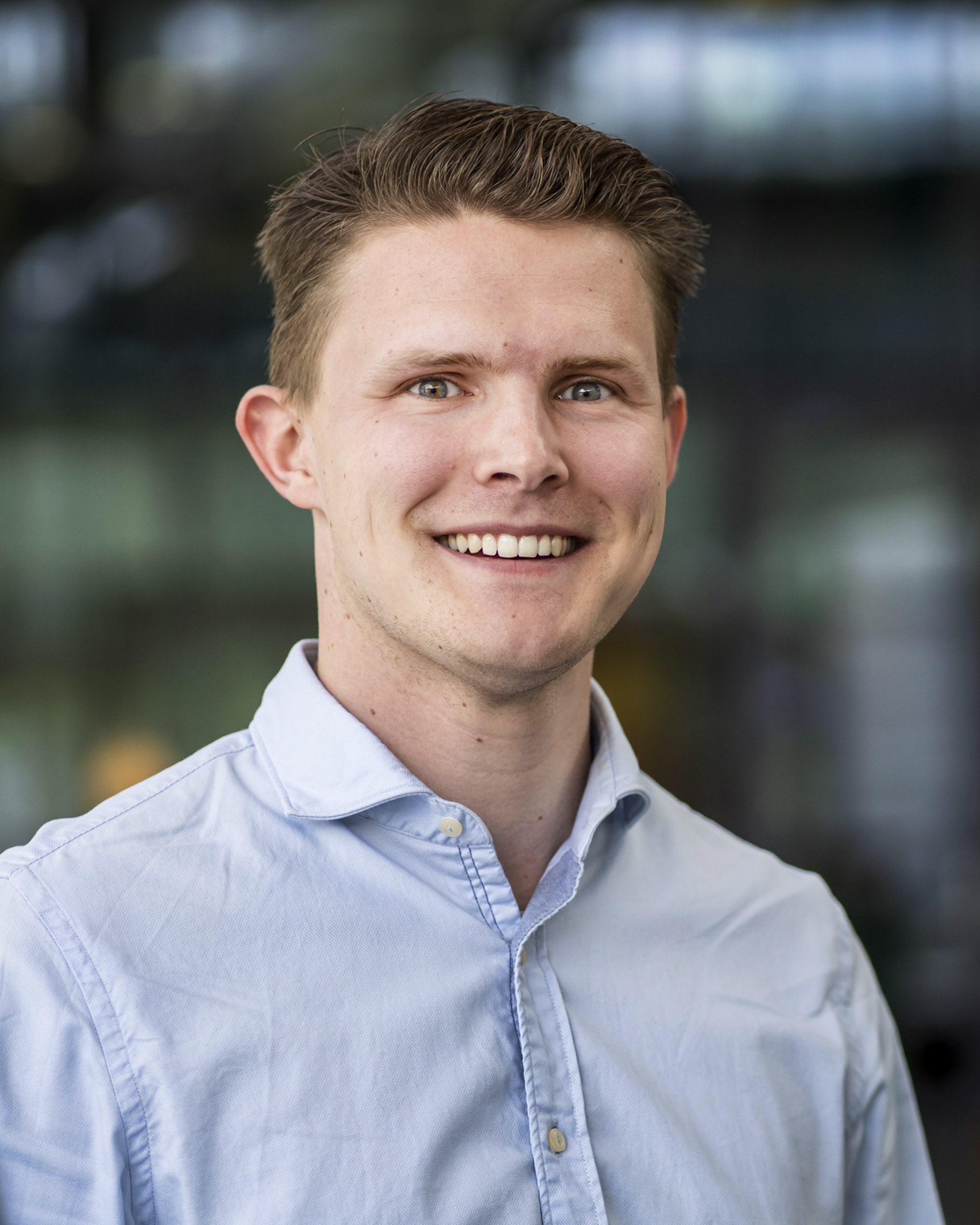}
\end{wrapfigure}
\noindent \textbf{Julius P.~J. Krebbekx} received his B.Sc. degree (Cum Laude) in Applied Physics from Eindhoven University of Technology in 2020 and double M.Sc. degree (both Cum Laude) in Theoretical Physics and Mathematical Sciences from Utrecht University in 2023. He is currently pursuing a Ph.D. degree at the Control Systems Group, Department of Electrical Engineering, Eindhoven University of Technology. His main research interests include the analysis of stability and performance of nonlinear systems and performance shaping for nonlinear systems. He has been awarded an Outstanding Student Paper Award at the 2025 European Control Conference.

\vspace{\baselineskip}

\begin{wrapfigure}{l}{20mm} 
    \includegraphics[width=1in,clip,keepaspectratio]{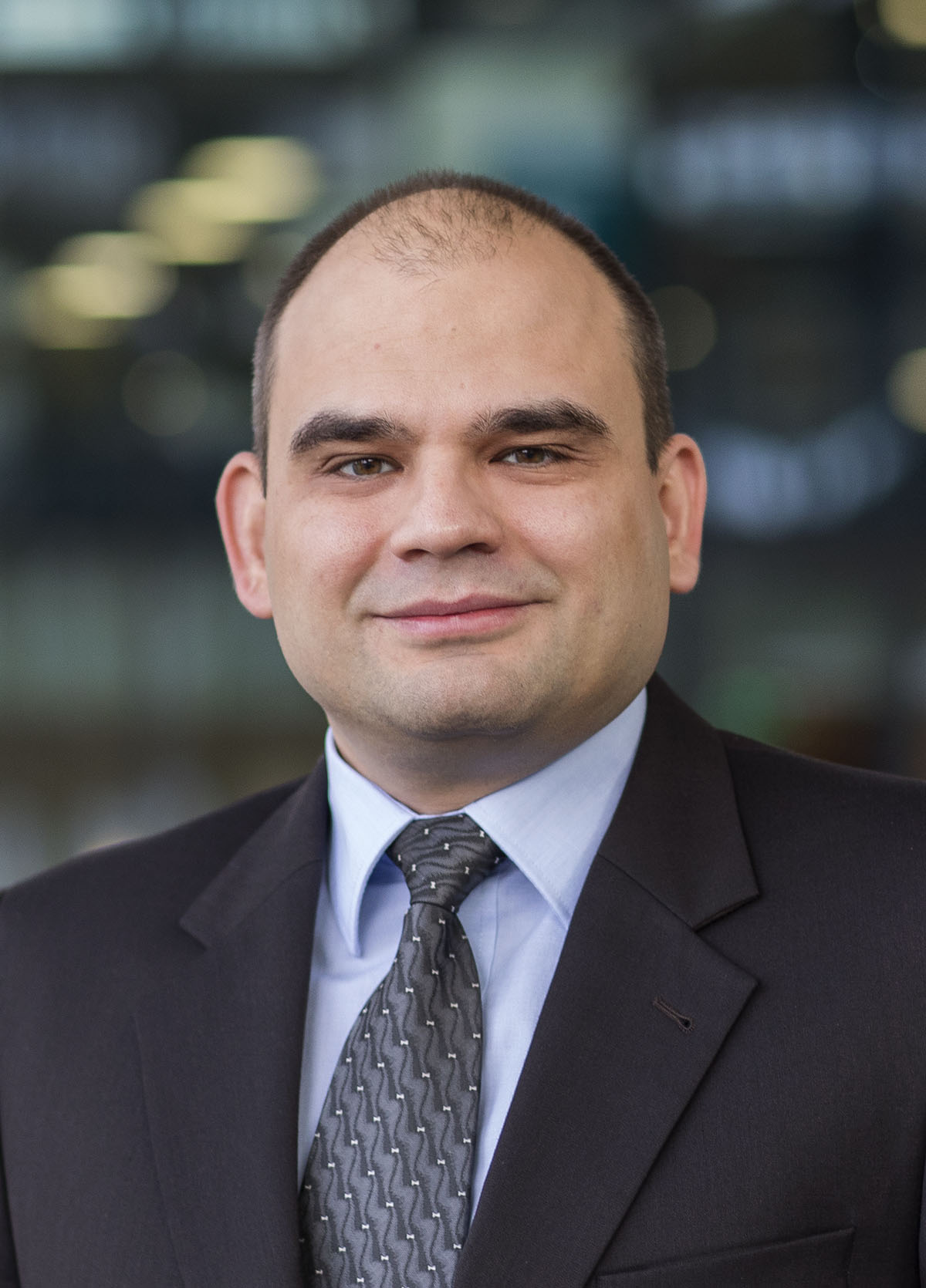}
\end{wrapfigure}
\noindent \textbf{Roland T\'oth} received his Ph.D. degree with Cum Laude distinction at the Delft University of Technology (TUDelft) in 2008.  He was a post-doctoral researcher at TUDelft in 2009 and at the Berkeley Center for Control and Identification, University of California in 2010. He held a position at TUDelft in 2011-12, then he joined to the Control Systems (CS) Group at the Eindhoven University of Technology (TU/e). Currently, he is a Full Professor at the CS Group, TU/e and a Senior Researcher at the Systems and Control Laboratory, HUN-REN Institute for Computer Science and Control (SZTAKI) in Budapest, Hungary. He is Senior Editor of the IEEE Transactions on Control Systems Technology and Assoicate Editor of Automatica. His research interests are in identification and control of linear parameter-varying (LPV) and nonlinear systems, developing data-driven and machine learning methods with performance and stability guarantees for modeling and control, model predictive control and behavioral system theory. On the application side, his research focuses on advancing reliability and performance of precision mechatronics and autonomous robots/vehicles with nonlinear, LPV and learning-based motion control. He has received the TUDelft Young Researcher Fellowship Award in 2010, the VENI award of The Netherlands Organization for Scientific Research in 2011, the Starting Grant of the European Research Council in 2016 and the DCRG Fellowship of Mathworks in 2022.

\vspace{\baselineskip}

\begin{wrapfigure}{l}{20mm} 
    \includegraphics[width=1in,clip,keepaspectratio]{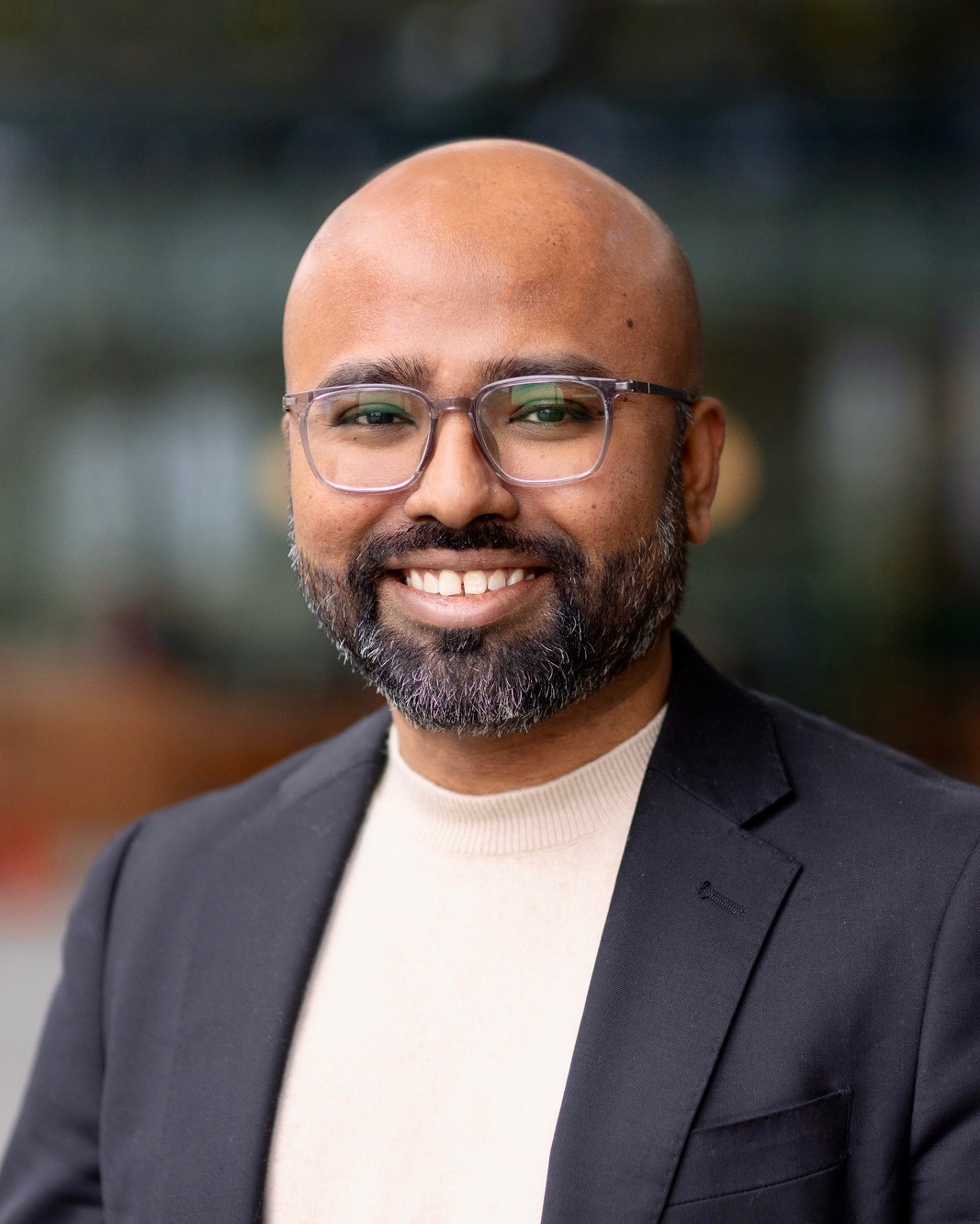}
\end{wrapfigure}
\noindent \textbf{Amritam Das} received MSc. degree (2016) in Systems and Control and Ph.D. degree (2020) in Electrical Engineering from Eindhoven University of Technology. He held the position of research associate at the the University of Cambridge (2020-2021) where he was affiliated with Sideny Sussex College as a college research associate. During 2021-2023, he was a post-doctoral researcher at KTH Royal Institute of Technology. Since February 2023, He is an assistant professor at the Control Systems group of Eindhoven University of Technology. He currently serves as an associate editor of the European Journal of Control and a member of IEEE CSS conference editorial board. He is also a member of IEEE/IFAC Technical Committee on Distributed Parameter Systems. His research interests are nonlinear control, physics-informed learning, control of PDEs, and model reduction.

\end{document}